\documentclass[11pt]{article}

\usepackage{times}
\usepackage{amsmath,amsthm,amsfonts,amssymb}
\usepackage{cancel}
\usepackage{fullpage}
\usepackage{liyang}
\usepackage{framed}
\usepackage{verbatim}
\usepackage{enumitem}
\usepackage{array}
\usepackage{multirow}
\usepackage{afterpage}
\usepackage{caption}
\usepackage{subcaption}
\usepackage{setspace}

\newcommand{\anote}[1]{}
\newcommand{\pnote}[1]{}

\usepackage[normalem]{ulem}

\usepackage{caption} 

\usepackage[usenames,dvipsnames]{xcolor}
\usepackage{fancyhdr}
\usepackage{todonotes}

\makeatletter
\newtheorem*{rep@theorem}{\rep@title}
\newcommand{\newreptheorem}[2]{
\newenvironment{rep#1}[1]{
 \def\rep@title{#2 \ref{##1}}
 \begin{rep@theorem}\itshape}
 {\end{rep@theorem}}}
\makeatother
\theoremstyle{plain}

\makeatletter
\newenvironment{proofof}[1]{\par
  \pushQED{\qed}%
  \normalfont \topsep6\p@\@plus6\p@\relax
  \trivlist
  \item[\hskip\labelsep
\emph{    Proof of #1\@addpunct{.}}]\ignorespaces
}{%
  \popQED\endtrivlist\@endpefalse
}
\makeatother

\def\colorful{1}

\ifnum\colorful=1

\newcommand{\orange}[1]{{\color{orange}{#1}}}

\newcommand{\red}[1]{{\color{red} {#1}}}
\newcommand{\green}[1]{{\color{green} {#1}}}

\newcommand{\gray}[1]{{\color{gray}{#1}}}

\fi
\ifnum\colorful=0

\newcommand{\orange}[1]{{{#1}}}

\newcommand{\red}[1]{{{#1}}}
\newcommand{\green}[1]{{{#1}}}
\newcommand{\gray}[1]{{{#1}}}

\fi

\ifnum\colorful=2

\newcommand{\orange}[1]{{\color{orange}{#1}}}

\newcommand{\red}[1]{{\color{red} {#1}}}
\newcommand{\green}[1]{{\color{green} {#1}}}

\newcommand{\gray}[1]{{\color{gray}{#1}}}

\fi

\usepackage{boxedminipage}

\newreptheorem{theorem}{Theorem}
\newtheorem*{theorem*}{Theorem}
\newreptheorem{lemma}{Lemma}
\newreptheorem{proposition}{Proposition}
\newtheorem*{noclaim*}{Claim}

\newcommand{\Bin}{\mathrm{Bin}}

\providecommand{\keywords}[1]{\textbf{\textit{Keywords--}} #1}

\newcommand{\twohigh}{{2,\mathrm{high}}}
\newcommand{\threehigh}{{3,\mathrm{high}}}

\newcommand{\supportset}{{\cal A}}

\newcommand{\BIG}{R}

\newcommand{\ul}[1]{\underline{#1}}
\newcommand{\ol}[1]{\overline{#1}}

\newcommand{\MIX}{\mathrm{MIX}}
\newcommand{\round}[1]{{\left\lfloor #1 \right\rceil}}

\begin{document}

\title{Learning Sums of Independent Random Variables \\
       with Sparse Collective Support \\
       }
%
%
%
%
%

\author{
Anindya De\thanks{Corresponding author.
Supported by a start-up grant from Northwestern University and NSF CCF-1814706.}\\
University of Pennsylvania \\
{anindyad@cis.upenn.edu}
\and
Philip M. Long\\
Google \\
{plong@google.com} 
\and
\and Rocco A.~Servedio\thanks{Supported by NSF grants CCF-1420349 and CCF-1563155.}\\ 
Columbia University \\{rocco@cs.columbia.edu} 
}

\maketitle

\begin{abstract}
We study the learnability of sums of independent integer random
variables given a bound on the size of the union of their
supports.  For $\supportset \subset \Z_{+}$,
a {\em sum of independent random variables with collective
  support $\supportset$} (called an $\supportset$-sum in this
paper) is a distribution $\bS = \bX_1 + \cdots + \bX_N$ where the
$\bX_i$'s are mutually independent (but not necessarily identically
distributed) integer random variables with $\cup_i \supp(\bX_i) \subseteq
\supportset.$

We give two main algorithmic results for learning such distributions:
\begin{enumerate}
\item 
For the case $| \supportset | = 3$,
we give an algorithm for learning 
$\supportset$-sums
to accuracy $\eps$ that uses $\poly(1/\eps)$ samples and runs in time $\poly(1/\eps)$, independent of $N$ and of the elements of $\supportset$.
\item For an arbitrary constant $k \geq 4$, if
$\supportset = \{ a_1,...,a_k\}$ with $0 \leq a_1 < ... < a_k$, 
we give an algorithm that uses $\poly(1/\eps) \cdot \log \log a_k$ samples (independent of $N$) and runs in time 
$\poly(1/\eps, \log a_k).$  
\end{enumerate}

We prove an essentially matching lower bound: if $|\supportset| = 4$,
then any algorithm 
must use 
\[
\Omega(\log \log a_4) 
\]
samples even for learning to constant accuracy.  We also give similar-in-spirit (but quantitatively very different) algorithmic results, and essentially matching lower bounds, for the case in which $\supportset$ is not known to the learner.

Our algorithms and lower bounds together settle the question of how the sample complexity of learning
sums of independent integer random variables
scales with the elements in the union of their supports,
both in the known-support and unknown-support settings.  Finally, all our algorithms easily extend to the ``semi-agnostic''
learning model, in which training data is generated from a distribution that
is only $c \epsilon$-close to some 
$\supportset$-sum for a constant $c>0$.

\end{abstract}


\thispagestyle{empty}

 

\newpage

\setcounter{page}{1}


\section{Introduction}\label{sec:intro}

The theory of sums of independent random variables forms a rich strand of research in probability. 
Indeed, many of the best-known and most influential results in probability theory are about such sums; prominent examples include the weak and strong law of large numbers, a host of central limit theorems, and (the starting point of) the theory of large deviations. 
Within computer science, the well-known ``Chernoff-Hoeffding'' bounds --- i.e., large deviation bounds for sums of independent random variables --- are a ubiquitous tool of great utility in many contexts.\ignore{Indeed, one of the most used probabilistic tools in computer science is the well-known `Chernoff-Hoeffding' bounds which is a large deviation bound for sums of independent random variables. In fact, it should come as little surprise that there are several books~\cite{P95, gnedenko1954independent} devoted to 
the study of this probabilistic process.}
Not surprisingly, there are several books~\cite{gnedenko1954independent,Petrov75,P95, prokhorov, klesov, borovkov1985advances} devoted to the study of sums of independent random variables.

Given the central importance of sums of independent random variables both within probability theory and for a host of applications,
it is surprising that even very basic questions about \emph{learning} these distributions were not rigorously investigated until very recently. The problem of learning probability distributions from independent samples has attracted a great deal of attention in theoretical computer science for almost two decades  (see~\cite{KMR+:94, Dasgupta:99, AroraKannan:01, VempalaWang:02,KMV:10,MoitraValiant:10,BelkinSinha:10} and a host of more recent papers), but most of this work 
has focused on other types of distributions such as mixtures of Gaussians, hidden Markov models, etc.\ignore{(partly due to the popularity of such models in statistics and machine learning).
} While sums of independent random variables may seem to be a very simple type of distribution, as we shall see below the problem of learning 
such distributions turns out to be surprisingly tricky.

Before proceeding further, let us recall the standard 
PAC-style model for learning distributions that was essentially introduced in \cite{KMR+:94} and that we use in this work. In this model the unknown target distribution $\bX$ is assumed to belong to some class $\mathcal{C}$ of distributions.  A learning algorithm has access to i.i.d.~samples from $\bX$, and must produce an efficiently samplable description of a hypothesis distribution $\bH$ such that with probability at least (say) $9/10$, the total variation distance $\dtv(\bX, \bH)$ between $\bX$ and $\bH$ is at most $\epsilon$. (In the language of statistics, this task is usually referred to as \emph{density estimation}, as opposed to \emph{parametric estimation} in which one seeks to approximately identify the parameters of the unknown distribution $\bX$ when $\mathcal{C}$ is a parametric class like Gaussians or mixtures of Gaussians.)
In fact, all our positive results hold for the more challenging {\em semi-agnostic} variant of this model, which is as above except that the assumption that $\bX \in \mathcal{C}$ is weakened to the requirement $\dtv(\bX, \bX^*) \leq c \epsilon$ for
some constant $c$ and some $\bX^* \in \mathcal{C}$.

\ignore{\subsubsection*{Learning sums of independent random variables (SIIRVs):} Unlike the strand of work surrounding learning of Gaussian mixture models (such  as \cite{Dasgupta:99, AroraKannan:01}), in this paper we will be interested in discrete probability distributions.}

\medskip

\noindent{\bf Learning sums of independent random variables:  Formulating the problem.}  To motivate our choice of learning problem it is useful to recall some relevant context.  Recent years have witnessed many research works in theoretical computer science studying the learnability and testability of discrete probability distributions (see e.g.~\cite{DDS:12kmodallearn, DDS12stoclong, DDOST13, RabaniSS14,AcharyaDK15,AcharyaD15,Canonne15,LiRSS15,CanonneDGR16,Canonne16,DKS16coltksiirv,DKS16stoc,DDKT16}); our paper belongs to this line of research.
\ignore{
Besides the obvious connections to statistics and machine learning, the authors also see this line of investigation as being related to the important but poorly understood topic of complexity theory of distributions~\cite{viola2012complexity}.
}
A folklore result in this area is that a simple brute-force algorithm can learn \emph{any} distribution over an $M$-element set using $\Theta(M/\eps^2)$ samples, and that this is best possible if the distribution may be arbitrary.  Thus it is of particular interest to learn classes of distributions over $M$ elements for which a sample complexity dramatically better than this ``trivial bound'' (ideally scaling as $\log M$, or even independent of $M$ altogether) can be achieved.

This perspective on learning, along with a simple result which we now describe, strongly motivates considering sums of random variables which have small \emph{collective support}.  Consider the following very simple learning problem:  Let $\{\bX_i\}_{i=1}^n$ be independent random variables where $\bX_i$ is promised to be supported on the two-element set $\{0,i\}$ but $\Pr[\bX_i=i]$ is unknown:  what is the sample complexity of learning $\bX = \bX_1 + \cdots + \bX_N$?  
\ignore{Unlike the theory of learning Boolean functions where VC dimension essentially characterizes the sample complexity of learning exactly, no such measure is known for distribution learning.} 
Even though each random variable $\bX_i$ is ``as simple as a non-trivial random variable can be'' --- supported on just two values, one of which is zero --- a straightforward lower bound given in \cite{DDS12stoclong} shows that any algorithm for learning $\bX$ even to constant accuracy must use $\Omega(N)$ samples, which is not much better than the trivial brute-force algorithm based on support size.  (We note that this learning problem is the problem of learning a weighted sum of independent Bernoulli random variables in which the $i$-th Bernoulli random variable has weight equal to $i$, and hence the collective support of $\bX_1,\dots,\bX_N$ is $|\{0,1,\dots,N\}|=N+1.$)

Given this lower bound, it is natural to restrict the learning problem by requiring the random variables $\bX_1,\dots,\bX_N$ to have \emph{small} collective support, i.e.~the union $\supp(\bX_1) \cup \cdots \cup \supp(\bX_N)$ of their support sets is small.  
Inspired by this,  Daskalakis \emph{et al.}~\cite{DDS12stoclong} studied the simplest non-trivial version of this learning problem, in which each $\bX_i$ is a Bernoulli random variable (so the union of all supports is simply $\{0,1\}$; note, though, that the $\bX_i$'s may have distinct and arbitrary biases).
\ignore{

}
The main result of \cite{DDS12stoclong} is that this class (known as \emph{Poisson Binomial Distributions}) can be learned to error $\epsilon$ with $\mathsf{poly}(1/\epsilon)$ samples --- so, perhaps unexpectedly, the complexity of learning this class is completely independent of $N$, the number of summands.
The proof in \cite{DDS12stoclong} relies on several sophisticated results from probability theory, including a discrete central limit theorem from \cite{CGS11} (proved using Stein's method) and a ``moment matching'' result due to Roos~\cite{Roos:00}.  (A subsequent sharpening of the \cite{DDS12stoclong} result in \cite{DKS16coltpbd}, giving improved time and sample complexities, also employed sophisticated tools, namely Fourier analysis and algebraic geometry.)

Motivated by this first success, there has been a surge of recent work which studies the learnability of sums of richer classes of random variables. In particular, Daskalakis et al.~\cite{DDOST13} considered a generalization of \cite{DDS12stoclong} in which each $\bX_i$ is supported on the set $\{0,1, \ldots, k-1\}$, and Daskalakis et al.~\cite{DKT15} considered a vector-valued generalization in which each $\bX_i$ is supported on the set  $\{e_1,\dots,e_k\}$, the standard basis unit vectors in $\R^k.$\ignore{   They gave an algorithm that learns an unknown $k$-PMD$_N$ using $\poly(k/\eps)$ samples and  this result was subsequently sharpened in \cite{DKS16stoc,DDKT16}. While} We will elaborate on these results shortly, but here we first highlight a crucial feature shared by all these results; in all of \cite{DDS12stoclong,DDOST13,DKT15} the collective support of the individual summands forms a ``nice and simple'' set (either $\{0,1\},$ $\{0,1,\dots,k-1\}$, or $\{e_1,\dots,e_k\}$). Indeed,\ignore{while the sample  complexity is independent of $N$ in all these settings,} the technical workhorses of all these results are various central limit theorems which crucially exploit the  simple structure of these collective support sets. (These central limit theorems have since found applications in other settings, such as the design of algorithms for approximating equilibrium~\cite{DDKT16, DKT15, DKS16stoc, Cheng:2017} as well as stochastic optimization~\cite{De:2018}.) 

In this paper we go beyond the setting in which the collective support of $\bX_1,\dots,\bX_N$ is a ``nice'' set, by studying the learnability of $\bX_1 + \cdots + \bX_N$ where the collective support may be an \emph{arbitrary} set of non-negative integers.  
{Two} 
questions immediately suggest themselves:
\begin{enumerate}
\item How (if at all) does the sample complexity depend on the elements in the common support? 
\item Does knowing the common support set help the learning algorithm --- how does the complexity vary depending on whether or not the learning algorithm knows the common support? 
\end{enumerate}

In this paper we give essentially complete answers to these questions.   Intriguingly, the answers to these questions emerge from  the interface of probability theory and number theory:  our algorithms rely on new central limit theorems for sums of independent random variables which we establish, while our matching lower bounds exploit delicate properties of continued fractions and sophisticated equidistribution results from analytic number theory. The authors find it quite surprising that these two disparate sets of techniques ``meet up'' to provide matching upper and lower bounds on sample complexity.

We now formalize the problem that we consider.


\medskip
\noindent {\bf Our learning problem.}
Let $\bX_1,\dots,\bX_N$ be independent (but not necessarily identically distributed) random variables.
Let $\supportset = \cup_i \supp(\bX_i)$ be the union of their supports and assume w.l.o.g. that $\supportset = \{ a_1 ,..., a_k \}$ for
$a_1< a_2 < \cdots < a_k \in \Z_{\geq 0}$.   Let $\bS$ be the sum of these independent random variables, $\bS= \bX_1 + \cdots + \bX_N.$  We refer to such a random variable $\bS$ as an \emph{$\supportset$-sum}.

We study the problem of learning a unknown $\supportset$-sum $\bS$, given access to i.i.d.~draws from $\bS$. $\supportset$-sums generalize several classes of distributions which have recently been intensively studied in unsupervised learning \cite{DDS12stoclong, DDOST13, DKS16coltksiirv}, namely Poisson Binomial Distributions and ``$k$-SIIRVs,'' and are closely related to other such distributions~\cite{DKS16stoc,DDKT16} ($k$-Poisson Multinomial Distributions).  These previously studied classes of distributions have all been shown to have
learning algorithms with sample complexity $\poly(1/\eps)$ for all constant $k$.  

In contrast, in this paper we show that the picture is more varied for the sample complexity of learning 
when $\supportset$ can be any finite set.  Roughly speaking (we will give more details soon), two of our main results are as follows:
\begin{itemize}
\item {\em Any} $\supportset$-sum with $|\supportset| = 3$ is learnable from $\poly(1/\eps)$ samples independent of $N$ and of the elements of $\supportset$.  This is a  significant (and perhaps unexpected) generalization of the efficient learnability of Poisson Binomial Distributions, which corresponds to the
case $|\supportset| = 2$.
\item No such guarantee is possible for $|\supportset| =4$: if $N$ is
  large enough, there are infinitely many sets
  $\supportset = \{ a_1,a_2,a_3,a_4 \}$ 
  with $0 \leq a_1 < ... < a_4$
  such that $\Omega(\log \log a_4)$ examples are
  needed even to learn to constant accuracy (for a small absolute
  constant).
\end{itemize}
\ignore{
}
%
Before presenting our results in more
detail, to provide context we recall relevant previous work on learning related distributions.


\subsection{Previous work}

A \emph{Poisson Binomial Distribution of order $N$}, or PBD$_N$, is a sum of $N$ independent (not necessarily identical) Bernoulli random variables, i.e.  
an $\supportset$-sum for $\supportset = \{0,1\}$.  
Efficient algorithms for learning PBD$_N$ distributions were given in 
\cite{DDS12stoc,DKS16coltpbd}, which gave learning algorithms using $\poly(1/\eps)$ samples and 
$\poly(1/\eps)$ runtime, independent of $N$.  

Generalizing a PBD$_N$ distribution, a
$k$-SIIRV$_N$ \emph{(Sum of Independent Integer Random Variables)} 
is a $\supportset$-sum for $\supportset = \{ 0,...,k-1 \}$.
Daskalakis et al.~\cite{DDOST13} 
(see also \cite{DKS16coltksiirv})
gave $\poly(k,1/\eps)$-time and sample algorithms for learning any $k$-SIIRV$_N$ distribution to accuracy $\eps$, independent of $N$.

Finally, a different generalization of PBDs is provided by the class of \emph{$(N,k)$-Poisson Multinomial Distributions}, or $k$-PMD$_N$ distributions.  Such a distribution is $\bS = \bX_1 + \cdots + \bX_N$ where the $\bX_i$'s are independent (not necessarily identical) $k$-dimensional vector-valued random variables each supported on $\{e_1,\dots,e_k\}$, the standard basis unit vectors in $\R^k.$    
Daskalakis et al.~\cite{DKT15} gave an algorithm that learns 
any unknown $k$-PMD$_N$ using $\poly(k/\eps)$ samples and running in time $\min\{2^{O(k^{O(k))}\cdot \log^{O(k)}(1/\eps)},2^{\poly(k/\eps)}\}$;  this result was subsequently sharpened in \cite{DKS16stoc,DDKT16}.

Any
$\supportset$-sum with $|\supportset| = k$
has an associated underlying $k$-PMD$_N$ distribution:
if $\supportset = \{ a_1,...,a_k\}$, then
writing $\bar{a}$ for the vector $(a_1,\dots,a_k) \in \Z^k$, an $\supportset$-sum $\bS'$ is equivalent to $\bar{a} \cdot \bS$ where $\bS$ is an $k$-PMD$_N$, as making a draw from $\bS'$ is equivalent to making a draw from $\bS$ and outputting its inner product with the vector $\bar{a}.$  However, 
%
this does \emph{not} mean that the \cite{DKT15} learning result for $k$-PMD$_N$ distributions implies a corresponding learning result for $\{ a_1,...,a_k\}$-sums.  If an $\supportset$-sum learning algorithm \emph{were given draws from the underlying $k$-PMD$_N$}, then of course it would be straightforward to run the \cite{DKT15} algorithm, construct a high-accuracy hypothesis distribution $\bH$ over $\R^k$, and output $\bar{a} \cdot \bH$ as the hypothesis distribution for the unknown $\supportset$-sum.  But when learning $\bS'$, the algorithm does {not} receive draws from the underlying $k$-PMD$_N$ $\bS$; instead it only receives draws from $\bar{a} \cdot \bS$.  In fact, as we discuss below, this more limited access causes a crucial \emph{qualitative} difference in learnability, namely an inherent dependence on the $a_i$'s in the necessary sample complexity once $k \geq 4.$  (The challenge to the learner arising from the blending of the
contributions to a $\supportset$-sum is roughly analogous to the challenge that arises in learning a DNF formula;
if each positive example in a DNF learning problem 
were annotated with an identifier for a term that it satisfies, 
learning would be trivial.)

\subsection{The questions we consider and our algorithmic results.}

As detailed above, previous work has extensively studied the learnability of PBDs, $k$-SIIRVs, and $k$-PMDs; however, we believe that the current work is the first to study the learnability of general $\supportset$-sums.  A first simple observation is that since any $\supportset$-sum with $|\supportset|=2$ is a scaled and translated PBD, the results on learning PBDs mentioned above easily imply that the sample complexity of learning any $\{ a_1, a_2 \}$-sum is $\poly(1/\eps)$, independent of the number of summands $N$ and the values $a_1,a_2$.  A second simple observation is that any $\{ a_1,...,a_k \}$-sum 
with $0 \leq a_1 < ... < a_k$
can be learned using $\poly(a_k,1/\eps)$ samples, simply by viewing it as an $a_k$-SIIRV$_N$.  But this bound is in general quite unsatisfying -- indeed, for large $a_k$  it could be even larger than the trivial $O(N^k/\eps^2)$ upper bound that holds since any $\supportset$-sum with $|\supportset|=k$ is supported on a set of size 
$O(N^{k}).$

Once $k \geq 3$ there can be non-trivial additive structure present in the set of values $a_1,\dots,a_k$.  This raises a natural question:   is $k=2$ the only value for which $\supportset$-sums are learnable from a number of samples that is independent of the domain elements $a_1,\dots,a_k$? Perhaps surprisingly, our first main result is an efficient algorithm which gives a negative answer.  We show that for $k=3$, the values of the $a_i$'s don't matter; we do this by giving an efficient learning algorithm (even a semi-agnostic one) for learning $\{ a_1, a_2, a_3 \}$-sums, whose running time and sample complexity are completely independent of $a_1, a_2$ and $a_3$:

\begin{theorem} [Learning $\supportset$-sums with $|\supportset|=3$, known support] \label{known-k-is-three-upper}
There is an algorithm and a positive constant $c$ with the following properties:  The algorithm is given $N$, an accuracy parameter $\eps>0$, distinct values $a_1 < a_2 < a_3 \in \Z_{\geq 0}$, and access to i.i.d. draws from an unknown distribution $\bS^*$ that
has total variation distance at most $c \epsilon$ from
an $\{a_1,a_2,a_3\}$-sum.  The algorithm uses $\poly(1/\eps)$ draws from $\bS^*$, runs in $poly(1/\eps)$ time\footnote{Here and throughout we assume
a unit-cost model for arithmetic operations $+$, $\times$, $\div$.}, and with probability at least $9/10$ outputs a concise representation of a hypothesis distribution $\bH$ such that $\dtv(\bH,\bS^*) \leq \eps.$ 
\end{theorem}

We also give an algorithm for $k \geq 4$.   More
precisely, we show:

\begin{theorem}[Learning $\supportset$-sums, known support]
\label{known-general-k-upper}
For any $k \geq 4$, there is an algorithm and a constant $c>0$ with the following properties:  it is given $N$, an accuracy parameter $\eps > 0$, distinct values $a_1 < \cdots < a_k \in \Z_{\geq 0}$, and access to i.i.d.\ draws from an unknown 
distribution $\bS^*$ that has total variation distance
at most $c \epsilon$ from some
$\{a_1,\dots,a_k\}$-sum. The algorithm runs in time $(1/ \eps) ^{2^{O(k^2)}} \cdot (\log a_k)^{\poly(k)}$\ignore{\rnote{Was ``$2^{\poly(k)} \cdot (1/ \eps) ^{2^{O(k^2)}} \cdot (\log a_k)^{\poly(k)}$'' but looking at it, of course the initial $2^{\poly(k)}$ is superfluous given what follows.}}, uses $(1/ \eps) ^{2^{O(k^2)}} \cdot \log \log a_k$ samples, and with probability at least $9/10$ outputs a concise representation of a hypothesis distribution $\bH$ such that $\dtv(\bH,\bS^*) \leq \eps.$
\end{theorem}

In contrast with $k=3$, our algorithm for general $k\geq 4$ has a sample complexity which depends (albeit doubly logarithmically) on $a_k$. This is a doubly exponential improvement over the naive $\poly(a_k)$ bound which follows from previous $a_k$-SIIRV learning algorithms \cite{DDOST13,DKS16coltksiirv}.

\medskip
\noindent
{\bf Secondary algorithmic results:  Learning with unknown support.}
We also give algorithms for a more challenging \emph{unknown-support} variant of the learning problem. In this variant the values $a_1,\dots,a_k$ are not provided to the learning algorithm, but instead only an upper bound $a_{\max} \geq a_k$ is given.  Interestingly, it turns out that the unknown-support problem is significantly different from the known-support problem: as explained below, in the unknown-support variant the dependence on $a_{\max}$ kicks in at a smaller value of $k$ than in the known-support variant, and this dependence is exponentially more severe than in the known-support variant.  

Using well-known results from hypothesis selection, it is straightforward to show that upper bounds for the known-support case yield upper bounds in the unknown-support case, essentially at the cost of an additional additive $O(k \log \amax)/\eps^2$ 
{term}
in the sample complexity.  This immediately yields the following:

\begin{theorem} [Learning with unknown support of size $k$]  \label{unknown-general-k-upper}
For any $k \geq 3$, there is an algorithm and a positive constant $c$ with the following properties:  The algorithm is given $N$, the value $k$, an accuracy   parameter $\eps>0$, an upper bound $\amax \in \Z_{\geq 0}$, and access to i.i.d. draws from an unknown 
distribution $\bS^*$ that has total variation distance 
at most $c \epsilon$ from an
$\supportset$-sum for
$\supportset = \{a_1,\dots,a_k\} \subset \Z_{\geq 0}$ where $\max_i a_i \leq \amax.$  The algorithm uses $O(k \log \amax)/\eps^2 +(1/ \eps) ^{2^{O(k^2)}} \cdot \log \log \amax$ draws from $\bS^*$, runs in $\poly((\amax)^k) \cdot$ $ (1/ \eps) ^{2^{O(k^2)}} \cdot (\log \amax)^{\poly(k)}$ time, and with probability at least $9/10$ outputs a concise representation of a hypothesis distribution $\bH$ such that $\dtv(\bH,\bS^*) \leq \eps.$
\end{theorem}

Recall that a $\{a_1,a_2\}$-sum is simply a rescaled and translated PBD$_N$ distribution.  Using known results for learning PBDs, it is not hard to show that the $k=2$ case is easy even with unknown support:

\begin{theorem} [Learning with unknown support of size $2$] \label{unknown-k-is-two-upper}
There is an algorithm and a positive constant $c$ with the following properties:  The algorithm is given $N$, an accuracy  parameter $\eps> 0$, an upper bound $\amax \in \Z_{+}$, and access to i.i.d. draws from an unknown 
distribution $\bS^*$ that has total variation distance
at most $c \epsilon$ from an 
$\{a_1,a_2\}$-sum where $0 \leq a_1 < a_2 \leq \amax.$  The algorithm uses $\poly(1/\eps)$ draws from $\bS^*$, runs in $\poly(1/\eps)$ time, and with probability at least $9/10$ outputs a concise representation of a hypothesis distribution $\bH$ such that $\dtv(\bH,\bS^*) \leq \eps.$
\end{theorem}

\subsection{Our lower bounds.}
We establish sample complexity lower bounds for learning $\supportset$-sums that essentially match the above algorithmic results.

\medskip
\noindent {\bf Known support.}  Our first lower bound deals with the
known support setting.  
We give an $\Omega(\log \log a_4)$-sample lower bound for the
problem of learning an $\{ a_1,...,a_4\}$-sum for  $0 \leq a_1 <
a_2 < a_3 <a_4$.
This matches the dependence on $a_k$ of our
 $\poly(1/\eps) \cdot \log \log a_k$ upper bound.
More precisely, we show:

\begin{theorem}[Lower Bound for Learning $\{a_1,...,a_4\}$-sums, known support]
\label{known-k-is-four-lower}
Let $A$ be any algorithm with the following properties:  algorithm $A$ is given $N$, an accuracy parameter $\eps > 0$, distinct values $0 \leq a_1 < a_2 < a_3 < a_4 \in \Z$, 
and access to i.i.d.\ draws from an unknown $\{ a_1,...,a_4\}$-sum $\bS^*$; and with probability at least $9/10$ algorithm $A$ outputs a hypothesis distribution $\tilde{\bS}$ such that $\dtv(\tilde{\bS},\bS^*) \leq \eps.$
Then there are infinitely many quadruples $(a_1,a_2,a_3,a_4)$ such that for sufficiently large $N$, $A$ must use $\Omega(\log \log a_4)$ samples even when run with $\eps$
set to a (suitably small) positive absolute constant.
\end{theorem}

This lower bound holds even though the target is exactly an 
$\{ a_1,...,a_4\}$-sum
(i.e. it holds even in the easier non-agnostic setting).

Since Theorem~\ref{known-k-is-three-upper} gives a $\poly(1/\eps)$ sample and runtime algorithm independent of the size of the $a_i$'s for $k=3$, the lower bound of Theorem~\ref{known-k-is-four-lower} establishes a phase transition between $k=3$ and $k=4$ for the sample complexity of learning $\supportset$-sums:  when $k=3$ the sample complexity is always independent of the actual set $\{a_1, a_2, a_3\}$, but for $k=4$ it can grow as $\Omega(\log \log a_4)$ (but no faster).

\medskip
\noindent {\bf Unknown support.}  Our second lower bound deal{s} with
the unknown support setting.  
We give an
$\Omega(\log \amax)$-sample lower bound for the problem of learning an
$\{a_1,a_2,a_3\}$-sum with unknown support $0 \leq a_1 < a_2 < a_3 \leq
\amax,$ matching the dependence on $\amax$ of our algorithm from Theorem~\ref{unknown-general-k-upper}.
More precisely, we prove:

\begin{theorem} [Lower Bound for Learning $\{ a_1,a_2,a_3\}$-sums, unknown support] \label{unknown-k-is-three-lower}
Let $A$ be any algorithm with the following properties:  algorithm $A$ is given $N$, an accuracy parameter $\eps > 0$, a value $0 < \amax \in \Z$, and access to i.i.d.\ draws from an unknown $\{ a_1,a_2,a_3\}$-sum $\bS^*$ where $0 \leq a_1 < a_2 < a_3 \leq \amax;$ and  $A$ outputs a hypothesis distribution $\tilde{\bS}$ which with probability at least $9/10$ satisfies $\dtv(\tilde{\bS},\bS^*) \leq \eps$.\ignore{(where the expectation is over the random samples drawn from $\bS^*$ and any internal randomness of $A$).} Then for sufficiently large $N$, $A$ must use $\Omega(\log \amax)$ samples even when run with $\eps$ set to a (suitably small) positive absolute constant.
\end{theorem}

Taken together with our algorithm from Theorem~\ref{unknown-k-is-two-upper} for the case $k=2$,
Theorem~\ref{unknown-k-is-three-lower} establishes another phase transition, but now between $k=2$ and $k=3$, for
the sample complexity of learning $\supportset$-sums when $\supportset$ is
unknown.  When $|\supportset|=2$ the sample complexity is always independent of the
actual set, but for $|\supportset|=3$ 
and $0 \leq a_1 < ... < a_3$
it can grow as $\Omega(\log a_3)$ (but no faster).

\medskip

In summary, taken together the algorithms and lower bounds of this paper essentially settle the question of how the sample complexity of learning sums of independent integer random variables with sparse
collective support scales with the elements in the collective support, both in the known-support and unknown-support settings. 

\medskip

\noindent {\bf Discussion.}   As described above, for an arbitrary set $\{a_1, \ldots, a_k \}$, the sample complexity undergoes a significant phase transition between $k=3$ and $k=4$ in the known-support case and between 2 and 3 in the unknown-support case.  In each setting the phase transition is a result of ``number-theoretic phenomena" (we explain this more later) which can only occur for the larger number and cannot occur for the smaller number of support elements.\ignore{Indeed, our lower bound proofs are based on delicate arithmetic phenomena such as equidistribution and properties of continued fractions. }
We
find it somewhat surprising that
the sample complexities of these learning problems are determined by number-theoretic properties of the support sets.\ignore{, and view the presence of number-theoretic phenomena as suggesting that there may not be a characterization of the sample complexity of learning distributions which is as ``clean'' as VC-dimension.}

\medskip
\noindent {\bf Organization.} In the next section we give some of the key ideas that underlie our algorithms.  See 
Section~\ref{s:ltechniques}
for an overview of the ideas behind our lower bounds.  Full proofs are given starting in Section~\ref{sec:prelim}.

%

\section{Techniques for our algorithms}
\label{s:utechniques}

In this section we give an intuitive explanation of some of the ideas
that underlie our algorithms and their analysis.  
While our learning results are for
the semi-agnostic model, for
  simplicity's sake, we focus on the case in which 
  the target distribution $\bS$ is actually 
  an $\supportset$-sum.

A first question, which must be addressed before studying the algorithmic (running time) complexity of learning $\supportset$-sums, is to understand the sample complexity of learning them.\ignore{ as mentioned earlier, for learning probability distributions no simple characterization of sample complexity akin to VC dimension is known.}\ignore{ While in supervised learning of Boolean functions the VC dimension is well known to  tightly characterize the sample complexity of PAC learning~\cite{VapnikChervonenkis:71,KearnsVazirani:94}, for unsupervised learning of probability distributions no such simple characterization is known.} In fact, in a number of recent works on learning various kinds of  ``structured" distributions, just understanding the sample complexity of the learning problem is a major goal that requires significant work~\cite{DDS12stoc, WY12, DDOST13, DDS14, DKT15}. 

In many of the above-mentioned papers, an upper bound on both sample complexity and algorithmic complexity is obtained via a structural characterization of the distributions to be learned; our work follows a similar conceptual paradigm.  To give a sense of the kind of structural characterization that can be helpful for learning, we recall the characterization of SIIRV$_N$ distributions that was obtained in~\cite{DDOST13} (which is the one most closely related to our work). The main result of  \cite{DDOST13} shows that if $\bS$ is any $k$-SIIRV$_N$ distribution, then at least one of the following holds:
\begin{enumerate}
\item $\bS$ is $\epsilon$-close to being supported on $\poly(k/\epsilon)$ many integers;
\item $\bS$ is $\epsilon$-close to a distribution $c \cdot \bZ  + \bY$, where $1 \le c \le k-1$, $\bZ$ is a discretized Gaussian, $\bY$ is a distribution supported on $\{0, \ldots, c-1\}$, and $\bY,\bZ$ are mutually independent.
\end{enumerate}
In other words, \cite{DDOST13} shows that a $k$-SIIRV$_N$ distribution is either close to sparse (supported on $\poly(k/\epsilon)$ integers), or close to a $c$-scaled discretized Gaussian convolved with a sparse component supported on $\{0,\dots,c-1\}$. This leads naturally to an efficient learning algorithm that handles Case (1) above ``by brute-force" and handles Case (2) by learning $\bY$ and $\bZ$ separately (handling $\bY$ ``by brute force'' and handling $\bZ$ by estimating its mean and variance).

In a similar spirit, in this work we seek a more general characterization of $\supportset$-sums.  It turns out, though, that even 
when $|\supportset| = 3$, $\supportset$-sums
can behave in significantly more complicated ways than the $k$-SIIRV$_N$ distributions discussed above.

To be more concrete, let $\bS$ be 
a $\{a_1,a_2,a_3\}$-sum with $0 \leq a_1 < a_2 < a_3$.
By considering a few simple examples it is easy to see that there are at least four distinct possibilities for  ``what $\bS$ is like'' at a coarse level:

\begin{itemize}

\item {\bf Example \#1:} One possibility is that $\bS$ is essentially sparse, with almost all of its probability mass concentrated on a small number of outcomes (we say that such an $\bS$ has ``small essential support'').

\item {\bf Example \#2:}  Another possibility is that $\bS$ ``looks like'' a discretized Gaussian scaled by $|a_i-a_j|$ for some $1 \leq i < j \leq 3$ (this would be the case, for example, if $\bS = \sum_{i=1}^N \bX_i$ where each $\bX_i$ is uniform over $\{a_1,a_2\}$).  

\item {\bf Example \#3:}
A third possibility is that $\bS$ ``looks like'' a discretized Gaussian with no scaling (the analysis of \cite{DDOST13} shows that this is what happens if, for example, $N$ is large and each $\bX_i$ is uniform over $\{a_1=6,a_2=10,a_3=15\}$, since $\gcd(6,10,15)=1$).  

\item {\bf Example \#4:}
Finally, yet another possibility arises if, 
say, 
$a_3$ is very large (say $a_3 \approx N^2$) while $a_2,a_1$ are very small (say $O(1)$), and  $\bX_1,\dots,\bX_{N/2}$ are each uniform over $\{a_1,a_3\}$ while $\bX_{N/2+1},\dots,\bX_N$ are each supported on $\{a_1,a_2\}$ and $\sum_{i=N/2+1}^N \bX_i$ has very small essential support.  In this case,
for large $N$,
$\bS$ would (at a coarse scale) ``look like'' a discretized Gaussian scaled by $a_3 - a_1 \approx N^2$, but zooming in, locally each ``point'' in the support of this discretized Gaussian would actually be a copy of the small-essential-support distribution $\sum_{i=N/2+1}^N \bX_i$.

\end{itemize}

Given these possibilities for how $\bS$ might behave, it should not be surprising that our actual analysis for 
the case $|\supportset| = 3$
(given in Section~\ref{sec:learn3}) 
involves four  cases (and the above four examples land in the four distinct cases).
The overall learning algorithm ``guesses'' which case the target distribution belongs to and runs a different algorithm for each one; the guessing step is ultimately eliminated using the standard tool of hypothesis testing from statistics. We stress that while the algorithms for the various cases differ in some details, there are many common elements across their analyses, and the well known \emph{kernel method} for density estimation provides the key underlying core learning routine that is used in all the different cases.

 In the following intuitive explanation we first consider the case of $\supportset$-sums
for general finite $|\supportset|$,
and later explain how we sharpen the algorithm and analysis in the case $|\supportset|=3$ 
to obtain our stronger results for that case.  Our discussion below highlights a new structural result (roughly speaking, a new limit theorem that exploits both ``long-range'' and ``short-range'' shift-invariance) that plays a crucial role in our algorithms.

\subsection{Learning $\supportset$-sums with $|\supportset| = k$} \label{sec:Aequalsk}

For clarity of exposition in this intuitive overview we make some simplifying assumptions.   First, we make the assumption that the $\supportset$-sum $\bS$ that is to be learned has $0$ as one value in its $k$-element support, i.e. we assume that $\bS= \bX_1 + \ldots + \bX_N$ where the support of each $\bX_i$ is contained in the set $\{0, a_1, \dots, a_{k-1}\}$.  In fact, we additionally assume that each $\bX_i$ is \emph{$0$-moded}, meaning that $\Pr[\bX_i = 0 ] \ge \Pr[\bX_i=a_j]$ for all $i \in [N]$ and all $j \in [k-1]$\ignore{ (and hence $\Pr[\bX_i=0] \geq 1/k$)}.  (Getting rid of this assumption in our actual analysis requires us to work with zero-moded variants of the $\bX_i$ distributions that we denote $\bX'_i$, supported on $O(k^2)$ values that can be positive or negative, but we ignore this for the sake of our intuitive explanation here.)  For $j \in [k-1]$ we define 
\[\gamma_j := \sum_{i=1}^N \Pr[\bX_i = a_j],\]
which can be thought of as the ``weight''  that $\bX_1,\dots,\bX_N$ collectively put on the outcome $a_j$.

\medskip \noindent {\bf A useful tool:  hypothesis testing.}  To explain our approach it is helpful to recall the notion of hypothesis testing in the context of distribution learning~\cite{DL:01}. Informally, given $T$ candidate hypothesis distributions, one of which is $\epsilon$-close to the target distribution $\bS$, a hypothesis testing algorithm uses $O(\epsilon^{-2} \cdot \log T)$ draws from $\bS$, runs in $\poly(T,1/\eps)$ time, and with high probability identifies a candidate distribution which is $O(\epsilon)$-close to $\bS$.  We use this tool  in a few different ways.  Sometimes we will consider algorithms that ``guess" certain parameters from a ``small'' (size-$T$) space of possibilities; hypothesis testing allows us to assume that such algorithms guess the right parameters, at the cost of increasing the sample complexity and running time by only\ignore{$\poly(\log T,1/\eps)$} small factors.\ignore{\pnote{It looks to me that you increase the running time by a factor of $T$, to try out all the guesses, then add on $\poly(\log T,1/\eps)$ at the end, to evaluate them.}}
In other settings we will show via a case analysis that one of several different learning algorithms will succeed; hypothesis testing yields a combined algorithm that learns no matter which case the target distribution falls into.  
(This tool has been  used in many recent works on distribution learning, see e.g. \cite{DDS12stoc, DDS15, DDOST13}.)


\medskip

\noindent {\bf Our analysis.}  Let $t_1 =O_{k,\eps}(1) \ll t_2 = O_{k,\eps}(1) \ll \cdots \ll t_{k-1} = O_{k,\eps}(1)$ be fixed values (the exact values are not important here).  Let us reorder $a_1,\dots,a_{k-1}$ so that the weights $\gamma_1 \leq \cdots \leq \gamma_{k-1}$ are sorted in non-decreasing order.  An easy special case for us (corresponding to Section~\ref{sec:allsmall}) is when each $\gamma_j \leq t_j$.  If this is the case, then $\bS$ has small ``essential support'':  in a draw from $\bS=\bX_1 + \cdots + \bX_N,$ with very high probability for each $j\in [k-1]$ the number of $\bX_i$ that take value $a_j$ is at most $\poly(t_{k-1})$, so w.v.h.p.~a draw from $\bS$ takes one of at most $\poly(t_{k-1})^k$ values.  In such a case it is not difficult to learn $\bS$ using $\poly((t_{k-1})^k,1/\eps)= O_{k,\eps}(1)$ samples (see Fact~24).  We henceforth may assume that some $\gamma_j > t_j.$

\medskip

For ease of understanding it is helpful to first suppose that \emph{every} $j \in [k-1]$ has $\gamma_j > t_j$, and to base our understanding of the general case (that {some} $j \in [k-1]$ has $\gamma_j > t_j$) off of how this case is handled; we note that this special case is the setting for the structural results of Section~\ref{sec:structural}.  (It should be noted, though, that our actual analysis of the main learning algorithm given in Section~\ref{sec:notallsmall} does not distinguish this special case.)  So let us suppose that for all $j \in [k-1]$ we have $\gamma_j > t_j.$  To analyze the target distribution $\bS$ in this case, we consider a multinomial distribution $\bM = \bY_1 + \cdots + \bY_N$ defined by independent vector-valued random variables $\bY_i$, supported on $0,\be_1,\dots,\be_{k-1} \in \Z^{k-1}$, such that for each $i \in [N]$ and $j \in [k-1]$ we have $\Pr[\bY_i=e_j]=\Pr[\bX_i=a_j].$  Note that for the multinomial distribution $\bM$ defined in this way we have $(a_1,\dots,a_{k-1}) \cdot \bM = \bS.$

Using the fact that each $\gamma_j$ is ``large'' (at least $t_{{j}}$), 
recent results from \cite{DDKT16} imply that
the multinomial distribution $\bM$ is close to a $(k-1)$-dimensional discretized Gaussian whose covariance matrix has all eigenvalues
large
(working with zero-moded distributions is crucial to obtain this intermediate result).
In turn, such a discretized multidimensional Gaussian can be shown to be close to a vector-valued random variable in which each marginal (coordinate) is a $(\pm 1)$-weighted sum of \emph{independent} large-variance Poisson Binomial Distributions.
It follows 
that $\bS=(a_1,\dots,a_{k-1}) \cdot \bM$ is close to a 
a weighted sum of $k-1$ signed PBDs.
\footnote{This is a simplification of what the actual analysis establishes, but it gets across the key ideas.} 
A distribution
$\tilde{\bS}$
is a weighted sum of $k-1$ signed PBDs
if
$\tilde{\bS} = a_1 \cdot \tilde{\bS}_1 + \cdots + a_{k-1} \cdot \tilde{\bS}_{k-1}$ where $\tilde{\bS}_1,\dots,\tilde{\bS}_{k-1}$ are \emph{independent} signed PBDs; in turn, a signed PBD is a sum of independent random variables each of which is either supported on $\{0,1\}$ or on $\{0,-1\}$.  The $\tilde{\bS}$ that $\bS$ is close to further has the property that each $\tilde{\bS}_i$ has ``large'' variance (large compared with $1/\eps$).

Given the above analysis, to complete the argument in this case that each $\gamma_j > t_j$ we need a way to learn a 
weighted sum of signed PBDs
$\tilde{\bS} = a_1 \cdot \tilde{\bS}_1 + \cdots + a_{k-1} \cdot \tilde{\bS}_{k-1}$ where each $\tilde{\bS}_{j}$ has large variance.  This is done with the aid of a new limit theorem, Lemma~\ref{l:mix}, that we establish for distributions of this form. We discuss (a simplified version of) this limit theorem in Section~\ref{sec:limit}; here, omitting many details, let us explain what this new limit theorem says in our setting and how it is useful for learning.  Suppose w.l.o.g. that $\Var[a_{k-1}\cdot  \tilde{\bS}_{k-1}]$ contributes at least a ${\frac 1 {k-1}}$ fraction of the total variance of $\tilde{\bS}$.  Let $\MIX$ denote the set of those $j \in \{1,\dots,k-2\}$ such that $\Var[\tilde{\bS}_j]$ is large compared with $a_{k-1}$, and let $\MIX'=\MIX \cup \{k-1\}.$  The new limit theorem implies that the sum $\sum_{j \in \MIX'} a_j \cdot \tilde{\bS}_j$ ``mixes,'' meaning that it is very close (in $\dtv$) to a \emph{single} scaled PBD $a_{\MIX'} \cdot \tilde{\bS}_{\MIX'}$ where $a_{\MIX'} = \gcd \{a_j:  j \in \MIX'\}.$  
(The proof of the limit theorem involves a generalization of the notion of shift-invariance from probability theory~\cite{BX99} and a coupling-based method.  We elaborate on the ideas behind the limit theorem in Section~\ref{sec:limit}.)

Given this structural result, it is enough to be able to learn a distribution of the form
\[
\bT := a_1 \cdot \tilde{\bS}_1 + \cdots + a_\ell \cdot \tilde{\bS}_\ell + a_{\MIX'} \cdot \tilde{\bS}_{\MIX'}
\]
for which we now know that $a_{\MIX'} \cdot \tilde{\bS}_{\MIX'}$ has at least ${\frac 1 {\ell+1}}$ of the total variance, and each $\tilde{\bS}_j$ for $j \in [\ell]$ has $\Var[\tilde{\bS}_j]$ which is ``not too large'' compared with $a_{k-1}$ (but large compared with $1/\eps$).  We show how to learn such a distribution using $O_{k,\eps}(1) \cdot \log \log a_{k-1}$ samples (this is where the $\log \log$ dependence in our overall algorithm comes from).  This is done, intuitively, by guessing various parameters that essentially define $\bT$, specifically the variances $\Var[\tilde{\bS}_1],\dots,\Var[\tilde{\bS}_{\ell}]$.  Since each of these variances is roughly at most $a_{k-1}$ (crucially, the limit theorem allowed us to get rid of the $\tilde{\bS}_j$'s that had larger variance), via multiplicative gridding there are $O_{\eps,k}(1) \cdot \log a_{k-1}$ possible values for each candidate variance, and via our hypothesis testing procedure this leads to an $O_{\eps,k}(1) \cdot \log \log a_{k-1}$ number of samples that are used to learn.

\medskip

We now turn to the general case, that some $j \in [k-1]$ has $\gamma_j > t_j$.  Suppose w.l.o.g.~that $\gamma_1 \leq t_1, \dots \gamma_{\ell-1} \leq t_{\ell-1}$ and $\gamma_{\ell} > t_{\ell}$ (intuitively, think of $\gamma_1,\dots,\gamma_{\ell-1}$ as ``small'' and $\gamma_{\ell},\dots,\gamma_{k-1}$ as ``large'').  Via an analysis (see Lemma~\ref{lem:light-heavy}) akin to the ``Light-Heavy Experiment'' analysis of \cite{DDOST13}, we show that in this case the distribution $\bS$ is close to a distribution $\tilde{\bS}$ with the following structure:  $\tilde{\bS}$ is a mixture of at most $\poly(t_{\ell-1})^{k-1}$ many distributions each of which is a different shift of a \emph{single} distribution, call it $\bS_{\mathrm{heavy}},$ that falls into the special case analyzed above:  all of the relevant parameters $\gamma_{\ell},\dots,\gamma_{k-1}$ are large (at least $t_{\ell}$).  Intuitively, having at most $\poly(t_{\ell-1})^{k-1}$ many components in the mixture corresponds to having $\gamma_1,\dots,\gamma_{\ell-1} < t_{\ell-1}$ and $\ell \leq k-1$, and having each  component be a shift of the same distribution $\bS_{\mathrm{heavy}}$ follows from the fact that there is a ``large gap'' between $\gamma_{\ell-1}$\ignore{(which is at most $t_{\ell-1}$)} and $\gamma_{\ell}$\ignore{ (which is greater than $t_{\ell}$)}.  

Thus in this general case, the learning task essentially boils down to learning a distribution that is (close to) a mixture of translated copies of a distribution of the form $\bT$ given above.  Learning such a mixture of translates is a problem that is well suited to the ``kernel method'' for density estimation.  This method has been well studied in classical density estimation, especially for continuous probability densities (see e.g. \cite{DL:01}), but results of the exact type that we need did not seem to previously be present in the literature.  (We believe that ours is the first work that applies kernel methods to learn 
sums of independent random variables.)

In Section~\ref{sec:kernel} we develop tools for multidimensional kernel based learning that suit our context.  At its core, the kernel method approach that we develop allows us to do the following: Given a mixture of ${r}$ translates of $\bT$ and constant-factor approximations to $\gamma_\ell, \ldots, \gamma_{k-1}$, the kernel method allows us to learn this mixture to error $O(\epsilon)$ using only $\poly(1/\epsilon^\ell,{r})$ samples. Further, this algorithm is robust in the sense that the same guarantee holds even if the target distribution is only $O(\epsilon)$ close to having this structure (this is crucial for us).   Theorem~\ref{thm:learn-some-large} in Section~\ref{sec:learn} combines this tool with the ideas described above for learning a $\bT$-type distribution, and thereby establishes our general learning result for 
$\supportset$-sums with $|\supportset| \geq 4$.

\subsection{The case $|\supportset|=3$}

In this subsection we build on the discussion in the previous subsection, specializing to $k=|\supportset|=3$, and explain the high-level ideas of how we are able to learn with sample complexity $\poly(1/\eps)$ independent of $a_1,a_2,a_3$.   

For technical reasons (related to zero-moded distributions) there are three relevant parameters $t_1 \ll t_2  \ll t_3 = O_\eps(1)$ in the $k=3$ case.  The easy special case that each $\gamma_j \leq t_j$ is handled as discussed earlier (small essential support).  As in the previous subsection, let $\ell \in [3]$ be the least value such that $\gamma_\ell > t_\ell$.

In all the cases $\ell=1,2,3$ the analysis proceeds by considering the Light-Heavy-Experiment as discussed in the preceding subsection, i.e.~by approximating the target distribution $\bS$ by a mixture $\tilde{\bS}$ of shifts of the \emph{same} distribution $\bS_{\mathrm{heavy}}$.   When $\ell=3$, the ``heavy'' component $\bS_{\mathrm{heavy}}$ is simply a distribution of the form $q_3 \cdot\bS_3$ where $\bS_3$ is a signed PBD.  Crucially, while learning the distribution $\bT$ in the previous subsection involved guessing certain variances (which could be as large as $a_k$, leading to $\log a_k$ many possible outcomes of guesses and $\log \log a_k$ sample complexity), in the current setting the extremely simple structure of $\bS_{\mathrm{heavy}} = q_3 \cdot \bS_3$ obviates the need to make $\log a_3$ many guesses.  Instead, as we discuss in Section~\ref{sec:a0-is-three}, its variance can be approximated in a simple direct way by sampling just two points from $\bT$ and taking their difference; this easily gives a constant-factor approximation to the variance of $\bS_3$ with non-negligible probability. This success probability can be boosted by repeating this experiment several times (but the number of times does not depend on the $a_i$ values.)   We thus can use the kernel-based learning approach in a sample-efficient way, without any dependence on $a_1,a_2,a_3$ in the sample complexity.

For clarity of exposition, in the remaining intuitive discussion (of the $\ell=1,2$ cases) we only consider a special case:  we assume that $\bS = a_1 \cdot \bS_1  + a_2 \cdot \bS_2$ where both $\bS_1$ and $\bS_2$ are large-variance PBDs (so each random variable $\bX_i$ is either supported on $\{0,a_1\}$ or on $\{0,a_2\}$, but not on all three values $0,a_1,a_2$).  We further assume, clearly without loss of generality, that $\gcd(a_1,a_2)=1.$  (Indeed, our analysis essentially proceeds by reducing the $\ell=1,2$ 
case to this significantly simpler scenario, so this is a fairly accurate rendition of the true case.) 
 Writing $\bS_1 = \bX_1 + \ldots + \bX_{N_1}$ 
and $\bS_2 = \bY_1 + \ldots + \bY_{N_2}$, by zero-modedness we have that $\Pr[\bX_i=0] \geq {\frac 1 2}$ {and} $\Pr[\bY_i=0] \geq {\frac 1 2}$ for all $i$, so $\Var[\bS_j] = \Theta(1) \cdot \gamma_j$ for $j=1,2.$  We assume w.l.o.g.~in what follows that $a_1^2 \cdot \gamma_1 \ge a_2^2 \cdot \gamma_2$, so $\Var[\bS]$, which we henceforth denote $\sigma^2$, is $\Theta(1) \cdot a_1^2 \cdot \gamma_1.$

We now branch into three separate possibilities depending on the relative sizes of $\gamma_2$ and $a_1^2$. Before detailing these possibilities we observe that using the fact that $\gamma_1$ and $\gamma_2$ are both large, it can be shown that if we sample two points $s^{(1)}$ and $s^{(2)}$ from $\bS$, then with constant probability the value $\frac{|s^{(1)}-s^{(2)}|}{a_1}$ provides a constant-factor approximation to $\gamma_1$. 

\medskip
\noindent {\bf First possibility:}  $\gamma_2 < \epsilon^2 \cdot a_1^2.$ The algorithm samples two more points $s^{(3)}$ and $s^{(4)}$ from the distribution $\bS$.  The crucial idea is that with constant probability these two points can be used to obtain a constant-factor approximation to $\gamma_2$; we now explain how this is done. For $j \in \{3,4\}$, let 
$s^{(j)} = a_1 \cdot s^{(j)}_1  + a_2 \cdot s^{(j)}_2$ where $s^{(j)}_1 \sim \bS_1$ and $s^{(j)}_2 \sim \bS_2$, and consider the quantity $s^{(3)}-s^{(4)}.$ Since $\gamma_2$ is so small relative to $a_1$, the ``sampling noise" from $a_1 \cdot s^{(3)}_1 - a_1 \cdot s^{(4)}_1$ is likely to overwhelm the difference $a_2 \cdot s^{(3)}_2 - a_2 \cdot s^{(4)}_2$   at a ``macroscopic'' level. The key idea to deal with this  is to \emph{analyze the outcomes modulo $a_1$}.  In the modular setting, because $\Var[\bS_2] =\Theta(1) \cdot \gamma_2 \ll a_1^2$, one can show that with constant probability $|(a_2^{-1} \cdot  (s^{(3)}_2 - s^{(4)}_2)) \mod a_1|$ is a constant-factor approximation to  $\gamma_2$. (Note that as $a_1$ and $a_2$ are coprime, the operation $a_2^{-1}$ is well defined modulo $a_1$.) A constant-factor approximation to $\gamma_2$ can be used together with the constant-factor approximation to $\gamma_1$ to employ the aforementioned ``kernel method" based algorithm to learn the target distribution $\bS$. The fact that here we can use only two samples (as opposed to $\log \log a_1$ samples) to estimate $\gamma_2$ 
is really the crux of why for the $k=3$ case, the sample complexity is independent of $a_1$. (Indeed, we remark that\ignore{ this case turns out to be the ``hard'' case when dealing with 4-supported SICSIRVs over $\{0, a_1, a_2, a_3\}$ --- the} our analysis of the lower bound given by Theorem~\ref{known-k-is-four-lower} takes place in the modular setting and this ``mod $a_1$'' perspective is crucial for constructing the lower bound examples in that proof.)

\medskip
\noindent {\bf Second possibility:}  $a_1^2 /\epsilon^2>\gamma_2 > \epsilon^2 \cdot a_1^2.$ Here, by multiplicative gridding we can create a list of $O(\log(1/\eps))$ guesses such that at least one of them is a constant-factor approximation to $\gamma_2$. Again, we use the kernel method and the approximations to $\gamma_1$ and $\gamma_2$ to learn $\bS$. 

\medskip
\noindent {\bf Third possibility:}  The last possibility is that $\gamma_2 \ge a_1^2 / \epsilon^2$.  In this case, we show that $\bS$ is in fact $\epsilon$-close to the discretized Gaussian (with no scaling; recall that  $\gcd(a_1,a_2)=1$) that has the appropriate mean and variance. Given this structural fact, it is easy to learn $\bS$ by just estimating the mean and the variance and outputting the corresponding discretized Gaussian.  This structural fact follows from our new limit theorem, Lemma~\ref{l:mix}, mentioned earlier; we conclude this section with a discussion of this new limit theorem.

\ignore{
\rnote{ignore out what follows to end of subsection}
Since $\sigma^2$ is large (say at least $\Omega(1/\eps^c)$), it follows that the distribution $\bS$ is $\epsilon^{c/2}$-close  to a Gaussian in Kolmogorov (cdf) distance.   A useful intuitive view on this is that the distribution of $\bS$ ``looks like'' a Gaussian at a coarse scale.  A bit more precisely, this Kolmogorov-distance closeness means that if the support of $\bS$ is partitioned into intervals of length $\epsilon^{c/2} \cdot \sigma$, then in each of these intervals the probability mass of $\bS$ is $O(\eps^{c/2})$-close to that of the corresponding Gaussian.  So $\bS$ looks like a Gaussian ``at the scale of $\sigma$,'' but at a finer scale $\bS$ may no longer look like a Gaussian; since we are working with the (demanding) total variation distance measure rather than Kolmogorov distance, more analysis is required. (Looking ahead, since $\bS$ is a discrete distribution, it is natural to hope that $\bS$ is close to a discretized Gaussian or a scaled discretized Gaussian, and indeed this is where the analysis below will lead us.)

\ignore{
Under the assumption that $a_1^2 \cdot \gamma_1 \ge a_2^2 \cdot \gamma_2$, it is not difficult to show that at the scale of $a_1$, $\bS$ is close to a Gaussian. However, it may be very far from a Gaussian at finer scales. Whether or not this happens gives rise to the subcases now: 
}

We now branch into two separate possibilities depending on the relative sizes of $\gamma_2$ and $a_1^2$.

\medskip

\noindent {\bf First possibility:}  $\gamma_2 < a_1^2/\eps^2.$   Within this case, it turns out that the sub-case in which $\gamma_2 \in [a_1^2 \cdot \eps, a_1^2/\eps^2]$ 
can be handled by a slight extension of the arguments for the complementary sub-case that $\gamma_2 < a_1^2 \cdot \eps$, so we focus our attention in the following on this complementary sub-case, that $\gamma_2 < a_1^2 \cdot \eps$ (or more simply put, that $\gamma_2 \ll a_1^2$).
In this case $\bS$ resembles a discretized Gaussian ``at the scale of $a_1$", but at finer scales the distribution of $\bS$ may be substantially irregular. For example, consider the case when $\gamma_2 = a_1/2$. In such a case we note that $(\bS \mod a_1)$ can only take $O(\log(1/\eps)) \cdot \sqrt{ a_1}$ many different values with non-negligible probability (by Bernstein's inequality applied to $\bS_2$), whereas a Gaussian with variance $a_1^2 \cdot \Omega(\poly(1/\epsilon))$ may take values in all equivalence classes modulo $a_1$ with roughly the same probability.

For this case (that $\gamma_2 < a_1^2 \cdot \eps)$, it is not difficult to show that to learn $\bS$ it suffices to estimate the variances of both $\bS_1$ and $\bS_2$.  Our learning algorithm in this case accordingly proceeds in two steps.  The first step is straightforward:
 It is easy to estimate the variance of $\bS_1$ up to a $(1 \pm \epsilon)$ multiplicative factor.  This is because we can estimate the variance of $\bS$, and we know that $\Var[\bS] \ge a_1 \Var[\bS_1] \ge \Theta(\Var[\bS])$ from above.  Thus we can obtain $\Var[\bS_1]$  up to a constant factor: by making a multiplicative grid of possibilities for $\Var[\bS_1]$ where successive possibilities differ by multiplicative factors of $(1 + \epsilon$) and using hypothesis testing, we can obtain a high-accuracy estimate of $\Var[\bS_1]$. 

The more challenging step turns out to be estimating $\Var[\bS_2]$. Roughly, this is difficult because of the possibility that $a_2^2 \cdot \Var[\bS_2] \ll \Var[\bS]$; in such an event, the ``sampling noise" from $a_1 \cdot \bS_1$ might overwhelm the distinction between different outcomes of 
$\bS_2$ at a ``macroscopic'' level. The crucial idea to deal with this possibility is to \emph{analyze the outcomes of $\bS$ modulo $a_1$}.  In the modular setting, because $\Var[\bS_2] =\Theta(1) \cdot \gamma_2 \ll a_1^2$, it becomes possible to accurately estimate $\Var[\bS_2]$. (We are ignoring some details here but this is the core underlying idea.)  We remark that this case turns out to be the ``hard'' case when dealing with 4-supported SICSIRVs over $\{0, a_1, a_2, a_3\}$ --- the analysis of the lower bound given in \cite{DLS16lower} takes place in the modular setting and this ``mod $a_1$'' perspective is crucial for constructing the lower bound examples in that work.

\medskip

\noindent {\bf Second possibility:}  The remaining possibility is that $\gamma_2 \ge a_1^2 / \epsilon^2$.  In this case, we argue that $\bS$ is in fact $\epsilon$-close to the discretized Gaussian (with no scaling; recall that we have assumed $\gcd(a_1,a_2)=1$) that has the appropriate mean and variance. Given this structural fact, it is easy to learn $\bS$ by just estimating the mean and the variance and outputting the corresponding discretized Gaussian.  This structural fact follows from our new limit theorem, Lemma~\ref{l:mix}, mentioned in the previous subsection; we conclude this section with a discussion of this new limit theorem and the ideas that underlie it.

\ignore{
We now explain, at an intuitive level, why $\bS$ is close to a discretized Gaussian in this case.  First, it should be clear that $\bS$ ``looks like'' a discretized Gaussian at scale $a_1$ (recalling that $\bS = a_1 \cdot \bS_1 + a_2 \cdot \bS_2$ where $\bS_1,\bS_2$ are PBDs with variance at least $1/\eps^c$).  Since $\gamma_2 \ge a_1^2/\epsilon^2$, the PBD $\bS_2$ puts significant probability weight on each of at least $a_1/\eps$ many different values, so intuitively the $a_2 \cdot \bS_2$ component can ``fill in the holes" lying between successive multiples of $a_1$. 
The way to formalize this intuition is via coupling-based techniques that were used to prove central limit theorems for  total variation distance in the work of R\"ollin~\cite{Rollin:07}. We also note that in theoretical computer science, \cite{GMRZ11} have used coupling-based techniques in their construction of pseudorandom generators for combinatorial shapes (and indeed we use a coupling-based lemma of theirs as a technical ingredient in our proof). 
}

}
\subsection{Lemma~\ref{l:mix} and limit theorems.} \label{sec:limit} Here is a simplified version of our new limit theorem, Lemma~\ref{l:mix}, specialized to the case in which its ``$D$'' parameter is set to 2:

\medskip
\noindent {\bf Simplified version of Lemma~\ref{l:mix}.} \emph{Let $\bS = r_1 \cdot \bS_1 + r_2 \cdot \bS_2$ where 
$\bS_1,\bS_2$ are independent signed PBDs and $r_1,r_2$ are nonzero integers such that $\gcd(r_1,r_2)=1,$ $\Var[r_1\cdot  \bS_1] \geq \Var[r_2 \cdot\bS_2]$, and
$\Var[\bS_2] \geq \max\{{\frac{1}{\eps^8}},{\frac {r_1} {\eps}}\}.$
Then $\bS$ is $O(\eps)$-close in total variation distance to a signed PBD $\bS'$ (and hence to a signed discretized Gaussian) with $\Var[\bS'] = \Var[\bS].$
}

\medskip


If a distribution $\bS$ is close to a discretized Gaussian in Kolmogorov  distance and is $1/\sigma$-shift invariant (i.e. $\dtv(\bS, \bS+1) \leq 1/\sigma$), then $\bS$ is close to a discretized Gaussian in total variation distance
\cite{Rollin:07, Barbour:15}.
%
Gopalan et al.~\cite{GMRZ11} used a coupling based argument to establish a 
{similar} central limit theorem
to obtain 
{pseudorandom generators}
for certain space bounded branching programs. 
Unfortunately, in the setting of the lemma stated above, it is not immediately clear why $\bS$ should have $1/\sigma$-shift invariance. 
To deal with this, we give a novel analysis exploiting \emph{shift-invariance at multiple different scales}.\ignore{
introduce a new idea of ``long-range'' and ``short-range'' shift invariance and use it to establish the desired result.}
Roughly speaking, because of the $r_1 \cdot \bS_1$ component of $\bS$, it can be shown that $\dtv(\bS, \bS+r_1) = 1/\sqrt{\Var[\bS_1]}$, i.e. $\bS$ has good ``shift-invariance at the scale of $r_1$''; by the triangle inequality  $\bS$ is also not affected much if we shift by a small integer multiple of $r_1$.  The same is true for a few shifts by $r_2$, and hence also for a few shifts by
{\em both} $r_1$ and $r_2$.  If $\bS$ is approximated well by
a discretized Gaussian, though, then it is also not affected by small shifts, including shifts by $1$, and in fact we need such a guarantee to prove
approximation by a discretized Gaussian through coupling.  However, since $\gcd(r_1,r_2)=1$, basic number theory implies that we can achieve
any small integer shift via a small number of shifts by $r_1$ and $r_2$, and therefore $\bS$ has the required ``fine-grained'' shift-invariance (at scale 1) as well.  Intuitively, for this to work we need samples from $r_2 \cdot \bS_2$ to ``fill in the gaps'' 
between successive values of $r_1 \cdot \bS_1$ -- this is why we need  $\Var[\bS_2] \gg r_1$.


Based on our discussion with researchers in this area~\cite{Barbour:15} the idea of exploiting both long-range and short-range shift invariance is new to the best of our knowledge and seems likely to be of use in proving new central limit theorems.

\ignore{ To give an illustrative example of the kind of central limit theorem that can be obtained using this approach, we state the following CLT which is implicit in our analysis:
\begin{theorem}
Let $\bX_1, \ldots,  \bX_N$ be independent integer random variables supported in $\{0, a_1, a_2\}$ where $\mathsf{gcd}(a_1, a_2)=1$. Assume that for all $i \in [N]$ we have that 
$\Pr[\bX_i=0]$ is at least (say) $1/10$. Let $\gamma_1 = \sum_{i=1}^N \Pr[\bX_i = a_1]$ and 
$\gamma_2 = \sum_{i=1}^N \Pr[\bX_i = a_2]$, and assume (w.l.o.g.) that  
$a_1^2 \gamma_1 \ge a_2^2 \gamma_2$.  Then if $\gamma_2 \ge (a_1 / \epsilon)^2$, we have 
$\dtv(\bS, \mathcal{N}(\mathbf{E}[\bS], \Var[\bS]) = O(\epsilon)$, where $\mathcal{N}(\mu,\sigma^2)$ is the discretized Gaussian obtained by rounding the mean-$\mu$, variance-$\sigma^2$ Gaussian distribution to the nearest integer.
\end{theorem}

}

\section{Lower bound techniques} \label{s:ltechniques}

 In this section we give an overview of the ideas behind our lower bounds.  Both of our lower bounds actually work by considering restricted 
$\supportset$-sums:
our lower bounds can be proved using only
distributions $\bS$ of the form $\bS = \sum_{i=1}^k a_i \cdot \bS_i$, where $\bS_1, \dots, \bS_k$ are independent PBDs; equivalently, $\bS = \sum_{i=1}^N \bX_i$ where each $\bX_i$ is supported on one of $\{0, a_1\}$, $\dots$, $\{0, a_k\}$.  

\medskip
\noindent {\bf A useful reduction.}  The problem of learning a 
distribution
modulo an integer plays a key role in both of our lower bound arguments.
More precisely, both lower bounds use a reduction, which we establish, showing that an efficient algorithm for learning weighted PBDs with weights 
$0 < a_1 < ...< a_k$
implies an efficient algorithm for learning 
with weights $a_1, ..., a_{k-1}$ modulo $a_k$.
This problem is specified as follows: Consider
an algorithm which is given access to i.i.d.\ draws from the distribution $(\bS \mod a_k)$ (note that this distribution is supported over $\{0,1,\dots,a_{k}-1\}$) where $\bS$ is 
of the form $a_1 \cdot \bS_1 + ... + a_{k-1} \cdot \bS_{k-1}$
and  $\bS_1,...,\bS_{k-1}$ are PBDs.
The algorithm should produce a high-accuracy hypothesis distribution for $(\bS \mod a_k)$.  We stress that the example points provided to the learning algorithm all lie in $\{0,\dots,a_{k}-1\}$ (so certainly any reasonable hypothesis distribution should also be supported on $\{0,\dots,a_k-1\}$).  Such a reduction is useful for our lower bounds because it enables us to prove a lower bound for learning $\sum_{i=1}^k a_i \cdot \bS_i$ by proving a lower bound
for learning $\sum_{i=1}^{k-1} a_i \cdot \bS_i \mod a_k$.

\ignore{


}

The high level idea of this reduction is fairly simple so we sketch it here.  
Let $\bS = a_1 \cdot \bS_1 + \cdots + a_{k-1} \cdot \bS_{k-1}$ be a 
weighted sum of PBDs
such that $(\bS \mod a_k)$ is the target distribution to be learned
and let $N$ be the total number of summands in all of the
PBDs.
Let $\bS_k$ be an independent PBD with mean and variance $\Omega(N^\star)$. The key insight is that by taking $N^\star$ sufficiently large relative to $N$, the distribution of $(\bS \mod a_k) + a_k \cdot \bS_k$ (which can easily be simulated by the learner given access to draws from $(\bS \mod a_k)$ since it can generate samples from $a_k \cdot \bS_k$ by itself) can be shown to be statistically very close to that of $\bS' := \bS + a_k \cdot \bS_k$.
Here is an intuitive justification:  We can think of the different possible outcomes of $a_k \cdot \bS_k$ as dividing the support of $\bS'$ into bins of
width $a_k$.  Sampling from $\bS'$ can be performed by picking a bin boundary (a draw from $a_k \cdot \bS_k$)
and an offset $\bS$.
While adding $\bS$ may take the sample across multiple bin boundaries,
if $\Var[\bS_k]$ is sufficiently large, then adding $\bS$ typically takes
$a_k \cdot \bS_k + \bS$ across a small fraction of the bin boundaries.
Thus, the
conditional distribution given membership in a bin is similar between bins that have
high probability under $\bS'$, which means that all of these conditional distributions
are similar to the distribution of $\bS' \mod a_k$ (which is a mixture of them).
Finally, $(\bS' \mod a_k)$ has the same distribution as $(\bS \mod a_k)$.
Thus,  given samples from $(\bS \mod a_k)$, the learner can essentially simulate samples from $\bS'$. However, $\bS'$  is 
is a weighted sum of $k$ PBDs,
which by the assumption of our reduction theorem can be learned efficiently.  Now, assuming the learner has a hypothesis $\bH$ such that $\dtv(\bH, \bS')\le \epsilon$, it immediately follows that $\dtv((\bH \mod  a_k), (\bS' \mod a_k)) \le \dtv(\bH,\bS') \le \epsilon$ as desired. \ignore{This finishes the sketch of the reduction; while some of  details need to be ironed out, the actual proof essentially just formalizes the reasoning summarized here.}

\medskip

\noindent {\bf Proof overview of Theorem~\ref{known-k-is-four-lower}.} 
At this point we have the task of proving a lower bound for 
learning weighted PBDs
over $\{0, a_1, a_2\}$ mod $a_3$.  We establish such a lower bound using Fano's inequality (stated precisely as Theorem~\ref{fano} in Section~\ref{sec:prelim}).
To get a sample complexity lower bound of $\Omega(\log \log a_3)$ from Fano's inequality, we must construct $T = \log^{\Omega(1)} a_3$ distributions $\bS_1 $, $\dots$, $\bS_T $, where each $\bS_i$ is a 
weighted PBD
on $\{0, a_1, a_2\}$  modulo $a_3$, meeting the following requirements: $\dtv(\bS_i, \bS_j) = \Omega(1)$ if $i \not =j$, and $D_{KL}(\bS_i || \bS_j) = O(1)$ for all $i,j \in T.$ In other words, applying Fano's inequality requires us to exhibit a large number of distributions (belonging to the family for which we are proving the lower bound) such that any two distinct distributions in the family are \emph{far} in total variation distance but \emph{close} in terms of KL-divergence. The intuitive reason for these two competing requirements is that if $\bS_i$ and $\bS_j$ are $2\eps$-far in total variation distance, then a successful algorithm for learning to error at most $\eps$ must be able to distinguish $\bS_i$ and $\bS_j$. On the other hand, if $\bS_i$ and $\bS_j$ are close in KL divergence, then it is difficult for any learning algorithm to distinguish between $\bS_i$ and $\bS_j$. 

Now we present the high-level idea of how we may construct distributions $\bS_1,\bS_2,\dots$ with the properties described above to establish Theorem~\ref{known-k-is-four-lower}.  The intuitive description of $\bS_i$ that we give below does not align perfectly with our actual construction, but this simplified description is hopefully helpful in getting across the main idea.

For the construction we fix $a_1=1$, $a_2 =p$ and $a_3=q$.  (We discuss how $p$ and $q$ are selected later; this is a crucial aspect of our construction.)  The $i$-th distribution $\bS_i$ is $\bS_i = \bU_i + p \bV_i$ mod $q$; we describe the distribution $\bS_i = \bU_i + p \bV_i \mod q$ in two
stages, first by describing each $\bV_i$, and then by describing the corresponding $\bU_i$.  In the actual construction $\bU_i$ and $\bV_i$ will be shifted binomial distributions.  
Since a binomial distribution is rather flat within one standard deviation of its mean, and decays exponentially after that, it is qualitatively somewhat like the uniform distribution over an interval;
for this intuitive sketch it is helpful to think of $\bU_i$ and $\bV_i$ as actually being uniform distributions over intervals.  We take the support of $\bV_1$ to be an interval of length $q/p$, so that adjacent members of the support of $(p \bV_1 \mod q)$ will be at distance $p$ apart from each other.  More generally, taking $\bV_i$ to be uniform over an interval of length $2^{i-1}q/p$, the average gap between adjacent members of $\supp(p \bV_i \mod q)$
is of length essentially $p/2^{i-1}$, and by a careful choice of $p$ relative to $q$ one might furthermore hope that the gaps would be ``balanced'', so that they are all of length roughly $p/2^{i-1}$.  (This ``careful choice'' is the technical heart of our actual construction presented later.) 

How does $\bU_i$ enter the picture?  The idea is to take each $\bU_i$ to be uniform over a \emph{short} interval, of length $3p/2^i$.  This  ``fills in each gap'' and additionally ``fills in the first half of the following gap;'' as a result, the first half of each gap ends up with twice the probability mass of the second half.   (As a result, every two points have probability mass within a constant factor of each other under every distribution --- in fact, any point under any one of our distributions has probability mass within a constant factor of that of any other point under any other one of our distributions.  This 
gives the $D_{KL}(\bS_i || \bS_j) \leq O(1)$ upper bound mentioned above.) For example, recalling that the ``gaps'' in $\supp(p  \bV_1 \mod q)$ are of length $p$, choosing $\bU_1$ to be uniform over $\{1,\dots,3p/2\}$ will fill in each gap along with the first half of the following gap.  Intuitively, each $\bS_i = \bU_i + p \bV_i$ is a ``striped'' distribution, with equal-width ``light stripes'' (of uniformly distributed smaller mass) and ``dark stripes'' (of uniformly distributed larger mass), and each $\bS_{i+1}$ has stripes of width half of the $\bS_i$-sum's stripes. Roughly speaking, two such distributions $\bS_i$ and $\bS_j$ ``overlap enough'' (by a constant fraction) so that they are difficult to distinguish; however they are also ``distinct enough'' that a successful learning algorithm must be able to distinguish which $\bS_i$ its samples are drawn from in order to generate a high-accuracy hypothesis.

\ignore{
\green{

end green:}
end ignore:}

We now elaborate on the careful choice of $p$ and $q$ that was mentioned above.  The critical part of this choice of $p$ and $q$ is that for $i  \ge 1$, in order to get ``evenly spaced gaps,'' the remainders of $p \cdot s$ modulo $q$ where $s \in \{1, \dots, 2^{i-1}q/p\}$ should be roughly
evenly spaced, or \emph{equidistributed}, in the group $\mathbb{Z}_q$. Here the notion of ``evenly spaced" is with respect to the ``wrap-around" distance (also known as the \emph{Lee metric}) on the group $\mathbb{Z}_q$ (so, for example, the wrap-around distance between $1$ and $2$ is $1$, whereas the wrap-around distance between $q-1$ and $1$ is $2$). Roughly speaking, we would like $p \cdot s$ modulo $q$ to be equidistributed in $\mathbb{Z}_q$ when $s \in \{1, \dots, 2^{i-1}q/p\}$, for a range of successive values of $i$ (the more the better, since this means more distributions in our hard family and a stronger lower bound). Thus, qualitatively, we would like the remainders of $p$ modulo $q$ to be \emph{equidistributed at several scales.} We note that equidistribution phenomena are well studied in number theory and ergodic theory, see e.g.~\cite{Tao:HOF}. 

While this connection to equidistribution phenomena is useful for providing visual intuition (at least to the authors), in 
our attempts to implement the construction using powers of two that was just sketched, it seemed that in order to control the errors that arise in fact a \emph{doubly} exponential growth was required, leading to the construction of only
$\Theta(\log \log q)$ such distributions and hence a $\Omega (\log \log \log q)$ sample complexity lower bound.
Thus to achieve an $\Omega(\log \log q)$ sample complexity lower bound, our actual choice of $p$ and $q$ comes from the theory of continued fractions. In particular, we choose $p$ and $q$ so that $p/q$ has a continued fraction representation with ``many'' ($\Theta(\log q)$, though for technical reasons we use only $\log^{\Theta(1)} q$ many) convergents that grow relatively slowly.  These $T = \log^{\Theta(1)} q$ convergents translate into $T$ distributions $\bS_1, \dots, \bS_T$ in our ``hard family'' of distributions, and thus into an $\Omega(\log \log q)$ sample lower bound via Fano's inequality.  

The key property that we use is a well-known fact in the theory of continued fractions:  if $g_i/h_i$ is the $i^{th}$ convergent of a continued fraction for $p/q$, then $|g_i/h_i  - p/q| \le 1/(h_i \cdot h_{i+1})$. In other words, the $i^{th}$ convergent $g_i/h_i$ provides a non-trivially good approximation of $p/q$ (note that getting an error of $1/h_i$ would have been trivial). From this property, it is not difficult to see that the remainders of $p \cdot \{1, \dots, h_i\}$ are roughly equidistributed modulo $q$. 

Thus, a more accurate description of our (still idealized) construction is that we choose $\bV_i$ to be uniform on $\{1, \dots, h_i\}$ and $\bU_i$ to be uniform on roughly $\{1,\dots,(3/2) \cdot (q/h_i)\}$. So as to have as many distributions as possible in our family, we would like $h_i \approx (q/p) \cdot c^i$ for some fixed $c>1$. This can be ensured by choosing $p,q$ such that all the numbers appearing in the continued fraction representation of $p/q$ are bounded by an absolute constant; in fact, in the actual construction, we simply take $p/q$ to be a convergent of $1/\phi$ where $\phi$ is the golden ratio. With this choice we have that the  $i^{th}$ convergent of the continued fraction representation of $1/\phi$ is $g_i/h_i$, where $h_i \approx ((\sqrt{5}+1)/2)^i$. This concludes our informal description of the choice of $p$ and $q$. 

Again, we note that in our actual construction (see Figure~\ref{f:four-sicsirvs}), we cannot use uniform distributions over intervals (since we need to use PBDs), but rather we have shifted binomial distributions.  This adds some technical complication to the formal proofs, but the core ideas behind the construction are indeed as described above.

\begin{figure}[tbh]
\begin{center}
\begin{tabular}{ccc}
\begin{subfigure}{2in}
\includegraphics[width=2in]{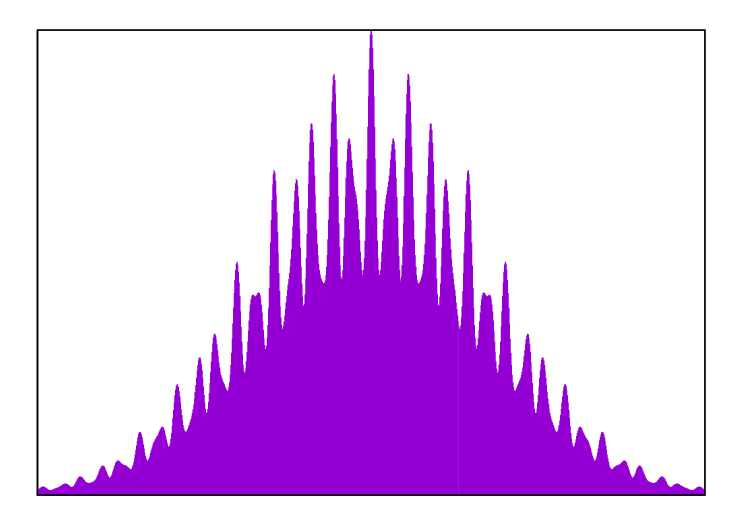}
\caption{}
\label{f:not_learnable_coarser}
\end{subfigure}
&
\begin{subfigure}{2in}
\includegraphics[width=2in]{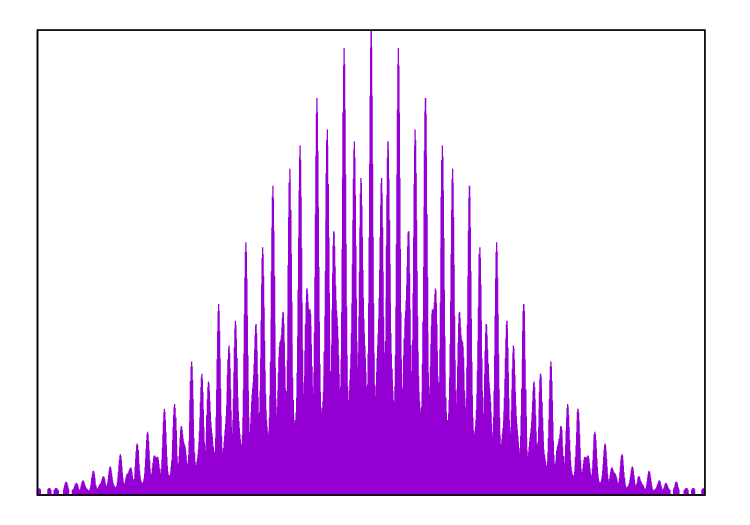}
\caption{}
\label{f:not_learnable_coarse}
\end{subfigure}
&
\begin{subfigure}{2in}
   \includegraphics[width=2in]{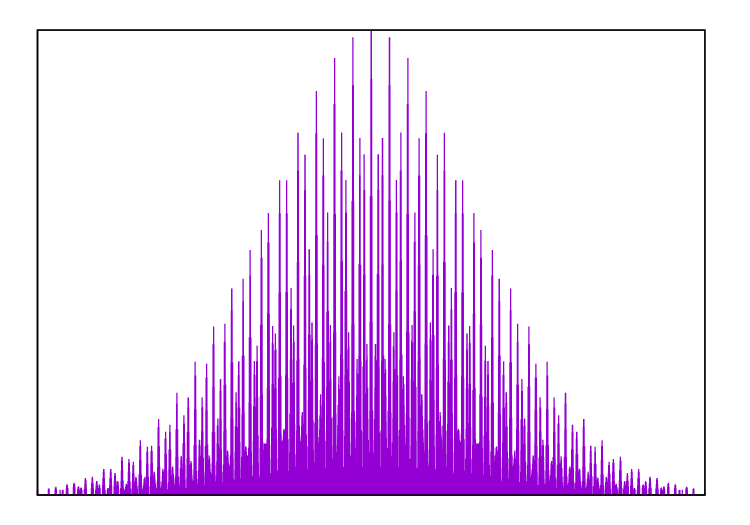}
 \caption{}
\label{f:not_learnable_fine}
\end{subfigure}
\end{tabular}
\end{center}
\caption{Examples of 
  targets
  used in our lower bound construction for Theorem~\ref{known-k-is-four-lower}.
  Very roughly,
  the 
  distribution
  in (\subref{f:not_learnable_coarse}) has peaks
  where the distribution
  (\subref{f:not_learnable_coarser}) does,
  plus a constant factor more peaks.  To compensate, its peaks are thinner.
  The distribution (\subref{f:not_learnable_fine}) has still more, 
  still thinner, peaks.
}
\label{f:four-sicsirvs}
\end{figure}

\medskip

\noindent {\bf Proof overview of Theorem \ref{unknown-k-is-three-lower}.}  As mentioned earlier, Theorem~\ref{unknown-k-is-three-lower} also uses our reduction from the modular 
learning problem.  Taking $a_1=0$ and $a_3 \approx \amax$ to be ``known'' to the learner, we show that any algorithm for learning a distribution of the form $(a_2 \bS_2 \mod a_3)$, where $0 < a_2 < a_3$ and $a_2$ is unknown to the learner and $\bS_2$ is a PBD$_N$, must use $\Omega(\log a_3)$ samples.  Like Theorem~\ref{known-k-is-four-lower}, we prove this using Fano's inequality, by constructing a ``hard family'' of $(a_3)^{\Omega(1)}$ many distributions of this type such that any two  distinct distributions in the family have variation distance $\Omega(1)$ but KL-divergence $O(1).$

We sketch the main ideas of our construction, starting with the upper bound on KL-divergence.  The value $a_3$ is taken to be a prime.  The same PBD$_N$ distribution $\bS_2$, which is simply a shifted binomial distribution and may be assumed to be ``known'' to the learner, is used for all of the distributions in the ``hard family'', so different distributions in this family differ only in the value of $a_2$.  The shifted binomial distribution $\bS_2$ is taken to have variance $\Theta((a_3)^2)$, so, very roughly, $\bS_2$ assigns significant probability on $\Theta(a_3)$ distinct values.  From this property, it is not difficult to show (similar to our earlier discussion) that any point in the domain $\{0,1,\dots,a_3-1\}$ under any one of our distributions has probability mass within a constant factor of that of any other point under any other one of our distributions (where the constant factor depends on the hidden constant in the $\Theta((a_3)^2)$).  This gives the required $O(1)$ upper bound on KL-divergence.

It remains to sketch the $\Omega(1)$ lower bound on variation distance.  As in our discussion of the Theorem~\ref{known-k-is-four-lower} lower bound, for intuition it is convenient to think of the shifted binomial distribution $\bS_2$ as being uniform over an interval of the domain $\{0,1,\dots,a_3-1\}$; by carefully choosing the variance and offset of this shifted binomial, we may think of this interval as being $\{0,1,\dots,r-1\}$ for $r = \kappa a_3$ for some small constant $\kappa >0$ (the constant $\kappa$ again depends on the hidden constant in the $\Theta((a_3)^2)$) value of the variance). So for the rest of our intuitive discussion we view the distributions in the hard family as being of the form $(a_2 \cdot \bU_r \mod a_3)$ where $\bU_r$ is uniform over $\{0,1,\dots,r-1\}$, $r=\kappa a_3$. 

Recalling that $a_3$ is prime, it is clear that for any $0 < a_2 < a_3$, the distribution $(a_2 \cdot \bU_r \mod a_3)$ is uniform over an $(r=\kappa a_3)$-element subset of $\{0,\dots,a_3-1\}.$  If $a_2$ and $a'_2$ are two independent uniform random elements from $\{1,\dots,a_3-1\}$, then since $\kappa$ is a small constant, intuitively the overlap between the supports of $(a_2 \cdot \bU_r \mod a_3)$ and
$(a'_2 \cdot \bU_r \mod a_3)$ should be small, and consequently the variation distance between these two distributions should be large.
This in turn suggests that by drawing a large random set of values for $a_2$, it should be possible to obtain a large family of distributions of the form $(a_2 \cdot \bU_r \mod a_3)$ such that any two of them have large variation distance.  We make this intuition precise using a number-theoretic equidistribution result of Shparlinski~\cite{Shparlinski:08} and a probabilistic argument showing that indeed a random set of $(a_3)^{1/3}$ choices of $a_2$ is likely to have the desired property.  This gives a ``hard family'' of size $(a_3)^{1/3}$, leading to an $\Omega(\log a_3) = \Omega(\log \amax)$ lower bound via Fano's inequality.  As before some technical work is required to translate these arguments for the uniform distribution over to the shifted binomial distributions that we actually have to work with, but we defer these technical details to Section~\ref{sec:unknown-lower}.

\section{Preliminaries}
\label{sec:prelim}

\subsection{Basic notions and useful tools from probability.}

\medskip \noindent {\bf Distributions.}  We will typically ignore the distinction between a random variable and its distribution.
We use bold font $\bX_i,\bS$, etc. to denote random variables (and also distributions).

For a distribution $\bX$ supported on the integers we write $\bX(i)$ to denote
the value $\Pr[\bX=i]$ of the probability density function of $\bX$ at point $i$, and $\bX(\leq i)$ to denote the value $\Pr[\bX \leq i]$ of the cumulative density function of $\bX$ at point $i$.
For $S \subseteq \Z$, we write $\bX(S)$ to denote $\sum_{i \in S}\bX(i)$ and $\bX_S$ to denote
the conditional distribution of $\bX$ restricted to $S.$  \ignore{Sometimes we write $X(I)$ and $X_I$ for a subset $I \subseteq [0,n]$, meaning $X(I \cap [n])$ and $X_{I \cap [n]}$ respectively.}

\medskip \noindent {\bf Total Variation Distance.} Recall that the {\em total variation distance} between two distributions
$\bX$ and $\bY$ over a countable set $D$ is
\begin{align*}
\dtv\left(\bX,\bY \right) &:= {\frac 1 2} \cdot \sum_{\alpha
\in D}{|\bX(\alpha)-\bY(\alpha)|} = \max_{S \subseteq D}[\bX(S)-\bY(S)],
\end{align*}
with analogous definitions for pairs of distributions over $\R$, over $\R^k$, etc.
Similarly, if $\bX$ and
$\bY$ are two random variables ranging over a countable set, their total
variation distance $\dtv(\bX,\bY)$ is defined as the total variation
distance between their distributions.
We sometimes write ``$\bX \stackrel{\eps}{\approx} \bY$'' as shorthand for ``$\dtv(\bX,\bY) \leq \eps$''.

For $\bX$ and $\bY$ with $\dtv(\bX,\bY) \leq \eps$, the following coupling lemma justifies thinking
of a draw from $\bY$ as being obtained by making a draw from $\bX$, and modifying it with
probability at most $\epsilon$.
\begin{lemma}[\cite{lindvall2002lectures}]
\label{l:coupling}
For random variables $\bX$ and $\bY$ with $\dtv(\bX,\bY) \leq \eps$, there is
a joint distribution whose marginals are $\bX$ and $\bY$ such that, with probability at
least $1 - \eps$, $\bX = \bY$.
\end{lemma}

\medskip \noindent {\bf Shift-invariance.}
Let $\bX$ be a finitely supported real-valued random variable. 
For an integer $k$ we write $\dshift{k}(\bX)$ to denote $\dtv(\bX, \bX+k)$.  We say that $\bX$ is \emph{$\alpha$-shift-invariant at scale $k$} if $\dshift{k}(\bX) \leq \alpha$; if $\bX$ is $\alpha$-shift-invariant at scale 1 then we sometimes simply say that $\bX$ is \emph{$\alpha$-shift-invariant}.
We will use the following basic fact:

\begin{fact} \label{fact:SI}

\begin{enumerate}
\item
If $\bX,\bY$ are independent random variables then $\dshift{k}(\bX + \bY) \leq \dshift{k}(\bX).$

\item Let $\bX$ be $\alpha$-shift-invariant at scale $p$ and $\bY$ (independent
 from $\bX$) be $\beta$-shift-invariant at scale $q$.  Then $\bX+\bY$ is
 both $\alpha$-shift-invariant at scale $p$ and
 $\beta$-shift-invariant at scale $q$.

\end{enumerate}
\end{fact}

\medskip \noindent {\bf Kolmogorov Distance and the DKW Inequality.}  Recall that the {\em Kolmogorov distance} $\dk(\bX,\bY)$ between probability distributions
over the integers is 
\[
\dk(\bX,\bY):=\max_{j\in \Z} |\Pr[\bX \leq j] - \Pr[\bY \leq j]|,
\]
and hence for any interval $I=\{a,a+1,\dots,a+b\} \subset \Z$ we have that
\[
|\Pr[\bX \in I] - \Pr[\bY \in I]| \leq 2 \dk(\bX,\bY).
\]

Learning any distribution with respect to the Kolmogorov distance is relatively easy,
which follows from the \emph{Dvoretzky-Kiefer-Wolfowitz} (DKW) inequality.
Let $\wh{\bX}_m$ denote the empirical distribution of $m$ i.i.d. samples drawn from $\bX.$
The DKW inequality states that for $m=\Omega((1/\eps^2)\cdot \ln(1/\delta))$, with probability $1-\delta$
(over the draw of $m$ samples from $\bX$)
the {empirical distribution} $\wh{\bX}_m$ will be
$\eps$-close to $\bX$ in Kolmogorov distance:
\begin{theorem}[\cite{DKW56,Massart90}] \label{thm:DKW}
Let $\widehat{\bX}_m$ be an empirical distribution of $m$ samples from distribution $\bX$ over the integers.
Then for all $\eps>0$, we have
$$\Pr [ \dk(\bX, \widehat{\bX}_m) > \eps ] \leq 2e^{-2m\eps^2}.$$
\end{theorem}

Convolving with an $\alpha$-shift invariant distribution can 
``spread the butter'' 
to transform distributions that
are close w.r.t.\ Kolmogorov distance into distributions that are close with respect to the more demanding total variation distance.
The following lemma makes this intuition precise:
\begin{lemma}[\cite{GMRZ11}]\label{lem:GMRZ}
Let $\bY, \bZ$ be distributions supported on the integers and $\bX$ be an $\alpha$-shift invariant distribution that is independent of $\bY,\bZ$.  Then for any $a,b$ such that\ignore{\pnote{needed for $d \leq b$ in their proof}} $\dk(\bY,\bZ) \leq \alpha b$, we have
\[
\dtv(\bY+\bX, \bZ+\bX) = O(\sqrt{\dk(\bY,\bZ) \cdot \alpha \cdot b}) + \Pr[\bY \notin [a,a+b)] + \Pr[\bZ \notin [a,a+b)].
\]
\end{lemma}

We will also require a multidimensional generalization of Kolmogorov distance and of the DKW inequality. 
Given probability distributions $\bX,\bY$ over $\Z^d$, the Kolmogorov distance between $\bX$ and $\bY$ is
\[
\dk(\bX,\bY):=\max_{(j_1,\dots,j_d) \in \Z^d} |\Pr[\bX_i \leq j_i \text{~for all~}i\in[d]] - \Pr[\bY \leq j_i \text{~for all~}i\in[d]]|,
\]
and so for any axis-aligned rectangle $R = \prod_{i=1}^d \{a_i,\dots,a_i +b_i\} \subset \Z^d$ we have
\[
|\Pr[\bX \in R] - \Pr[\bY \in R]| \leq 2^d \dk(\bX,\bY).
\]  
We will use the following generalization of the DKW inequality to the multidimensional setting.
\begin{lemma}[\cite{talagrand1994sharper}]
\label{lem:DKW-multidim} 
Let $\widehat{\bX}_m$ be an empirical distribution of $m$ samples from distribution $\bX$ over $\Z^d$.
There are absolute constants $c_1$, $c_2$ and $c_3$ such that, 
for all $\eps>0$, for all $m \geq c_1 d/\epsilon^2$,\;
$$\Pr [ \dk(\bX, \widehat{\bX}_m) > \eps ] \leq c_2^d e^{-c_3 \epsilon^2 m}. $$
\end{lemma}

\medskip \noindent {\bf Covers.} Let ${\cal P}$ denote a set of
distributions over the integers. Given $\delta > 0$, a set of distributions ${\cal Q}$ is said to be a \emph{$\delta$-cover of ${\cal P}$} (w.r.t. the
total variation distance) if for every distribution $\bP$ in ${\cal P}$ there
exists some distribution $\bQ$ in ${\cal Q}$ such that
$\dtv(\bP,\bQ) \leq \delta.$
We sometimes say that distributions $\bP,\bQ$ are
\emph{$\delta$-neighbors} if $\dtv(\bP,\bQ) \leq \delta$, or that $\bP$ and $\bQ$ are \emph{$\delta$-close}.

\medskip \noindent {\bf Support and essential support.}  We write $\supp(\bP)$ to denote the support of distribution $\bP.$  Given a  distribution $\bP$ over the integers, we say that $\bP$ is \emph{$\tau$-essentially supported on $S \subset \Z$} if $\bP(S) \geq 1-\tau.$  \ignore{We will use the following
simple fact, which says that the set of all distributions with essential support $S$ has a ``small'' cover which can be easily enumerated given $S$:\rnote{I think we should introduce the phrase ``essential support'', but it doesn't look like we ever use Fact~\ref{fact:small-cover}; shall we remove it?}

\begin{fact} [``Small'' cover for sparsely supported distributions.] \label{fact:small-cover}
Fix a finite set $S$ of integers, and let ${\cal P}_{\tau,S}$ be the set of all distributions that are $\tau$-essentially supported on $S$.  There is a set ${\cal Q}$ of $|{\cal Q}| \leq (|S|/\tau + 1)^{|S|}$ many distributions, all of which are supported on 
$|S|$, that form a $(2 \tau)$-cover of $S$.  Moreover, given $S$ and $\tau$ it is possible to list the elements of ${\cal Q}$ in
$\poly((|S|/\tau)^{|S|})$ time.
\end{fact}
}

\subsection{The distributions we work with.}

We recall the definition of 
an $\supportset$-sum
and give some related definitions.
For $0 \leq a_1 < ... < a_k$ and $\supportset = \{ a_1,...,a_k\}$,
a $\supportset$-sum is a distribution
$\bS = \sum_{i=1}^N \bX_i$ where the $\bX_i$'s are independent integer random variables (\emph{not} assumed to be identically distributed) all of which are supported on the same set of integer values $a_1 < a_2 < \cdots < a_k \in \Z_{\geq 0}.$  A \emph{Poisson Binomial Distribution}, or PBD$_N$, is a 
$\{ 0,1 \}$-sum.
\ignore{We will sometimes consider \emph{generalized SICSIRVs} in which the $a_i$'s are allowed to be negative.  If $\bS$ is a generalized SICSIRVs over a set of values $\{a_1,\dots,a_k\}$ that includes 0, we say that $\bS$ is \emph{$0$-moded} if each $\bX_i$ has mode 0.}

A {\em weighted sum of PBDs} is a distribution
$\bS = a_2 \bS_2 + \cdots + a_k \bS_k$  where each $\bS_i$ is an independent PBD$_{N_i}$ and $N_2 + \cdots + N_k = N.$  Equivalently we have that $\bS = \sum_{i=1}^{N} \bX_i$ where $N_2$ of the $\bX_i$'s are supported on $\{0,a_2\}$, $N_3$ are supported on $\{0,a_3\}$, and so on.\ignore{  (Note that this include the case in which there are only $N$ instead of $(k-1)N$ summands $\bX_i$ since the ``extra'' $(k-2)N$ $\bX_i$'s can all be viewed as trivial, i.e. supported only on 0.)}  

Let us say that a \emph{signed PBD}$_N$ is a random variable $\bS = \sum_{i=1}^N \bX_i$ where the $\bX_i$'s are independent and each is either supported on $\{0,1\}$ or is supported on $\{0,-1\}.$  
We defined a weighted sum of signed PBDs analogously to
the unsigned case.

Finally, we say that an integer valued random variable $\bX$ has \emph{mode 0} if $\Pr[\bX=0] \geq \Pr[\bX = b]$ for all $b \in \Z.$

\medskip \noindent {\bf Translated Poisson Distributions and Discretized Gaussians.} We will make use of the translated Poisson distribution for approximating signed PBDs with large variance. 
\begin{definition}[\cite{Rollin:07}]
We say that an integer random variable $\bY$ is distributed according to the {\em translated Poisson distribution with parameters $\mu$ and $\sigma^2$}, denoted $TP(\mu, \sigma^2)$,
iff $\bY$ can be written as
\[
\bY = \lfloor \mu-\sigma^2\rfloor + \bZ,
\]
where $\bZ$ is a random variable distributed according to ${\rm Poisson}(\sigma^2+ \{\mu-\sigma^2\})$, where
$\{\mu-\sigma^2\}$ represents the fractional part of $\mu-\sigma^2$.
\end{definition}
The following lemma gives a useful bound on the variation distance between a signed PBD and a suitable translated Poisson distribution.  

\begin{lemma} \label{lem:translated Poisson approximation}
Let $\bS$ be a signed PBD$_N$ with mean $\mu$ and variance $\sigma^2 \geq 1.$
Then
$$\dtv\left(\bS, TP(\mu, \sigma^2)\right) \le O(1/\sigma).$$
\end{lemma}
\begin{proof}
Without loss of generality we may suppose that $\bS = \bX_1 + \cdots + \bX_N$ where $\bX_1,\dots,\bX_M$
are supported on $\{0,-1\}$ with $\E[\bX_i] = -p_i$ for $i \leq M$, and $\bX_{M+1},\dots,\bX_N$ are supported on $\{0,1\}$ with $\E[\bX_i]=p_i$ for $i >M.$  Let $\bX'_i = \bX_i + 1$ for $1 \leq i \leq M$, so $\bS' := \bX'_1 + \cdots + \bX'_M + \bX_{M+1} + \cdots + \bX_N$ are independent Bernoulli random variables where $\E[\bX'_i]=1-p_i$ for $i \leq M.$  

\cite{Rollin:07} (specifically equation (3.4)) shows that 
if $\bJ_1,\ldots,\bJ_N$ are independent Bernoulli random variables with $\E[\bJ_i]=p_i$, then 
\[
\dtv\left(\sum_{i=1}^N \bJ_i, TP(\mu,\sigma^2)\right) \leq 
\frac{\sqrt{\sum_{i=1}^Np_i^3(1-p_i)}+2}{\sum_{i=1}^Np_i(1-p_i)}\]
where $\mu = \sum_{i=1}^N p_i.$
Applying this to $\bS'$, we see that for $\mu' = \E[\bS']$, we have
\begin{align*}
\dtv\left(\bS',TP(\mu',\sigma^2)\right) & \leq
\frac{\sqrt{\sum_{i=1}^Mp_i(1-p_i)^3 + \sum_{i=M+1}^N p_i^3(1-p_i)}+2}{\sum_{i=1}^Np_i(1-p_i)}\\
& \leq
\frac{\sqrt{\sum_{i=1}^Np_i(1-p_i) }+2}{\sum_{i=1}^Np_i(1-p_i)} \leq O(1/\sigma).
\end{align*}
The claimed bound follows from this on observing that $\bS'$ is a translation of $\bS$ by $M$ and $TP(\mu',\sigma^2)$ is likewise a translation of $TP(\mu,\sigma^2)$ by $M$.
\end{proof}

The following bound on the total variation distance between translated Poisson distributions will be useful.  
\begin{lemma}[Lemma~2.1 of \cite{BarbourLindvall}] \label{lem: variation distance between translated Poisson distributions}
For $\mu_1, \mu_2 \in \mathbb{R}$ and $\sigma_1^2, \sigma_2^2 \in \mathbb{R}_+$ with $\lfloor \mu_1-\sigma_1^2 \rfloor \le \lfloor \mu_2-\sigma_2^2 \rfloor$, we have
$$\dtv(TP(\mu_1,\sigma_1^2), TP(\mu_2,\sigma_2^2)) \le \frac{|\mu_1-\mu_2|}{\sigma_1}+\frac{|\sigma_1^2-\sigma_2^2|+1}{\sigma_1^2}.$$
\end{lemma}

We will also use discretized Gaussians, both real-valued and vector-valued (i.e. multidimensional).  A draw from
the \emph{discretized Gaussian} $\calN_D(\mu,\sigma^2)$ is obtained by making a draw from the normal
distribution $\calN(\mu,\sigma)$ and rounding to the nearest integer.  We refer to $\mu$ and $\sigma^2$ respectively as the ``underlying mean'' and ``underlying variance'' of $\calN_D(\mu,\sigma).$  Similarly, a draw from the \emph{multidimensional 
discretized Gaussian} $\calN_D(\mu,\Sigma)$ is obtained by making a draw from the multidimensional
Gaussian $\calN(\mu,\Sigma)$ with mean vector $\mathbf{\mu}$ and covariance matrix $\Sigma$ and rounding each coordinate to the nearest integer.   To avoid confusion we will always explicitly write ``multidimensional'' when dealing with a multidimensional Gaussian.

We recall some simple facts about the variation distance between different discretized Gaussian distributions (see 
Appendix~B of the full version of \cite{DDOST13}):

\begin{lemma} [Proposition B.5 of \cite{DDOST13}] 
\label{lem:B5}
Let $\bG$ be distributed as $\calN(\mu,\sigma^2)$ and let $\lambda \in \R$.  Then
$\dtv(\lfloor \bG + \lambda \rfloor, \lfloor \bG \rfloor + \lfloor \lambda \rfloor) \leq {\frac 1 {2 \sigma}}.$
\end{lemma}
The same argument that gives Lemma~\ref{lem:B5} also gives the following small extension:

\begin{lemma} \label{lem:B5plus}
Let $\bG$ be distributed as $\calN(\mu,\sigma^2)$ and let $\lambda \in \R, \rho \in \Z.$ Then
$\dtv(\lfloor \bG + \lambda \rfloor, \lfloor \bG \rfloor + \rho) \leq {\frac {|\rho - \lambda|} {2 \sigma}}.$
\end{lemma}

We will use the following theorem about approximation
of signed PBDs.
\begin{theorem}[\cite{CGS11} Theorem 7.1\footnote{The theorem in \cite{CGS11} is stated only for PBDs, but the result for signed PBDs is easily derived from the result for PBDs via a simple translation argument similar to the proof of Lemma~\ref{lem:translated Poisson approximation}.}]\label{thm:CGS}
For $\bS$ a signed PBD, $\dtv( \bS, \mathcal{N}_D(\mu, \sigma^2)) \le O(1/\sigma)$ where $\mu = \E[\bS]$ and $\sigma^2 = \Var[\bS]$.
\end{theorem}

The following is a consequence of Theorem~\ref{thm:CGS} and Lemma~\ref{lem:B5plus} which we explicitly record for later reference:

\begin{fact} \label{fact:good-shift-invariance}
Let $\bS$ be a signed PBD with $\Var[\bS] = \sigma_{\bS}^2.$  Then $\bS$ is $\tau$-shift-invariant at scale 1 for $\tau = O(1/\sigma_{\bS})$,
and hence for any integer $c$, the distribution $c \bS$ is $\tau$-shift-invariant at scale $c$.
\ignore{Moreover, for all $\mu,\sigma,$ the distribution $TP(\mu, \sigma^2)$ is $\tau$-shift-invariant at scale 1 for $\tau = O(1/\sigma).$}
\end{fact}

We also need a central limit theorem for \ignore{ $0$-moded }multinomial distributions.  We recall the following result, which is a direct consequence of the ``size-free CLT'' for Poisson Multinomial Distributions in  \cite{DDKT16}.  (Below we write $\mathbf{e}_i$ to denote the real vector in $\{0,1\}^{d}$ that has a 1 only in the $i$-th coordinate.)

\begin{theorem}\label{thm:multinomial}
Let $\bX_1, \ldots,\bX_N$ be independent $\Z^d$-valued random variables where the support of each $\bX_i$ is contained in the set $\{0,\pm  \mathbf{e}_1,\dots, \pm \mathbf{e}_{d}\}$.\ignore{Assume the variables are $0$-moded, so for all $ 1 \le i \le N$ we have $\Pr[\bX_i=0] \ge 1/2k$, and l} Let $\bM=\bX_1 + \ldots + \bX_N$.  Then we have 
\[
\dtv( \bM, \mathcal{N}_D(\mu, \Sigma)) \le O\left({\frac {d^{7/2}}{\sigma^{1/10}}}\right),
\]
where $\mu = \E[\bM]$ is the mean and $\Sigma$ is the $d\times d$ covariance matrix of $\bS$, and $\sigma^2$ is the minimum eigenvalue of $\Sigma$. 
\end{theorem}

\medskip \noindent {\bf Covers and structural results for PBDs.}
\ignore{A {\em Poisson binomial distribution (PBD) of order $N \in \mathbb{N}$} is a sum $\sum_{i=1}^N \bX_i$ of $N$ mutually independent Bernoulli (0/1) random variables $\bX_1,\ldots,\bX_N$. We denote the set of all PBD$_N$ distributions of order $N$ by ${\cal S}_{N}$.}Our proof of Theorem~\ref{unknown-k-is-two-upper}, which is about learning PBDs that have been subject to an unknown shifting and scaling, uses the fact that for any $\eps$ there is a ``small cover'' for the
set of all PBD$_N$ distributions.  We recall the following from \cite{DP14}:
\begin{theorem} [Cover for PBDs] \label{thm: sparse cover theorem}
Let $\bS$ be any PBD$_N$ distribution.  Then for any $\eps>0$, we have that either

\begin{itemize}

\item $\bS$ is $\eps$-essentially supported on an interval of $O(1/\eps^3)$ consecutive integers (in this case we say that $\bS$ is in \emph{sparse form}); or if not,

\item $\bS$ is $\eps$-close to some distribution $u + \Bin(\ell,q)$ where $u, \ell \in \{0,1,\dots,N\},$ and $\Var[\Bin(\ell,q)] = \Omega(1/\eps^2)$ (in this case we say that $\bS$ is in \emph{$1/\eps$-heavy Binomial form}).

\end{itemize}
\end{theorem}

We recall some well-known structural results on PBDs that are in $1/\eps$-heavy Binomial form (see e.g. \cite{CGS11}, Theorem~7.1 and p. 231):

\begin{fact} \label{fact:heavy-PBD-nice}
Let $\bY$ be a PBD$_N$ distribution that is in $1/\eps$-heavy Binomial form
as described in Theorem \ref{thm: sparse cover theorem}.
Then

\begin{enumerate}

\item $\dtv(\bY,\bZ) = O(\eps)$, where $\bZ$ is a discretized $\calN(\E[\bY],\Var[\bY])$ Gaussian.
\item $\dshift{1}(\bY) = O(\eps)$.

\end{enumerate}

\end{fact}

\subsection{Extension of the Barbour-Xia coupling lemma}
In \cite{BX99}, Barbour and Xia proved the following lemma concerning the shift-invariance of sums of independent integer random variables. 
\begin{lemma}[\cite{BX99}, Proposition~4.6]\label{lem:BX}
Let $\bX_1, \ldots, \bX_N$ be $N$ independent integer valued random variables and let $\bS= \bX_1 + \ldots + \bX_N$. Let $\dshift{1} (\bX_i) \le 1 - \delta_i$. Then, 
$$
\dshift{1}(\bS) \le O\bigg( \frac{1}{\sqrt{\sum_{i=1}^N \delta_i}}\bigg). 
$$
\end{lemma}
We require a $\dshift{p}$ analogue of this result. To obtain such an analogue we first slightly generalize the above lemma so that it does not require $\bX_i$ to be supported on $\mathbb{Z}$. The proof uses a simple reduction to the integer case. 
\begin{claim}\label{clm:soupedup}
Let $\bX_1, \ldots, \bX_N$ be $N$ independent  finitely supported random variables and let $\bS= \bX_1 + \ldots + \bX_N$. Let $\dshift{1} (\bX_i) \le 1 - \delta_i$. Then, 
$$
\dshift{1}(\bS) \le O\bigg( \frac{1}{\sqrt{\sum_{i=1}^N \delta_i }}\bigg). 
$$
\end{claim}
\begin{proof}
Assume that  for any $i$, the support of $\bX_i$ is of size at most $k$ and is supported in the interval $[-k, k]$. (By the assumption of finite support this must hold for some integer $k$.)
Given any $\bX_i$, create a new random variable $\bY_i$ which is defined as follows: 
First, let us partition the support of $\bX_i$ by putting two outcomes into the same cell whenever the difference between them is
an integer.
Let $S^{(i)}_1,\dots,S^{(i)}_{k'}$ be the non-empty cells, so $k' \leq k$, and, for each $S_j^{(i)}$, there is a real $\beta_j$ such that
$S_j^{(i)} \subseteq \{ \beta_j + \ell: \ell \in \Z \}$.  Let $\gamma_{j,i}$ denote the smallest element of $S^{(i)}_j.$\ignore{\rnote{This was ``Let the support of $\bX_i$ be partitioned into sets $S_1^{(i)}, \ldots, S_k^{(i)}$ where we have the following:
\begin{itemize}
\item The difference between any two elements in $S_j^{(i)}$ belongs to $\mathbb{Z}$. 
\item The smallest element in $S_j^{(i)}$ is $\gamma_{j,i}$. 
\end{itemize}
''.  If what's there now looks okay, feel free to ignore out this rnote.}}
Let us define integers $\{m_{j,i}\}_{1\le j \le k , 1 \le i \le N}$ as follows: $m_{j,i} = (N\cdot k)^{k \cdot i +j}$. 
 The random variable $\bY_i$ is defined as follows:   For all $\ell \in \mathbb{Z}^+$, let the map $M_i$ send $\gamma_{j,i}+\ell$ to $ m_{j,i} +\ell$. The probability distribution of $\bY_i$ is the distributed induced by the map $M_i$ when acting on $\bX_i$, i.e. a draw from $\bY_i$ is obtained by drawing $x_i$ from $\bX_i$ and outputting $M_i(x_i).$   It is clear that $\bY_i$ is integer-valued and satisfies
\[
\dshift{1}(\bX_i)  = \dshift{1}(\bY_i).  
\]
Now  consider a sequence of outcomes $\bY_1 = y_1, \ldots, \bY_N=y_N$ and $\bY'_1=y'_1, \ldots, \bY'_N = y'_N$ such that 
\[
\big| \sum_{i=1}^N (y_i - y'_i)  \big| = 1. 
\]
We can write each $y_i$ as $m_{\alpha_i,i} + \delta_i$ where each $1 \le \alpha_i \le k$ and each $-k \le \delta_i \le k$. Likewise, $y'_i = m_{\alpha'_i,i} + \delta'_i$ where each $1 \le \alpha'_i \le k$ and each $-k \le \delta'_i \le k$. Since $m_{j,i} = (N\cdot k)^{k \cdot i +j}$, it is easy to see that the following must hold:
$$
\textrm{For all }i=1,\dots,N,  \quad  m_{\alpha_i,i}=m_{\alpha'_i,i} \quad \textrm{and} \quad \left| \sum_{i=1}^N (\delta_i - \delta'_i) \right|=1. 
$$
This immediately implies that
$
\dshift{1}\left(\sum_{i=1}^N \bX_i\right) = \dshift{1}\left(\sum_{i=1}^N \bY_i\right). 
$
Applying Lemma~\ref{lem:BX}, we have that 
\[
\dshift{1}(\sum_{i=1}^N \bY_i) \le  O\left( \frac{1}{\sqrt{\sum_{i=1}^N \delta_i}}\right),
\]
which finishes the upper bound. 
\end{proof}
This immediately yields the following corollary.
\begin{corollary}\label{corr:pshift}
Let $\bX_1, \ldots, \bX_N$ be finitely supported independent integer valued random variables. Let $\dshift{p}(\bX_i) \le 1- \delta_i$. Then, for $\bS= \sum_{i=1}^N \bX_i$, we have
$$
\dshift{p}(\bS) = O \left( \frac{1}{\sqrt{\sum_{i=1}^N  \delta_i}}\right).
$$
\end{corollary}
\begin{proof}
Let $\bY_i = \bX_i/p$ for all $1 \le i \le N$. Then for $\bS' = \sum_{i=1}^N \bY_i$, it is clear that $\dshift{1}(\bS') = \dshift{p}(\bS)$. Applying Claim~\ref{clm:soupedup}, we get the corollary. 
\end{proof}

\ignore{

\subsection{Learning results for PBDs and $k$-SIIRVs.}

We will use some known learning results for PBDs and $k$-SIIRVs, specifically the following:

\begin{theorem} [Learning PBDs] \label{thm:learn-PBD} 
Let $\bX = \sum_{i=1}^n \bX_i$ be an unknown PBD.   There is an algorithm with the following properties:
given $n, \eps,\delta$ and access to independent draws from $\bX$, the algorithm uses \green{$\tilde{O}\left( (1/\eps^3)  \cdot \log(1/\delta) \right)$ samples from $X$,
runs in time $\tilde{O} \left( (1/\eps^3) \cdot \log n  \cdot \log^2 {1 \over \delta} \right)$,}
and with probability at least $1-\delta$ outputs a (succinct description of a) distribution $\hat{\bX}$ over $[n]$
which is such that $\dtv(\hat{\bX},\bX) \leq \eps.$
\end{theorem}

}

\subsection{Other background results on distribution learning.} 

\ignore{

\medskip \noindent {\bf Estimating the mean of a distribution.}  
\green{\rnote{I think perhaps we won't need this -- we can just guess, and may as well given all the guessing we are doing already.  So maybe remove this green stuff eventually.}
We recall the following standard fact which says that it is possible to estimate the mean of a distribution to high accuracy (measured relative to the variance of the distribution) with high confidence.  (For a detailed proof, which uses Chebychev's inequality and a standard median-based confidence boosting procedure, see e.g. Lemma 6 of \cite{DDS15algorithmica}.)\rnote{I'm not sure if there is a similarly generic way to argue that we can estimate variance easily for non-PBD distributions.  For a completely general distribution over $[pN]$ I guess variance estimation to $(1\pm \eps)$ multiplicative accuracy might take a number of samples depending on $p,N,$ right?  Of course we're never going to be dealing with completely generic distributions and in fact it is possible to estimate the variance on the cheap in th situations where we need to do so --- we just want to do this in the lowest hassle way possible.  Right now I've been taking a guessing-based approach, see e.g. Sections 3.2.3 and 3.3.1 ``Setup'' (where there is some amount of hassle).  This is not a major concern but eventually we should decide how we will do these sorts of things.}

\begin{fact} \label{fact:estimate-mean}
Let $\bD$ be any distribution with finite variance.  There exists an algorithm with the following properties:  given access to independent draws from $\bD$, it produces an estimate $\hat{\mu}$  for $\mu=\E[\bD]$such that with probability at least $1-\delta$, we have
\[|\mu-\hat{\mu}| \le \epsilon \cdot \sigma,\]
where $\sigma^2 = \Var[\bD].$
The algorithm uses $O(\log(1/\delta)/\epsilon^2)$ samples and runs in time
$O(\log(1/\delta)/\epsilon^2).$
\end{fact}
}
}
\ignore{
\noindent {\bf Efficiently estimating variance of a PBD$_N.$}
A useful fact is that it is possible to efficiently estimate the variance of a PBD$_N$.  We recall the following from \cite{DDS12stoc}:

\pnote{I'm sorry if I'm wrong, but I think that this is equivalent to what was here before, and is a little simpler and more
interpretable.{ \bf{\orange{Rocco:}}} This looks fine to me.} \rnote{Update: It looks like thanks to Anindya's kernel kleverness in the $k=3$ case, now we do not use Fact~\ref{fact:estimate-PBD-variance} at all.  Shall we get rid of it?}
\begin{fact} \label{fact:estimate-PBD-variance}
Let $\bS$ be any PBD$_N.$  There exists an algorithm with the following properties:  given access to independent draws from $\bS$, it produces an estimate $\hat{\sigma}^2$ for $\sigma^2 = \Var[\bS]$  such that, if $\sigma \geq 1$,  with probability at least $1-\delta$:
\[
|\sigma^2 - \hat{\sigma}^2| \le \epsilon \cdot \sigma^2.
\]
The algorithm uses $O(\log(1/\delta)/\epsilon^2)$ samples and runs in time $O(\log(1/\delta)/\epsilon^2).$
\end{fact}
}

\medskip \noindent {\bf Learning distributions with small essential support.}  We recall the following folklore result, 
which says that distributions over a small essential support can be learned efficiently:
\ignore{\pnote{Not explicitly mentioning the semi-agnostic model here.  } \rnote{I think that's fine -- the statement of the fact as-is implies that it works semi-agnostically}
}



\begin{fact} \label{fact:learn-sparse-ess-support}
There is an algorithm $A$ with the following performance guarantee:  $A$ is given a positive integer $s$, an accuracy parameter 
$\eps$, a confidence parameter $\delta$, and access to i.i.d.\ draws from an unknown distribution $\bP$ over $\Z$ that is promised to be
$\eps$-essentially supported on some set $S$ with $|S| = s.$  Algorithm $A$ makes $m=\poly(s, 1/\eps, \log(1/\delta))$ draws from $\bP$, runs for time $\poly(s,1/\eps,\log(1/\delta))$ and with probability at least $1-\delta$ outputs a hypothesis distribution $\tilde{\bP}$ such that $\dtv(\bP,\tilde{\bP}) \leq 2\eps.$
\end{fact}

(The algorithm of Fact \ref{fact:learn-sparse-ess-support} simply
returns the empirical distribution of its $m$ draws from $\bP.$) Note
that by Fact \ref{fact:learn-sparse-ess-support}, if $\bS$ is a
sum of $N < \poly(1/\eps)$ integer random variables
then there is a
$\poly(1/\eps)$-time, $\poly(1/\eps)$-sample algorithm for learning
$\bS$, simply because the support of $\bS$ is contained in a set of
size $\poly(1/\eps)$.  Thus in the analysis of our algorithm for $k=3$
we can (and do) assume that $N$ is larger than any fixed
$\poly(1/\eps)$ that arises in our analysis.


\subsubsection{Hypothesis selection and ``guessing''.} \label{sec:hyp-test}

To streamline our presentation as much as possible, many of the learning algorithms that we present are described as ``making guesses'' for different values at various points in their execution.  For each such algorithm our analysis will establish that with very high probability there is a ``correct outcome'' for the guesses which, if it is achieved (guessed), results in an $\eps$-accurate hypothesis.  This leads to a situation in which there are multiple hypothesis distributions (one for each possible outcome for the guesses that the algorithm makes), one of which has high accuracy, and the overall learning algorithm must output (with high probability) a high-accuracy hypothesis.  Such situations have been studied by a range of authors (see e.g. \cite{Yatracos85,DaskalakisKamath14,AJOS14,DDS12stoc,DDS15}) and a number of different procedures are known which can do this.  For concreteness we recall one such result, Proposition 6 from \cite{DDS15}:

\begin{proposition} \label{prop:log-cover-size}
Let $\bD$ be a distribution over a finite set $W$
and let $\calD_\eps = \{ \bD_j\}_{j=1}^M$ be a collection of $M$ hypothesis distributions
over $W$ with the property that there exists $i \in [M]$ such that
$\dtv(\bD,\bD_i) \leq \eps$.
There is an algorithm~$\mathrm{Select}^{\bD}$ which is given
$\eps$ and a confidence parameter $\delta$, and is provided
with access to
(i) a source of i.i.d. draws from $\bD$ and from $\bD_i$, for all $i \in [M]$; and
(ii) an ``evaluation oracle'' $\eval_{\bD_i}$,
for each $i \in [M]$, which, on input $w \in W$, deterministically outputs
the value $\bD_i(w).$  
The $\mathrm{Select}^{\bD}$ algorithm has the following behavior:
It makes
$m = O\left( (1/ \eps^{2}) \cdot (\log M + \log(1/\delta)) \right)
$ draws from $\bD$ and {from} each $\bD_i$, $i \in [M]$,
and $O(m)$ calls to each oracle $\eval_{\bD_i}$, $i \in [M]$.  It runs in time $\poly(m,M)$ (counting each call to an $\eval_{\bD_i}$ oracle and draw from a $\bD_i$ distribution as unit time),
and with probability $1-\delta$ it outputs an index $i^{\star} \in [M]$ that satisfies $\dtv(\bD,\bD_{i^{\star}}) \leq 6\eps.$
\end{proposition}

We shall apply Proposition \ref{prop:log-cover-size} via the following simple corollary (the algorithm
$A'$ described below works simply by enumerating over all possible outcomes of all the guesses and then running the $\mathrm{Select}^{\bD}$ procedure of Proposition \ref{prop:log-cover-size}):

\begin{corollary} \label{cor:guess} Suppose that an algorithm $A$ for learning
an unknown distribution $\bD$ works in the following way: (i)  it ``makes guesses'' in such a way that there are a total of $M$ possible different vectors of outcomes for all the guesses; (ii) for each vector of outcomes for the guesses, it makes $m$ draws from $\bD$ and runs in time $T$; (iii) with probability at least $1-\delta$, at least one vector of outcomes for the guesses results in a hypothesis $\tilde{\bD}$ such that $\dtv(\bD,\tilde{\bD}) \leq \eps,$ and (iv) for each hypothesis distribution $\bD'$ corresponding to a particular vector of outcomes for the guesses, $A$ can simulate a random draw from $\bD'$ in time $T'$ and can simulate a call to the evaluation oracle $\eval_{\bD'}$ in time $T'$.
Then there is an algorithm $A'$  that makes $m + O\left( (1/ \eps^{2}) \cdot (\log M + \log(1/\delta)) \right)$ draws from $\bD$; runs in time $O(T M) + \poly(m,M,T')$; and with probability at least $1-2\delta$ outputs a hypothesis distribution $\tilde{\bD}$ such that $\dtv(\bD,\tilde{\bD}) \leq 6 \eps.$
\end{corollary}

We will often implicitly apply Corollary~\ref{cor:guess} by indicating a series of guesses and specifying the possible outcomes for them.  It will always be easy to see that the space of all possible vectors of outcomes for all the guesses can be
enumerated in the required time.  In Appendix~\ref{sec:evaluation} we discuss the specific form of the hypothesis distributions that our algorithm produces and show that the time required to sample from or evaluate any such hypothesis is not too high (at most $1/\eps^{2^{\poly(k)}}$ 
when $|\supportset| = 3$,
hence negligible given our claimed running times).

\ignore{
 Often it will also be clear that we can evaluate probabilities for each resulting hypothesis sufficiently efficiently, 
for example when the
hypothesis has small support, or when it is a mixture of members of parameterized classes of distributions such as translated Poissons or discretized Gaussians.
When our guesses involve PBDs, they will have high enough variance that Lemma~\ref{lem:translated Poisson approximation} will imply that approximation by translated Poissons can be used for $\beta$-approximate evaluation.
}

\ignore{
\subsubsection{Learning and approximation}
Many of our arguments will have the following high-level structure:  to learn a target distribution $\bD$ belonging to a class $\calC$ we argue that every distribution in $\calC$ can be approximated by a distribution in another class $\calC'$, and then give an algorithm for learning distributions in $\calC'$.  This approach succeeds if the quality of the approximation compares favorably with the sample complexity of learning $\calC'$; the following observation formalizes the argument. \rnote{Here too I think the kernel approach to covering all $k=3$ cases means that now we do not use Observation~\ref{obs:approx-learn} -- I don't see any references to it in the file.  Shall we remove this whole subsubsection?}

\begin{observation} \label{obs:approx-learn}
Let $\calC$ be a class of distributions, and suppose that given $\gamma > 0$, every distribution in $\calC$ is $\gamma$-close to some distribution in a class $\calC'_{\gamma}.$  Given $\gamma>0$ let $A$ be an algorithm which, given parameters $\eps',\delta',\gamma$, learns $\calC'_{\gamma}$ to accuracy $\eps'$ with confidence $1-\delta'$ using $m(\eps',\delta',\gamma)$ samples and running in time $T(\eps',\delta',\gamma).$  

Suppose that $\eps,\delta,\gamma$ satisfy
\begin{equation}
m(\eps/2,\delta/2,\gamma) \cdot \gamma \leq \delta/2 \quad \quad \text{and} \quad \quad
\gamma < \eps/2. \label{eq:its-cool}
\end{equation}
Then there is an algorithm that learns $\calC$ to accuracy $\eps$ and confidence $\delta$ using $m(\eps/2,\delta/2,\gamma)$ samples and running in time $T(\eps/2,\delta/2,\gamma)$.
\end{observation}

\begin{proof}
The algorithm simply runs $A$ with parameters $\eps/2,\delta/2,\gamma$.  To see that it works, let $\bD \in \calC$ be the target distribution and $\bD' \in \calC'_{\gamma}$ be such that $\dtv(\bD,\bD') \leq \gamma.$  We observe that since $A$ uses $m(\eps/2,\delta/2,\gamma)$ samples and 
(by Lemma~\ref{l:coupling}) 
we may assume that each draw from the target distribution $\bD \in \calC$ is with probability $1-\gamma$ distributed exactly as if it were drawn from $\bD'$, a union bound over (\ref{eq:its-cool}) gives that with failure probability at most $m(\eps/2,\delta/2,\gamma) \cdot \gamma < \delta/2$, the $m(\eps/2,\delta/2,\gamma)$ samples which $A$ uses are distributed exactly as if they were drawn from $\bD' \in \calC'_\gamma$.  Hence with overall probability at least $1-\delta$ algorithm $A$ constructs a hypothesis distribution $\bH$ that has $\dtv(\bH,\bD') \leq \eps/2$ and hence, by (\ref{eq:its-cool}), $\dtv(\bH,\bD) \leq \eps.$
\end{proof}

}

\subsection{Small error}  We freely assume throughout that the desired error parameter $\eps$ is at most some sufficiently small absolute constant value.  \ignore{At various points in our upper bounds for ease of presentation we prove that our algorithms succeed in learning the unknown distribution to accuracy $\eps^c$ for some absolute constant $c>0$; since all of our claimed sample complexity and running time bounds are stated only up to $\poly(\cdot)$ factors, this clearly suffices to establish our claimed results.}

\subsection{Fano's inequality and lower bounds on distribution learning.} \label{sec:fano}

A useful tool for our lower bounds is Fano's inequality, or more precisely, the following extension of it
given by Ibragimov and Khasminskii \cite{IK79} and Assouad and Birge \cite{AB83}:

\begin{theorem}[Generalization of Fano's Inequality.] \label{fano}
Let $\bP_1,\dots,\bP_{t+1}$ be a collection of $t+1$ distributions such that for any $i \neq j \in [t+1]$, we
have (i) $\dtv(\bP_i,\bP_j) \geq \alpha/2$, and (ii) $D_{KL}(\bP_i || \bP_j) \leq \beta$, where $D_{KL}$ denotes Kullback-Leibler divergence.  
Let $A$ be a learning algorithm which is given samples from an unknown distribution $\bP$ which is promised to be one of $\bP_1,\dots,\bP_{t+1}$ and which outputs an index $i \in [t+1]$ specifying a distribution $\bP_i$. Then, to achieve expected error 
$\E[\dtv(\bP,\bP_i)] \leq \alpha/4$ (where the expectation is over the random samples from $\bP$), algorithm $A$ must have sample complexity
$\Omega \left(\frac{\ln t}{\beta}\right).$
\end{theorem}

\section{Tools for kernel-based learning} \label{sec:kernel}

At the core of our actual learning algorithm is the well-known
technique of learning via the ``kernel method'' (see \cite{DL:01}). In this
section we set up some necessary machinery for applying this technique in our context.  

The goal of this section is ultimately to establish Lemma~\ref{lem:siirv-kernel}, which we will use later in our main learning algorithm.   Definition~\ref{def:kernel} and Lemma~\ref{lem:siirv-kernel} together form a ``self-contained take-away'' from this section.

We begin with the following important definition.
\begin{definition} \label{def:kernel}
Let $\bY,\bZ$ be two distributions supported on $\Z.$  We say that $\bY$ is \emph{$(\eps,\delta)$-kernel learnable from $T=T(\eps,\delta)$ samples using $\bZ$} if the following holds:   Let $\hat{Y}=\{y_1,\ldots, y_{{T_1}}\}$ be a multiset of $T_1 \geq T$ i.i.d. samples drawn from $\bY$ and let $\bU_{\hat{Y}}$ be the uniform distribution over $\hat{Y}.$ Then with probability $1-\delta$ (over the outcome of $\hat{Y}$) it is the case that $\dtv(\bU_{\hat{Y}} + \bZ,\bY) \leq \epsilon$.
\end{definition}
Intuitively, the definition says that convolving the empirical distribution $\bU_{\hat{Y}}$ with $\bZ$ gives a distribution which is close to $\bY$ in total variation distance. Note that once $T_1$ is sufficiently large, $\bU_{\hat{Y}}$ is close to $\bY$ in Kolmogorov distance by the DKW inequality.  Thus, convolving with $\bZ$ smoothens $\bU_{\hat{Y}}$ . 

The next lemma shows that if $\bY$ is $(\eps,\delta)$-kernel learnable, then a mixtures of shifts of $\bY$ is also $(\eps,\delta)$-kernel learnable with comparable parameters (provided the number of components in the mixture is not too large).

\begin{lemma} \label{lem:mix}
Let $\bY$ be $(\eps,\delta)$-kernel learnable using $\bZ$ from $T{(\eps,\delta)}$ samples.  If $\bX$ is a mixture (with arbitrary mixing weights) of distributions $c_1 + \bY,\dots,c_k + \bY$ for some integers $c_1,\dots,c_k$, then $\bX$ is
$(7\eps,2\delta)$-kernel learnable from ${T'}$ samples using $\bZ$,
{provided that $T' \geq \max\left\{{\frac {kT(\eps,\delta/k)}{\eps}}, C \cdot {\frac {k^2 \log(k/\delta)}{\eps^2}}\right\}$}.
\end{lemma}

\begin{proof}
Let $\pi_j$ denote the weight of distribution $c_j + \bY$ in the mixture $\bX$.  We view the draw of a sample point from $\bX$ as a two stage process, where in the first stage an index $1\le j \le k$ is chosen with probability $\pi_j$ and in the second stage a random draw is made from the distribution $c_j +\bY$. 

Consider a draw of $T'$ independent samples $x_1,\dots,x_{T'}$ from $\bX$. In the draw of $x_i$, let the index chosen in the first stage be denoted $j_i$ (note that $1 \le j_i \le k$). For $j \in [k]$ define 
$$
S_j = \{1 \le i \leq T': j_i = j \}.
$$

{The idea behind Lemma \ref{lem:mix} is simple.  Those $j$ such that $\pi_j$ is small will have $|S_j|$ small and will not contribute much to the error.  Those $j$ such that $\pi_j$ is large will have $|S_j|/T'$ very close to $\pi_j$ so their cumulative contribution to the total error will also be small since each such $\bU_{\{x_i: i \in S_j\}}+\bZ$ is very close to the corresponding $c_j + \bY$.  We now provide details.}

{Since $T' \geq O\left({\frac {k^2 \log(k/\delta)}{\eps^2}}\right)$, a simple Chernoff bound and union bound over all $j \in [k]$} gives that
\begin{equation} \label{eq:empirical-ok}
\left| \frac{|S_j|}{T'} - \pi_j \right| \le  \epsilon/k \quad \text{for all~}j \in [k]
\end{equation}
with probability at least $1-\delta$. For the rest of the analysis we assume that indeed (\ref{eq:empirical-ok}) holds.  We observe that even after conditioning on (\ref{eq:empirical-ok}) and on the outcome of $j_1,\dots,j_{T'}$, it is the case that for each $i \in [T']$ the value $x_i$ is drawn independently from $c_{j_i} + \bY.$

Let $\mathrm{Low}$ denote the set {$\{1 \le j \le k: T' \cdot (\pi_j - \eps/k) \leq T(\eps,\delta/k)\}$, so each $j \notin \mathrm{Low}$ satisfies $T'\cdot (\pi_j - \eps/k) \geq T(\eps,\delta/k)$.}
Fix any $j \not \in \mathrm{Low}$.  From (\ref{eq:empirical-ok}) and the definition of $\mathrm{Low}$ we have that  $|S_j| \geq T' \cdot (\pi_j-\eps/k) \geq T(\eps,\delta/k)$, and since  $c_j+\bY$  is $(\eps,{\delta/k})$-kernel learnable from $T(\eps,{\delta/k})$ samples using $\bZ$, it follows that {with probability at least $1-\delta/k$ we have}
$$
\dtv \left( \bU_{\{x_i: i \in S_j\}} + \bZ, c_j +\bY\right)  \le \epsilon,
$$
and thus
\begin{align}
\nonumber
\sum_{z \in \Z} \left| \frac{|S_j|}{T'} \Pr[\bU_{\{x_i: i \in S_j\}} + \bZ = z]
              - \pi_j \Pr[\bY = z] \right| \\
\label{eq:ell1}
\leq \bigg|\frac{|S_j|}{T'} - \pi_j\bigg| + 
       {\max\left\{\frac{|S_j|}{T'}, \pi_j \right\}} \cdot \frac{\epsilon}{2}.
\end{align}
By a union bound, with probability at least $1-\delta$ the bound (\ref{eq:ell1}) holds for all $j \notin \mathrm{Low}$.
For $j \in \mathrm{Low}$, we trivially have 
$$
\sum_{z \in \Z} \left| \frac{|S_j|}{T'} \Pr[\bU_{\{x_i: i \in S_j\}} + \bZ = z]
              - \pi_j \Pr[Y = z] \right| 
\le \frac{|S_j|}{T'}  + \pi_j  \le \bigg|\frac{|S_j|}{T'} - \pi_j\bigg| + 2 \cdot \pi_j. 
$$
Next, note that 
\[
\sum_{j \in \mathrm{Low}} \pi_j \le \sum_{j \in \mathrm{Low}} \left(\frac{T(\eps,\delta/k)}{T'} + \eps/k\right)
\le
\sum_{j \in \mathrm{Low}} (\eps/k + \eps/k) \leq 2\eps.
\] 
Thus, we obtain that 
\begin{align*}
& \sum_{z \in \Z} \left| 
       \sum_{j=1}^k \frac{|S_j|}{T'} \Pr[\bU_{\{x_i: i \in S_j\}} + \bZ = z]
              - \sum_{j=1}^k \pi_j \Pr[Y = z] \right| \\
& \leq \sum_{j=1}^k \sum_{z \in \Z} 
        \left| \frac{|S_j|}{T'} \Pr[\bU_{\{x_i: i \in S_j\}} + \bZ = z]
              - \sum_{j=1}^k \pi_j \Pr[Y = z] \right| \\
& \le \sum_{j=1}^k \bigg|\frac{|S_j|}{T'} - \pi_j\bigg| + \sum_{j \not \in \mathrm{Low}}{\max\left\{\frac{|S_j|}{T'}, \pi_j\right\}} \cdot \frac{\epsilon}{2} + 2 \sum_{j \in \mathrm{Low}} \pi_j
\le 7 \epsilon.
\end{align*}

As $\bX$ is obtained by mixing 
$c_1 + \bY,...,c_k + \bY$ with
weights $\pi_1,...,\pi_k$ and 
$\bU_{x_1, \ldots, x_{T'}}$ is obtained by mixing
$\bU_{\{x_i: i \in S_1\}},...,\bU_{\{x_i: i \in S_k\}}$
with weights $\frac{|S_1|}{T'},...,\frac{|S_k|}{T'}$, the lemma is proved.
\end{proof}

The next lemma is a formal statement of the well-known robustness of kernel learning; roughly speaking, it says that if $\bX$ is kernel learnable using $\bZ$ then any $\bX'$ which is close to $\bX$ is likewise kernel learnable using $\bZ$.
\begin{lemma}~\label{lem:kernel-close}
Let $\bX$ be $(\eps,\delta)$-kernel learnable using $\bZ$ from $T(\eps,\delta)$ samples, and suppose that $0 < \dtv(\bX,\bX') = \kappa < 1.$  {If $T_0>\max\{T(\eps,\delta),C \cdot {\frac {\log(1/\delta)}{\eps^2}}\}$, then $\bX'$ is $(2\eps + 2 \kappa,2\delta)$-kernel learnable from $T_0$ samples using $\bZ.$}
\end{lemma}
\begin{proof}
We establish some useful notation:  let $\bX_{\mathrm{common}}$ denote the distribution defined by
\[
\Pr[\bX_{\mathrm{common}}=i] = {\frac {\min\{\Pr[\bX=i],\Pr[\bX'=i]\}}{\sum_i \min\{\Pr[\bX=i],\Pr[\bX'=i]\}}},
\]
let $\bX_{\mathrm{residual}}$ denote the distribution defined by
\[
\Pr[\bX_{\mathrm{residual}}=i] = {\frac {\Pr[\bX=i] - \min\{\Pr[\bX=i],\Pr[\bX'=i]\}}{\sum_i (\Pr[\bX=i] - \min\{\Pr[\bX=i]\Pr[\bX'=i]\})}},
\]
and likewise let 
$\bX'_{\mathrm{residual}}$ denote the distribution defined by
\[
\Pr[\bX'_{\mathrm{residual}}=i] = {\frac {\Pr[\bX'=i] - \min\{\Pr[\bX=i],\Pr[\bX'=i]\}}{\sum_i (\Pr[\bX'=i] - \min\{\Pr[\bX=i]\Pr[\bX'=i]\})}}.
\]

A draw from $\bX$ (from $\bX'$
  respectively) may be obtained as follows: draw from
  $\bX_{\mathrm{common}}$ with probability $1-\kappa$ and from
  $\bX_{\mathrm{residual}}$ (from $\bX'_{\mathrm{residual}}$
  respectively) with the remaining $\kappa$ probability.
  To see this, note that if $\tilde{\bX}$ is a random variable generated
  according to this two-stage process and $\bC \in \{ 0, 1\}$ is an
  indicator variable for whether the draw was from
  $\bX_{\mathrm{common}}$, then, 
  since $\kappa = 1 - \sum_i \min\{\Pr[\bX=i],\Pr[\bX'=i]\}$, we have
\begin{align*}
\Pr[\tilde{\bX} = i] & = \Pr[\tilde{\bX} = i \wedge \bC = 1] + \Pr[\tilde{\bX} = i \wedge \bC = 0] \\
& = \frac{\min\{\Pr[\bX=i],\Pr[\bX'=i]\}}{1 - \kappa} \times (1 - \kappa)
   + \frac{\Pr[\bX'=i] - \min\{\Pr[\bX=i],\Pr[\bX'=i]\}}{1 - (1 - \kappa)}
     \times \kappa \\
& = \Pr[\bX=i].
\end{align*}

We consider the following coupling of $(\bX,\bX')$:  to make a draw of $(x,x')$ from the coupled joint distribution $(\bX,\bX')$, draw $x_{\mathrm{common}}$ from $\bX_{\mathrm{common}}$, draw $x_{\mathrm{residual}}$ from $\bX_{\mathrm{residual}}$, and draw $x'_{\mathrm{residual}}$ from $\bX'_{\mathrm{residual}}.$  With probability $1-\kappa$ output $(x_{\mathrm{common}},x_{\mathrm{common}})$ and with the remaining $\kappa$ probability output  $(x_{\mathrm{residual}},x'_{\mathrm{residual}}).$  

Let $((x_1,x'_1),\dots,(x_{T_0},x'_{T_0}))$ be a sample of $T_0$ pairs each of which is independently drawn from the coupling of $(\bX,\bX')$ described above.  Let $\hat{X}=(x_1,\dots,x_{T_0})$ and $\hat{X}'=(x'_1,\dots,x'_{T_0})$ and observe that $\hat{X}$ is a sample of $T_0$ i.i.d. draws from $\bX$ and similarly for $\hat{X}'$.  We have
\begin{align}
\dtv(\bU_{\hat{X}'}+\bZ,\bX') &\leq
\dtv(\bU_{\hat{X}'}+\bZ,\bU_{\hat{X}} + \bZ) + \dtv(\bU_{\hat{X}}+\bZ,\bX) + \dtv(\bX,\bX') \nonumber \\
&\leq \dtv(\bU_{\hat{X}'},\bU_{\hat{X}}) + \eps + \kappa \quad \quad \text{(by the data processing inequality for $\ell_1$)},
\nonumber
\end{align}
where the second inequality holds with probability $1-\delta$ over the draw of $\hat{X}$ since $T_0 \geq T(\eps,\delta)$.  A simple Chernoff bound tells us that with probability at least $1-\delta$, the fraction of the $T_0 \geq C \cdot {\frac {\log(1/\delta)}{\eps^2}}$ pairs that are of the form $(x_{\mathrm{residual}},x'_{\mathrm{residual}})$ is at most $\kappa + \eps$.  Given that this happens we have $ \dtv(\bU_{\hat{X}'},\bU_{\hat{X}}) \leq \kappa + \eps$, and the lemma is proved.
\ignore{

Since $\dtv(X,X') = \kappa$, we can view drawing $T$ samples from $X'$ in the following multistage process. Fix a distribution $D$. To draw $T$ samples from $X'$, we first draw $T$ samples from $X$. Call these $x_1, \ldots, x_T$. Subsequently, we draw $T$ independent $\{0,1\}$ valued random (call them $z_1, \ldots, z_T$) where $\mathbf{E}[z_i]=\kappa$. For each $i: z_i=1$, we replace $x_i$ by an independent sample $y_i$ from $D$. Call the final set of samples as $x'_1, \ldots, x'_T$. With this, we have 
\begin{eqnarray*}
d_{\ell_1} (U_{x'_1,\ldots, x'_t} + Z,X') \le d_{\ell_1} (U_{x'_1,\ldots, x'_t} + Z,X) + (X,X') \le \epsilon + d_{\ell_1} (U_{x'_1,\ldots, x'_t} + Z,X).
\end{eqnarray*} 
Thus, we now focus on bounding $d_{\ell_1} (U_{x'_1,\ldots, x'_t} + Z,X)$. To bound this, notice that with probability $(1-\delta)$, 
$\dtv(U_{x_1,\ldots, x_t} + Z,X) \le \epsilon$. Call this event $\mathsf{Good}$. We will bound $d_{\ell_1} (U_{x'_1,\ldots, x'_t} + Z,X)$ assuming the event $\mathsf{Good}$ happened. 
\begin{eqnarray*}
\dtv(U_{x'_1, \ldots, x'_t} + Z,X) &\le& \dtv(U_{x'_1, \ldots, x'_t} + Z,U_{x_1, \ldots, x_t} +Z) + \dtv(U_{x_1, \ldots, x_t} +Z, X) \\
&\leq& \dtv(U_{x'_1, \ldots, x'_t} + Z,U_{x_1, \ldots, x_t} +Z) +\epsilon \le \dtv(U_{x'_1, \ldots, x'_t} ,U_{x_1, \ldots, x_t}) +\epsilon.
\end{eqnarray*}
It is easy to see that with probability $1-\kappa$, $\dtv (U_{x'_1, \ldots, x'_t} ,U_{x_1, \ldots, x_t}) \le 2 \kappa$. Putting everything together, this concludes the proof. }
\end{proof}

To prove the next lemma (Lemma~\ref{lem:kernel-learn1} below) we will need a multidimensional generalization of the usual coupling argument used to prove the correctness of the kernel method.  This is given by the following proposition:

\begin{proposition}~\label{prop:multidim-kernel}
For all $1 \le j \le k$, let $a_j, b_j \in \mathbb{Z}$ 
with $b_j \geq 1$
and let $\mathcal{B}$ be the subset of $\Z^k$ given by $\mathcal{B}= [a_1, a_1+b_1] \times  \ldots \times [a_k, a_k+b_k]$. 
Let $\bX, \bY$ be random variables supported on $\mathbb{Z}^k$ such that 
$\Pr[\bX \not \in \mathcal{B}], \Pr[\bY \not \in \mathcal{B}] \le \delta$. Let 
$\bZ$ be a random variable supported on $\mathbb{Z}^k$ such that for all $1 \le j \le k$, $\dtv(\bZ, \bZ+\mathbf{e}_j) 
\le \beta_j$.
If $\dk(\bX,\bY) \le \lambda$,
$\beta_j \leq \rho/b_j$ for all $j$ (where $\rho \geq 1$),
and $\bZ$ is independent of $\bX$ and $\bY$, then $$\dtv(\bX+ \bZ, \bY + \bZ) \le  2 \delta + 
O \left( 4^k \lambda^{\frac{1}{k+1}} \rho^{1 - \frac{1}{k+1}} \right).
$$
\end{proposition}
\begin{proof}
\newcommand{\cB}{\mathcal{B}} 
Let
$d_1 \leq b_1, \ldots, d_k \leq b_k$ be positive integers that we will fix later.
Divide the box $\mathcal{B}$ into boxes of size at most
$d_1 \times \ldots \times d_k$ by dividing each $[a_i,a_i + b_i]$ into intervals of
size $d_i$ (except possibly the last interval which may be smaller)\ignore{ and then taking the cartesian products of all pairs of intervals}.  Let
$\mathcal{S}$ denote the resulting set of $k$-dimensional boxes induced by these intervals, and note that the number of boxes in $\mathcal{S}$ is $\ell_1 \times \ldots \times \ell_k$ where $\ell_j = \lceil b_j/d_j \rceil$.

Let $\mu_{\bX}$ and $\mu_{\bY}$ be the probability measures associated
with $\bX$ and $\bY$, and let $\mu_{\bX, \cB}$ and $\mu_{\bY, \cB}$ be
the restrictions of $\mu_{\bX}$ and $\mu_{\bY}$ to the box $\cB$ (so $\mu_{\bX}$ and $\mu_{\bY}$ assign value zero to any point not in $\cB$).  
For a box $S \in \mathcal{S}$, let
$\mu_{\bX,S}$ denote the restriction of $\mu_{\bX}$ to
$S$. Let $x_S = \Pr[\bX \in S]$ and $y_S = \Pr[\bY \in S]$. Let $w_S
= \min \{x_S, y_S\}$.  Let $\bX_S$ and $\bY_S$ be the random variables
obtained by conditioning $\bX$ and $\bY$ on $S$, and
$\mu_{\bX_S}$ and $\mu_{\bY_S}$ be their measures.
Note that
$\mu_{\bX,S} = x_S \cdot \mu_{\bX_S}$ and $\mu_{\bY,S} = y_S \cdot \mu_{\bY_S}$. 
With this notation in place, using $f * g$ to denote the convolution of the measures $f$ and $g$, we now have
\begin{eqnarray*}
\dtv(\bX+\bZ, \bY+\bZ) & = & \frac{1}{2} \ell_1(\mu_{\bX + \bZ}, \mu_{\bY + \bZ}) \\
    & = & \frac{1}{2} \ell_1(\mu_{\bX} * \mu_{\bZ}, \mu_{\bY} * \mu_{\bZ}) \\
   &\le& \Pr[\bX \not \in \mathcal{B}] + \Pr[\bY \not \in \mathcal{B}] 
    + \frac{1}{2} \ell_1(\mu_{\bX, \cB} * \mu_{\bZ}, \mu_{\bY,\cB} * \mu_{\bZ}) \\
&\leq& 2 \delta 
   + \frac{1}{2} \sum_{S \in \mathcal{S}} \ell_1(\mu_{\bX, S} * \mu_{\bZ}, \mu_{\bY,S} * \mu_{\bZ}) \\
&\leq& 2 \delta + \frac{1}{2} \sum_{S \in \mathcal{S}} \ell_1(x_S \mu_{\bX_S} * \mu_{\bZ}, y_S \mu_{\bY_S} * \mu_{\bZ}) \\
&\leq& 2 \delta + \frac{1}{2} \sum_{S \in \mathcal{S}} \ell_1(w_S \mu_{\bX_S} * \mu_{\bZ}, w_S \mu_{\bY_S} * \mu_{\bZ}) 
                + \sum_{S \in \mathcal{S}} |x_S - y_S|  \\
&\leq& 2 \delta + \frac{1}{2} \sum_{S \in \mathcal{S}} w_S \ell_1(\mu_{\bX_S} * \mu_{\bZ}, \mu_{\bY_S} * \mu_{\bZ}) 
                + | \mathcal{S} | \cdot 
                   2^k
                      \lambda \\
&\leq& 2 \delta + \sum_{S \in \mathcal{S}} w_S \dtv(\bX_S  + \bZ, \bY_S + \bZ)
                + | \mathcal{S} | \cdot 2^k \lambda. \\
\end{eqnarray*}
Here the second to last inequality uses the fact that the definition of $\dk(\bX, \bY)$
gives $\sup |x_S - y_S| \le 2^k \lambda$. 
Next, notice that since $\dtv(\bZ, \bZ+ \mathbf{e}_j) \le \rho/b_j$ and each box in $\mathcal{S}$ has size at most
$d_1 \times \ldots \times d_k$, we get that 
$\dtv(\bX_S + \bZ, \bY_S + \bZ) \le \sum_{i=1}^k \beta_i (d_i - 1)$. 
Thus, using that $|\mathcal{S}| =  \prod_{j=1}^k \lceil b_j/d_j \rceil$ and $\sum_{S \in \mathcal{S}} w_S \le 1$, we have
\begin{align*}
\dtv(\bX+\bZ, \bY+\bZ) 
 & \leq  2 \delta + \sum_{S \in \mathcal{S}} w_S \cdot \big(\sum_{i=1}^k \beta_i (d_i-1) \big) 
      + 4^k \lambda \cdot\prod_{j=1}^k (b_j/d_j). \\
\end{align*}
Optimizing the parameters $d_1, \ldots, d_k$, we set
each
$
d_i  
  = 
\left\lceil \left(\frac{\lambda}{\rho}\right)^{\frac{1}{k+1}} b_i \right\rceil
$
which yields
\begin{align*}
\dtv(\bX+\bZ, \bY+\bZ) 
& \leq  2 \delta + 
  (k + 4^k) \lambda^{\frac{1}{k+1}} \rho^{1 - \frac{1}{k+1}}. \qedhere
\end{align*}
\end{proof}

Now we can prove Lemma~\ref{lem:kernel-learn1}, which
we will use to prove
that a weighted sum of high-variance PBDs
is kernel-learnable for appropriately chosen smoothening distributions.
 
\begin{lemma}~\label{lem:kernel-learn1}
Let independent random variables $\bX_1, \ldots, \bX_k$ over $\mathbb{Z}$,
and $\rho \geq 1$, be such that 
\begin{enumerate} 
 \item For $1 \le j \le k$, there exist $a_j, b_j \in \mathbb{Z}$, 
$ \delta_j \ge 0$ such that $\Pr[\bX_j \not \in [a_j , a_j + b_j]] \le \delta_j$, 
\item For all $1 \le j \le k$,  
      $\dshift{1} (\bX_j) \le \beta_j \leq \rho/b_j$.
\end{enumerate} 
Let $\bY = \sum_{j=1}^k p_j \cdot \bX_j$ for some integers $p_1,\dots,p_k$.  Let $\bZ_j$ be the uniform distribution on the set $\mathbb{Z} \cap[-c_j, c_j]$
where\pnote{We might have considered setting $c_j = \frac{\Theta(\epsilon)}{k \cdot \beta_j}$, analogously to before, but this could be
inconsistent with $c_j \leq b_j$.  I don't see
where we use $c_j \leq b_j$ though.  Even if we
do set $c_j = \frac{\Theta(\epsilon)}{k \cdot \beta_j}$, it
seems that we will still need $\beta_i \leq \epsilon^2/k$,
will seems like it may require significant downstream
changes.
}
$c_j \in \mathbb{Z}$ satisfies $c_j=\frac{\Theta(\epsilon) b_j}{k \cdot \rho}$ and
$1 \leq c_j \leq b_j$ 
and $\bZ_1,\dots,\bZ_k$ are mutually independent and independent of $\bX_1, \ldots, \bX_k$.
Define $\bZ = \sum_{j=1}^k p_j \cdot \bZ_j$.  Then,
 $\bY$ is $(\epsilon + 4(\delta_1 + \ldots + \delta_k), \delta)$-kernel learnable using $\bZ$ from $ T= 
 \frac{\exp(O(k^2))}{\epsilon^{O(k)}} \cdot 
\rho^{O(k)}
\cdot \log (1/\delta) +\log ({\frac {4k} {\delta}}) \cdot \max_{j} 1/\delta_j^2$ samples.
\end{lemma}

\begin{proof} 
We first observe that 
\begin{eqnarray}
\dtv(\bY + \bZ , \bY) &\le& \sum_{j=1}^k \dtv(p_j \cdot \bX_j + p_j \cdot \bZ_j, p_j \cdot \bX_j) \nonumber \\
&=& \sum_{j=1}^k \dtv(  \bX_j +   \bZ_j, \bX_j) \le \sum_{j=1}^k \frac{\rho c_j}{b_j} = \Theta(\eps) \label{eq:shiftbx1} 
\end{eqnarray}
where the last inequality uses the fact that $\bZ_j$ is supported on the interval $[-c_j, c_j]$ and $\dshift{1}(\bX_j) \le \frac{\rho}{b_j}$. Now, consider a two-stage sampling process for an element $y \leftarrow \bY$: For $1 \le j \le k$, we sample $x_j^{(y)} \sim \bX_j$ and then output $y = \sum_{j=1}^k p_j \cdot x_j^{(y)}$. Thus, for every sample $y$, we can associate a sample $x^{(y)} = (x_1^{(y)}, \ldots, x_k^{(y)})$. For $y_1, \ldots, y_T \leftarrow \bY$, let $x^{(y_1)}, \ldots, x^{(y_T)}$ denote the corresponding samples from $\mathbb{Z}^k$.   Let $\bU_{\hat{X}}$ denote the uniform distribution over the multiset of $T$ samples $x^{(y_1)}, \ldots, x^{(y_T)}$, and let $\bU_{\hat{Y}}$ denote the uniform distribution on $y_1, \ldots, y_T$.
 Let $\bX_{\mathsf{multi}} = (\bX_1, \bX_2, \ldots , \bX_k) $. 
By Lemma~\ref{lem:DKW-multidim},
 we get that if $T \ge c (k + \log (1/\delta))/\eta^2$ (for a parameter $\eta$ we will fix later), then with probability $1-\delta/2$ we have $\dk(\bX_{\mathsf{multi}}, \bU_{\hat{X}}) \le \eta$; moreover, if $T \geq \log ({\frac {4k} {\delta}}) \cdot \max_{j} 1/\delta_j^2$, then by a Chernoff bound and a union bound we have that $\Pr[(\bU_{\hat{X}})_j \not \in [a_j, a_j+b_j]] \le  2\delta_j$  for $1 \leq j \leq k$ (which we will use later) with probability $1-\delta/2$.  In the rest of the argument we fix such an $\hat{X}$ satisfying these conditions, and show that for the corresponding $\hat{Y}$ we have $ \dtv(\bY  , \bU_{\hat{Y}} + \bZ ) \leq 4 (\delta_1 + \ldots + \delta_k) + \epsilon$, thus establishing kernel learnability of $\bY$ using $\bZ$.
 
Next, we define
$\bZ_{\mathsf{multi}} = (\bZ_1,  \ldots, \bZ_k)$,
with the aim of applying Proposition~\ref{prop:multidim-kernel}. We observe that $\dtv(\bZ_{\mathsf{multi}}, \bZ_{\mathsf{multi}} + \mathbf{e}_j) \le \frac{1}{c_j}$ and as noted above, for $1 \le j \le k$ we have $\Pr[(\bU_{\hat{X}})_j \not \in [a_j, a_j+b_j]] \le  2 \delta_j$.
Define the box $\mathcal{B} = [a_1, a_1 + b_1] \times \ldots \times [a_k, a_k + b_k]$.
Applying Proposition~\ref{prop:multidim-kernel}, 
we get
\begin{align*}
\dtv(\bX_{\mathsf{multi}}+\bZ_{\mathsf{multi}} , \bU_{\hat{X}}+\bZ_{\mathsf{multi}} ) 
 & \le 4 (\delta_1 + \ldots + \delta_k) +  
O\left( 
4^k \eta^{\frac{1}{k+1}} \left( \frac{k \rho}{\Theta(\epsilon)} \right)^{1 - \frac{1}{k+1}}
 \right)\\
\end{align*}
since $\sum_j \beta_j \leq \epsilon$.
Taking an inner product with $\overline{p}=(p_1, \ldots, p_k)$, we get
\begin{eqnarray*}
 \dtv(\bY + \bZ , \bU_{\hat{Y}} + \bZ )&=& \dtv(\langle\overline{p},\bX_{\mathsf{multi}}+\bZ_{\mathsf{multi}} \rangle , \langle \overline{p}, \bU_{\hat{X}}+\bZ_{\mathsf{multi}}\rangle )  \\
&\le& \dtv(\bX_{\mathsf{multi}}+\bZ_{\mathsf{multi}} , \bU_{\hat{X}}+\bZ_{\mathsf{multi}} ) \\
&\leq&  4 (\delta_1 + \ldots + \delta_k) +  
O\left( 
4^k \eta^{\frac{1}{k+1}} \left( \frac{k \rho}{\Theta(\epsilon)} \right)^{1 - \frac{1}{k+1}}
 \right).
\end{eqnarray*}
Combining this with (\ref{eq:shiftbx1}), we get that 
\begin{align} 
 \dtv(\bY  , \bU_{\hat{Y}} + \bZ ) &\leq 4 (\delta_1 + \ldots + \delta_k) 
+ 
O\left(
4^k \eta^{\frac{1}{k+1}} \left( \frac{k \rho}{\Theta(\epsilon)} \right)^{1 - \frac{1}{k+1}}
 \right)
+ \Theta(\eps),
 \label{eq:zz}
 \end{align}
\ignore{
}
Setting $\eta = {\frac {\eps^{2k+1}}{4^{k (k+1)} k^k \rho^k}}$, the condition 
$T \ge c (k + \log (1/\delta))/\eta^2$ from earlier becomes 
$$
T \geq  c (k + \log (1/\delta))/\eta^2 =  \frac{e^{O(k^2)}}{\epsilon^{O(k)}} \cdot \rho^{2k} \cdot \log (1/\delta)
$$
\ignore{
} and we have $ \dtv(\bY  , \bU_{\hat{Y}} + \bZ ) \leq 4 (\delta_1 + \ldots + \delta_k) + \Theta(\epsilon)$, proving the lemma.
\end{proof}

We specialize Lemma~\ref{lem:kernel-learn1} to establish kernel learnability of
weighted sums of signed PBDs
as follows:

\begin{corollary}~\label{corr:kernel-signed}
Let $\bS_1, \ldots, \bS_k$ be independent signed PBDs and let $\bY = \sum_{j=1}^k p_j \cdot \bS_j$. Let $\sigma_j^2 = \Var[\bS_j] =\omega(k^2/\eps^2)$ and let $\bZ_j$ be the uniform distribution on $[-c_j, c_j] \cap \mathbb{Z}$ where $c_j = \Theta(\epsilon \cdot \sigma_j / k)$.  Let $\bZ = \sum_{j=1}^k p_j \cdot \bZ_j$. Then $\bY$ is $(\epsilon, \delta)$-kernel learnable using $\bZ$ from $T = \frac{e^{O(k^2)}}{\epsilon^{O(k)}} \cdot \log(1/\delta)$ samples. 
\end{corollary} 
\begin{proof}
Note that for $1 \le j \le k$, there are integers $a_j$ such that for $b_j = O(\sigma_j \cdot \ln (k / \epsilon))$, by Bernstein's inequality we have $\Pr[ \bS_j \not \in [a_j, a_j + b_j]] \le \epsilon/k$. Also, recall from Fact~\ref{fact:good-shift-invariance}  that $\dshift{1}(\bS_j) = \frac{O(1)}{\sigma_j}$. Since each $c_j$
satisfies $1 \leq c_j \leq b_j$, we may apply
Lemma~\ref{lem:kernel-learn1} and we get that $\bY$ is $(O(\epsilon), \delta)$-kernel learnable using $T = \frac{e^{O(k^2)}}{\epsilon^{O(k)}} \cdot \log(1/\delta)$ samples.
\end{proof}

(It should be noted that while the previous lemma shows that 
a weighted sum of signed
PBDs that have ``large variance'' are kernel learnable, the hypothesis $\bU_{\hat{Y}}+\bZ$ is based on $\bZ$ and thus constructing it requires knowledge of the variances $\sigma_1,\dots,\sigma_j$; thus Lemma~\ref{lem:kernel-learn1} does not immediately yield an efficient learning algorithm when the variances of the underlying PBDs are unknown. We will return to this issue of knowing (or guessing) the variances of the constituent PBDs later.)

Finally, we generalize Corollary~\ref{corr:kernel-signed} to obtain a robust version.
Lemma~\ref{lem:siirv-kernel} will play an important role in our ultimate learning algorithm.

\begin{lemma}~\label{lem:siirv-kernel}
Let $\bS$ be $\kappa$-close to a distribution of the form $\bS' = \bS_{\textsf{offset}} + \sum_{j=1}^K p_j \cdot \bS_j$, where $\bS_{\textsf{offset}},\bS_1,\dots,\bS_K$ are all independent and $\bS_1,\dots,\bS_K$ are signed PBDs. For $a \in [K]$ let $\sigma_a^2 = \Var[\bS_a] =\omega(K^2/\eps^2)$.\ignore{\rnote{Let's go for a self contained statement in this lemma statement, one should be able to read it independently of all the setup in section 5 etc.  So should this be ``$K$'' or ``$k$''?}}  Let $m = |\supp (\bS_{\textsf{offset}})|$ and let $\gamma_1, \ldots, \gamma_K$ be such that for all $1 \le a \leq K$ we have $\sigma_a \le\gamma_a \le 2\sigma_a$.  Let $\bZ_j$ be the uniform distribution on the interval $[-c_j, c_j] \cap \mathbb{Z}$ where $c_j = \Theta(\epsilon \cdot \gamma_j/K)$. Then for $\bZ = \sum_{a=1}^K p_a \cdot \bZ_a$, the distribution $\bS$ is $(O(\epsilon +  \kappa), O(\delta))$-kernel learnable  using $\bZ$ from $ \frac{\exp(O(K^2))}{\epsilon^{O(K)}} \cdot m^2 \cdot \log (m/\delta) $ samples. 
\end{lemma}
\begin{proof}
Applying Corollary~\ref{corr:kernel-signed} and Lemma~\ref{lem:mix}, we first obtain that the distribution $\bS'$ is $(O(\epsilon),O(\delta))$-kernel-learnable using $\bZ$ from $ \frac{\exp(O(K^2))}{\epsilon^{O(K)}} \cdot m^2 \cdot \log (m/\delta) $  samples. Now, applying Lemma~\ref{lem:kernel-close}, we obtain that $\bS$ is 
$(O(\epsilon + \kappa), O(\delta))$-learnable 
using
$\bZ$ from $ \frac{\exp(O(K^2))}{\epsilon^{O(K)}} \cdot m^2 \cdot \log (m/\delta) $  samples.
\end{proof}

\section{Setup for the upper bound argument} \label{sec:setup-upper-bound}

Recall that 
an $\supportset$-sum
is $\bS = \bX_1 + \cdots + \bX_N$ where the $\bX_1,\dots,\bX_N$ distributions are independent
(but not identically distributed) and each $\bX_i$ is supported on the set 
$\supportset = \{a_1,\dots,a_k\}$ where 
$\{a_1,\dots,a_k \}\subset \Z_{\geq 0}$ and $a_1 < \cdots < a_k$ (and $a_1,\dots,a_k,N,\eps$ are all given to the learning algorithm in the known-support setting).

For each $\bX_i$ we define $\bX'_i$ to be the ``zero-moded'' variant of $\bX_i$, namely  $\bX'_i = \bX_i - \mode(\bX_i)$ where $\mode(\bX_i) \in \{a_1,\dots,a_k\}$ is a mode of $\bX_i$ (i.e. $\mode(\bX_i)$ satisfies $\Pr[\bX_i=\mode(\bX_i)]
\geq \Pr[\bX_i = a_{i'}]$ for all $i' \in [k]$).  We define $\bS'$ to be $\sum_{i=1}^N \bX'_i$.  It is clear that $\bS' + V=\bS$ where $V $ is an (unknown) ``offset'' in $\Z$.  Below we will give an algorithm that learns $\bS' + V$ given independent draws from it.

For each $i \in [N]$ the support of random variable $\bX'_i$ is contained in $\{0,\pm q_1,\dots,\pm q_K\}$, where $K=O(k^2)$ and $\{q_1,\dots,q_K\}$ is the set of all distinct values achieved by $|a_\ell - a_{\ell'}|, 1 \leq \ell < \ell' \leq k.$  
As noted above each $\bX'_i$ has
$\Pr[\bX'_i=0] \geq 1/k \geq 1/K.$

To help minimize confusion we will consistently use letters $i,j,$ etc. for dummy variables that range over $1,\dots,N$ and $a,b,c,d$ etc. for dummy variables that range over $1,\dots,K.$

We define the following probabilities and associated values:
\begin{align}
\text{For~}i \in [N] \text{~and~}a\in[K]: \quad c_{q_a,i} &= \Pr[\bX'_i = \pm q_a] \label{eq:cqai}\\
\text{For~}a \in [K]: \quad 
c_{q_a} &= \sum_{i=1}^N c_{q_a,i}. \label{eq:cqa}
\end{align}

We may think of the value $c_{q_a}$ as the ``weight'' of $q_a$ in $\bS'$.  \ignore{We will divide our analysis into two cases.  The first case, roughly speaking, is that all of the $K$ values $c_{q_1},\dots,c_{q_K}$ are ``large'', and the second case is that at least one value $c_{q_a}$ is ``small.''  We give two distinct algorithms, one for each case.  The overall algorithm works by running both algorithms and using the hypothesis selection procedure, Proposition \ref{prop:log-cover-size}, to construct one final hypothesis.

Before entering into the two cases:}

It is useful for us to view 
$\bS'=\sum_{i=1}^N \bX'_i$ in the following way.  Recall that the support of $\bX'_i$ is contained in $\{0,\pm q_1,\dots,\pm q_K\}.$  For $i \in [N]$ we define a vector-valued random variable $\bY_i$ that is supported on $\{0,\pm \be_1,\dots,\pm \be_K\}$ by 
\begin{equation} \label{eq:Y}
\Pr[\bY_i = 0]=\Pr[\bX'_i=0] \geq {\frac 1 K}, \quad \quad \Pr[\bY_i=\tau \be_a] = \Pr[\bX'_i= \tau q_a] \text{~for~}\tau \in \{-1,1\},a \in [K].
\end{equation}
We define 
the vector-valued random variable $\bM=\sum_{i=1}^N \bY_i$, so we have $\bX'_i=(q_1,\dots,q_K) \cdot
\bY_i$ for each $i$ and $\bS'=(q_1,\dots,q_K) \cdot \bM.$  Summarizing for convenient later reference:

\begin{align}
\bX'_1,\dots,\bX'_N:&  \text{~~independent, each supported in~} \{0,\pm q_1,\dots,\pm q_K\} \label{eq:Xi}\\
\bS' = \bX'_1 + \cdots + \bX'_N : &  \text{~~supported in~}\Z \label{eq:Sprime_as_sum}\\
\bY_1,\dots,\bY_N: &  \text{~~independent, each supported in~} \{0,\pm \be_1,\dots,\pm \be_K\} \label{eq:Yi}\\
\bM = \bY_1 + \cdots + \bY_N: &  \text{~~supported in~}\Z^k \label{eq:M}\\
\bS' &= (q_1,\dots,q_K) \cdot \bM.  \label{eq:Sprime_as_dot}
\end{align}

From this perspective, in order to analyze $\bS'$ it is natural to analyze the multinomial random variable $\bM$, and indeed this is what we do in the next section.  

Finally, we note that while it suffices to learn $\bS'$ of the form captured in (\ref{eq:Xi}) and (\ref{eq:Sprime_as_sum})
for the $K$ and $\bS'$ that arise from our reduction to this case, our analysis will hold for all $K \in \Z^+$ and
all $\bS'$ of this form.

\section{Useful structural results when all $c_{q_a}$'s are large} \label{sec:structural}

In this section we establish some useful structural results for dealing with 
a distribution
$\bS' = \sum_{i=1}^N \bX'_i$ for which, roughly speaking, all the values $c_{q_1},\dots,c_{q_K}$ (as defined in Section~\ref{sec:setup-upper-bound}) are ``large.'' More formally, we shall assume throughout this section that each $c_{q_a} \geq \BIG$, where the exact value of the parameter $\BIG$ will be set later in the context of our learning algorithm in (\ref{eq:setbig}) (we note here only that $\BIG$ will be set to a fixed ``large'' polynomial in $K$ and $1/\eps$).  Looking ahead, we will later use the results of this section to handle whatever $c_{q_a}$'s are ``large''.

The high-level plan of our analysis is as follows:  In Section \ref{sec:multi-to-dG} we show that the multinomial distribution $\bM$ (recall (\ref{eq:M}) and (\ref{eq:Sprime_as_dot})) is close in total variation distance to a suitable discretized multidimensional Gaussian.  In Section \ref{sec:dG-to-comb-of-signed-PBD} we show in turn that such a discretized multidimensional Gaussian is close to a vector-valued random variable that can be expressed in terms of independent signed PBDs.  Combining these results, in Section \ref{sec:sicsirv-to-pure} we show that 
$\bS'$ is close in variation distance to a 
weighted sum of signed PBDs.  The lemma stating this, Lemma \ref{lem:close-to-pure-signed-sicsirv} in Section~\ref{sec:sicsirv-to-pure}, is one of the two main structural results in this section.  The second main structural result in this section, Lemma~\ref{l:mix}, is stated and proved in Section~\ref{sec:mix}.  Roughly speaking, it shows that,
for a weighted sum of signed PBDs,
it is possible to replace the scaled sum of the ``high-variance'' PBDs by a single scaled PBD.  This is useful later for learning since it leaves us in a situation where we only need to deal with scaled PBDs whose variance is ``not too high.''

We record some useful notation for this section:  for $i \in [N]$, $a \in [K]$ and $\tau \in \{-1,1\}$ let $p_{i,a,\tau}$ denote 
\begin{equation} \label{eq:piatau}
p_{i,a,\tau} := \Pr[\bY_i=\tau\be_a] = \Pr[\bX'_i = \tau q_a].
\end{equation}

\subsection{From multinomials to discretized multidimensional Gaussians} \label{sec:multi-to-dG} 

The result of this subsection, Lemma \ref{lem:CLT-mean-0-multinomial}, establishes that the multinomial distribution $\bM$ is close in total variation distance to a discretized multidimensional Gaussian.  

\ignore{

}

\begin{lemma}\label{lem:CLT-mean-0-multinomial}
  Let $\bY_1, \ldots, \bY_N$ be as in (\ref{eq:Yi}), so each $\bY_i$ has $\Pr[\bY_i=0] \geq 1/K.$  
Assume that
$c_{q_a} \geq \BIG$ for all $a \in [K]$.  As in (\ref{eq:M}) let $\bM = \bY_1 + \cdots + \bY_N$, and let $\tilde{\mu} = \E[\bM]$ be the $K$-dimensional mean of $\bM$ and $\tilde{\Sigma}$ be the $K \times K$ covariance matrix $\Cov(\bM)$. 
Then 

\begin{enumerate}

\item [(1)] Defining $\tilde{\sigma}^2$ to be the smallest eigenvalue of $\tilde{\Sigma}$, we have that $\tilde{\sigma}^2 \geq \BIG/K.$

\item [(2)] $\dtv(\bM, \mathcal{N}_D (\tilde{\mu}, \tilde{\Sigma})) \le O(K^{71/20} / \BIG^{1/20})$.

\end{enumerate}

\end{lemma}

\begin{proof}
Given part (1), Theorem \ref{thm:multinomial} directly gives the claimed variation distance bound in part (2), so in the following we establish (1).  

Since $\bY_1,\dots,\bY_N$ are independent
we have that
\[
\tilde{\Sigma} = \sum_{i=1}^N \tilde{\Sigma}_i, \quad \text{where~} \tilde{\Sigma}_i = \Cov(\bY_i).  
\]
Fix $i \in [N]$.  Recalling (\ref{eq:piatau}), we have that $\tilde{\Sigma}_i$ is the $K \times K$ matrix defined by
\[
(\tilde{\Sigma}_i)_{a,b} = \begin{cases}
(p_{i,a,1} + p_{i,a,-1})(1-p_{i,a,1}-p_{i,a,-1}) + 4p_{i,a,1}p_{i,a,-1} & \text{~if~}a=b\\
-(p_{i,a,1}-p_{i,a,-1})(p_{i,b,1}-p_{i,b,-1}) & \text{~if~}a \neq b.
\end{cases}
\]
Hence we have
\begin{equation} \label{eq:tildeSigma}
(\tilde{\Sigma})_{a,b} = \begin{cases}
\sum_{i=1}^N (p_{i,a,1} + p_{i,a,-1})(1-p_{i,a,1}-p_{i,a,-1}) + 4p_{i,a,1}p_{i,a,-1} & \text{~if~}a=b\\
\sum_{i=1}^N -(p_{i,a,1}-p_{i,a,-1})(p_{i,b,1}-p_{i,b,-1}) & \text{~if~}a \neq b.
\end{cases}
\end{equation}
For later reference (though we do not need it in this proof) we also note that the mean vector $\tilde{\mu}$ is defined by
\begin{equation} \label{eq:tildemu}
\tilde{\mu}_a = \sum_{i=1}^N (p_{i,a,1} - p_{i,a,-1}).
\end{equation}

Let $\delta_i = \Pr[\bY_i=0] = 1 - p_{i,1,1} - p_{i,1,-1} - \cdots - p_{i,K,1} - p_{i,K,-1}$ and observe that by assumption we have $\delta_i \geq 1/K$ for all $i \in [N].$
We lower bound the smallest eigenvalue using the variational characterization.  For any unit vector $\bx$ in $\R^K$, we have
\begin{eqnarray}
\bx^T \cdot \tilde{\Sigma} \cdot \bx &=& \sum_{a=1}^K x_a^2\left(\sum_{i=1}^N (p_{i,a,1} + p_{i,a,-1})(1-p_{i,a,1}-p_{i,a,-1}) + 4p_{i,a,1}p_{i,a,-1}\right) \nonumber \\
&&- \sum_{a=1}^K \sum_{b \in [K],b \neq a} x_a x_b 
\left(
\sum_{i=1}^N (p_{i,a,1}-p_{i,a,-1})(p_{i,b,1}-p_{i,b,-1}) 
\right). \label{eq:a}
\end{eqnarray}
Let $p'_{i,a,1}=p_{i,a,1}+p_{i,a,-1}$.  Recalling that each $p_{i,a,1},p_{i,a,-1} \geq 0$, it is not difficult to see that then we have
\begin{eqnarray}
(\ref{eq:a}) &\geq&  \sum_{a=1}^K x_a^2\left(\sum_{i=1}^N p'_{i,a,1}(1-p'_{i,a,1})\right) 
- \sum_{a=1}^K \sum_{b \in [K],b \neq a} |x_a| \cdot |x_b| 
\left(
\sum_{i=1}^N p'_{i,a,1}p'_{i,b,1} 
\right), \label{eq:b}
\end{eqnarray}
so for the purpose of lower bounding  (\ref{eq:a}) it suffices to lower bound (\ref{eq:b}).  
Rewriting $p'_{i,a,1}$ as $p_{i,a}$ for notational simplicity, so now 
$\delta_i = 1-p_{i,1} - \cdots - p_{i,K}$,we have
\begin{eqnarray}
(\ref{eq:b}) & \geq &
\sum_{a=1}^K x_a^2 \sum_{i=1}^N p_{i,a}(1-p_{i,a})
- \sum_{a=1}^K \sum_{b \in [K], b \neq a} {|x_a| \cdot | x_b| }\sum_{i=1}^N p_{i,a}p_{i,b} \nonumber\\
& = &
\sum_{i=1}^N \left( \sum_{a=1}^K p_{i,a}(1-p_{i,a})x_a^2 - \sum_{a=1}^K \sum_{b\in [K], b \neq a} p_{i,a}p_{i,b}{|x_a| \cdot | x_b| } \right) \nonumber \\
&=& \sum_{i=1}^N \left( \sum_{a=1}^K \delta_i p_{i,a} x_a^2 + \sum_{a=1}^K p_{i,a} x_a^2 \left(\sum_{b\in [K], b \neq a} p_{i,b}\right) - \sum_{a=1}^K \sum_{b\in [K], b \neq a} p_{i,a}p_{i,b}{|x_a| \cdot | x_b| } \right)  \nonumber \\
&=& \sum_{i=1}^N \left( \delta_i \sum_{a=1}^K p_{i,a} x_a^2 + \sum_{a=1}^K \sum_{b\in [K], b \neq a} (p_{i,a}p_{i,b}x_a^2 - p_{i,a}p_{i,b}{|x_a| \cdot | x_b| }) \right)  \nonumber \\
&=& \sum_{i=1}^N \left( \delta_i \sum_{a=1}^K p_{i,a} x_a^2 +  \sum_{a=1}^K \sum_{b<a} p_{i,a}p_{i,b}({|x_a|}-{|x_b|})^2 \right)  \nonumber \\
&\geq& \sum_{i=1}^N \delta_i  \sum_{a=1}^K  p_{i,a} x_a^2 .
\end{eqnarray}
Recalling that $\delta_i \geq 1/K$ for all $i \in [N]$ and $\sum_{i=1}^n p_{i,a} = c_{q_a} \geq \BIG$ for all $a \in [K]$, we get
\[
\sum_{i=1}^N \delta_i  \sum_{a=1}^K p_{i,a} x_a^2 \geq {\frac 1 K}  \sum_{a=1}^K x_a^2 \sum_{i=1}^N p_{i,a} \geq {\frac 1 K} \sum_{a=1}^K c_{q_a} x_a^2 \geq {\frac \BIG K} \sum_{a=1}^K x_a^2 = {\frac \BIG K},
\]
so $\tilde{\sigma}^2 \geq \BIG/K$ and the lemma is proved.
\end{proof}

\subsection{From discretized multidimensional Gaussians to combinations of independent signed PBDs}
\label{sec:dG-to-comb-of-signed-PBD}

The first result of this subsection, Lemma \ref{lem:dG-to-comb}, is a technical lemma establishing that the discretized multidimensional Gaussian given by Lemma \ref{lem:CLT-mean-0-multinomial} is close to a vector-valued random variable in which each marginal (coordinate) is a $(\pm 1)$-weighted linear combination of independent discretized Gaussians, certain of which are promised to have large variance.

\begin{lemma} \label{lem:dG-to-comb}~
Under the assumptions of Lemma~\ref{lem:CLT-mean-0-multinomial} the following items (1) and (2) both hold:
\begin{itemize}

\item [(1)]
The pair $\tilde{\mu} \in \R^K, \tilde{\Sigma} \in \R^{K \times K}$, defined in (\ref{eq:tildemu}) and (\ref{eq:tildeSigma}), are such that there exist $\mu_{a,b}\in \R$, $1 \leq a \leq b \leq K,$ satisfying
\begin{equation} \label{eq:tildemucondition}
\tilde{\mu}_a = \mu_{a,a} + \sum_{c < a} \mu_{c,a} + \sum_{a < d}  \sign(\tilde{\Sigma}_{a,d}) \cdot \mu_{a,d},
\end{equation}
and there exist $\sigma_{a,b} \in \R$,
$1 \leq a \leq b \leq K,$ such that
\begin{equation} \label{eq:tildeSigmacondition}
\sigma_{a,b}^2 = |\tilde{\Sigma}_{a,b}| = |\tilde{\Sigma}_{b,a}|  \text{~~~for all~}a<b \quad \quad \text{and} \quad \quad
 \tilde{\Sigma}_{a,a}  = \sigma_{a,a}^2 
                     + \sum_{c<a} \sigma_{c,a}^2
                     + \sum_{a<d} \sigma_{a,d}^2.
\end{equation}
Furthermore, for all $a \in [K]$ we have $\sigma^2_{a,a} \geq \sigma^2$, where we define $\sigma^2 := \BIG/K.$

\item [(2)]
Let $\bU_{a,b},$ $1 \leq a < b \leq K$ be discretized Gaussians $\bU_{a,b} = \calN_D(\mu_{a,b},\sigma^2_{a,b})$ that are all mutually independent.  For $a \in [K]$ let $\bX_a$ be defined as
\[
\bX_a = \bU_{a,a} + \sum_{c<a} \bU_{c,a} + \sum_{a<d} \sign(\tilde{\Sigma}_{a,d}) \cdot \bU_{a,d}.
\]
Then
\[
\dtv\left((\bX_1,\dots,\bX_K),\calN_D(\tilde{\mu},\tilde{\Sigma})\right) \leq {\frac {K^2} \sigma} = {\frac {K^{5/2}} {\BIG^{1/2}}}.
\]
\end{itemize}
\end{lemma}

\begin{proof}
We first prove part (1).  Existence of the desired $\mu_{a,b}$ values is immediate since for each $a \in [K]$ the variable $\mu_{a,a}$ appears in only one equation given by (\ref{eq:tildemucondition}) (so we can select arbitrary values for each $\mu_{a,b}$ with $a < b$, and there will still exist a value of $\mu_{a,a}$ satisfying (\ref{eq:tildemucondition})).  The first part of (\ref{eq:tildeSigmacondition}) is trivial since for $a < b$ we take $\sigma^2_{a,b} = |\tilde{\Sigma}_{a,b}|$ (which of course equals $|\tilde{\Sigma}_{b,a}|$ since the covariance matrix $\tilde{\Sigma}$ is symmetric).  For the second part we take
$\sigma^2_{a,a}=
 \tilde{\Sigma}_{a,a}  - \sum_{c<a} \sigma_{c,a}^2 - \sum_{a<d} \sigma_{a,d}^2$ which we now proceed to lower bound.
\begin{align*}
\sigma^2_{a,a}
&= \tilde{\Sigma}_{a,a}  - \sum_{c<a} \sigma_{c,a}^2 - \sum_{a<d} \sigma_{a,d}^2 = \tilde{\Sigma}_{a,a} - \sum_{b \neq a} |\tilde{\Sigma}_{b,a}| \\
&\geq \tilde{\Sigma}_{a,a} -  \sum_{b \neq a}  \sum_{i=1}^N(p_{i,a,1}+p_{i,a,-1})(p_{i,b,1}+p_{i,b,-1}) \tag{by (\ref{eq:tildeSigma})}\\
&= \sum_{i=1}^N (p_{i,a,1} + p_{i,a,-1}) \left((1 - p_{i,a,1}-p_{i,a,-1})  
- \sum_{b \neq a} (p_{i,b,1}+p_{i,b,-1})\right) + 4 p_{i,a,1}p_{i,a,-1} \tag{again by (\ref{eq:tildeSigma})}\\
&\geq \sum_{i=1}^N (p_{i,a,1} + p_{i,a,-1}) \left((1 - p_{i,a,1}-p_{i,a,-1})  
- \sum_{b \neq a} (p_{i,b,1}+p_{i,b,-1})\right)\\
&= \sum_{i=1}^N (p_{i,a,1} + p_{i,a,-1}) \left(\delta_i +   \sum_{b \neq a} (p_{i,b,1}+p_{i,b,-1})
- \sum_{b \neq a} (p_{i,b,1}+p_{i,b,-1})\right) \tag{by definition of $\delta_i$}\\
&=\sum_{i=1}^N \delta_i (p_{i,a,1} + p_{i,a,-1}) \geq {\frac 1 K} \sum_{i=1}^N (p_{i,a,1}+p_{i,a,-1}) \tag{since $\delta_i \geq {\frac 1 K}$}\\
&\geq {\frac 1 K} c_{q_a} \geq {\frac \BIG K}.
\end{align*}

With $\mu_{a,b}$ and $\sigma_{a,b}$ in hand, now we turn to proving part (2) of the lemma.
For $1 \leq a \leq b \leq K$ let $\bU'_{a,b}$ be the (non-discretized) univariate Gaussian $\calN(\mu_{a,b},\sigma^2_{a,b})$ that $\bU_{a,b}$ is based on, so $\bU_{a,b} = \round{\bU'_{a,b}}$ and the distributions $\bU'_{a,b}$ are all mutually independent.
For $a \in [K]$ we define random variables\ignore{\pnote{How about calling these $\bV'_{a,a}$ and $\bV_{a,a}$?}}
$\bV'_{a,a}$, $\bV_{a,a}$ as
\begin{align*}
\bV'_{a,a} &= \sum_{c<a} \bU'_{c,a} + \sum_{a<d} \sign(\tilde{\Sigma}_{a,d}) \cdot \bU'_{a,d},\\
\bV_{a,a} &= \sum_{c<a} \round{\bU'_{c,a}} + \sum_{a<d} \round{ \sign(\tilde{\Sigma}_{a,d}) \cdot \bU'_{a,d}}
= \sum_{c<a} \round{\bU'_{c,a}} +  \sum_{a<d}  \sign(\tilde{\Sigma}_{a,d}) \cdot \round{ \bU'_{a,d} }\\
&= \sum_{c<a} \bU_{c,a} +  \sum_{a<d}  \sign(\tilde{\Sigma}_{a,d}) \cdot \bU_{a,d} .
\end{align*}
Fix a possible outcome $(u'_{a,b})_{a < b}$ of $(\bU'_{a,b})_{a < b}$ and for each $a < b$ let $u_{a,b}= \round{ u'_{a,b} }$ be the corresponding outcome of $\bU_{a,b}.$  For $a \in [K]$ let
\begin{align*}
v'_{a,a} &= \sum_{c<a} u'_{c,a} +\sum_{a<d}   \sign(\tilde{\Sigma}_{a,d}) \cdot   u'_{a,d},\\
v_{a,a} &= \sum_{c<a} \round{ u'_{c,a} } + \sum_{a<d} \round{ \sign(\tilde{\Sigma}_{a,d}) \cdot u'_{a,d} } =  \sum_{c<a} u_{c,a} + \sum_{a<d} \sign(\tilde{\Sigma}_{a,d}) \cdot u_{a,d}.
\end{align*}

Recalling Lemma~\ref{lem:B5plus}, we have that 
\[
\dtv(\round{ \bU'_{a,a} + v'_{a,a} }, \round{ \bU'_{a,a} } + v_{a,a}) \leq {\frac K \sigma}
\]
for each $a \in [K]$, and hence by independence we get that
\[
\dtv \left((\round{\bU'_{1,1} + v'_{1,1}},...,\round{\bU'_{K,K}+v'_{K,K}}),\;
        (\round{\bU'_{1,1}} + v_{1,1},...,\round{\bU'_{K,K}} + v_{K,K}) \right)
   \leq \frac{K^2}{\sigma}.
\]
Averaging over all outcomes of $(u'_{a,b})_{a < b} \leftarrow (\bU'_{a,b})_{a < b}$, we get that
\[
\dtv \left((\round{\bU'_{1,1} + \bV'_{1,1}},...,\round{\bU'_{K,K}+\bV'_{K,K}}),\;
        (\round{\bU'_{1,1}} + \bV_{1,1},...,\round{\bU'_{K,K}} + \bV_{K,K}) \right)
   \leq \frac{K^2}{\sigma}.
\]
To complete the proof it remains to show that the vector-valued random variable \[(\round{\bU'_{1,1} + \bV'_{1,1}},\dots,\round{\bU'_{K,K}+\bV'_{K,K}})\] is distributed according to $\calN_D(\tilde{\mu},\tilde{\Sigma}).$  It is straightforward to verify, using (\ref{eq:tildemucondition}) and linearity of expectation, that $\E[\bU'_{a,a} + \bV'_{a,a}] = \tilde{\mu}_a.$  For the  covariance matrix, we first consider the diagonal terms:  we have $\Var[\bU'_{a,a} + \bV'_{a,a}] = \tilde{\Sigma}_{a,a}$ by the second part of (\ref{eq:tildeSigmacondition}) and independence of the $\bU'_{a,b}$ distributions.  Finally, for the off-diagonal terms, for $a < b$ we have
\begin{align*}
&\Cov(\bU'_{a,a} + \bV'_{a,a}, \bU_{b,b}' + \bV_{b,b}')\\
  & = \Cov\left(\bU'_{a,a} + \sum_{c < a} \bU'_{c,a} 
                + \sum_{a < d } \sign(\tilde{\Sigma}_{a,d}) \cdot \bU_{a,d}', 
                \bU'_{b,b} + \sum_{c < b } \bU'_{c,b} + \sum_{b < d }  \sign(\tilde{\Sigma}_{b,d}) \cdot \bU'_{b,d} \right) \\
 & = \Cov( \sign(\tilde{\Sigma}_{a,b}) \cdot \bU'_{a,b}, \bU'_{a,b}) =  \sign(\tilde{\Sigma}_{a,b}) \cdot \Var[\bU'_{a,b}] =  \sign(\tilde{\Sigma}_{a,b}) \cdot \sigma^2_{a,b} = \tilde{\Sigma}_{a,b} = \tilde{\Sigma}_{b,a}
\end{align*}
as desired.
\end{proof}

We would like a variant of Lemma~\ref{lem:dG-to-comb} where signed PBDs play the role of discretized Gaussians.\ignore{\rnote{As Phil suggested earlier, we don't absolutely have to do this, the alternative would be just to work with discretized Gaussians.  But whether or not we go to signed PBDs, we will need something like Lemma \ref{lem:dG-to-comb-signed-PBDs}
 below.  If we stay with discretized Gaussians, we still have to get to a final situation where we just have 
$q_1 \bZ_1 + \cdots + q_K \bZ_K$ (where the $\bZ_i$'s would now be large-variance discretized Gaussians).  This would be done following the same kind of approach carried out below for signed PBDs:  use the fact that if you add two large-variance discretized Gaussians you get (to within small error) a discretized Gaussian, and if you add a small-variance discretized Gaussian to a large-variance discretized Gaussian, then the result is close to a shifted version of the large-variance discretized Gaussian summand (because up to a small error, thanks to shift-invariance of the big guy, you can replace the small-variance discretized Gaussian by a delta distribution located anywhere within the interval where it puts almost all of its mass).

So at this particular spot I don't see that it is significantly easier to go one way (signed PBDs) versus another (discretized Gaussians).}} This is given by the following lemma.  (Note that the lemma also ensures that every nontrivial signed PBD has high variance; this will be useful later.)

\begin{lemma} \label{lem:dG-to-comb-signed-PBDs}
\ignore{\rnote{This lemma is analogous to Corollary 31 of the STOC submission.  Do we need to add to the lemma statement a lower bound on $\Var[\bW_{a,a}]$ and an upper bound on $\Var[\bW_{a,b}]$ for $a<b$, analogous to the end of the Corollary 31 statement?}}
Given the $\tilde{\mu} \in \R^K, \tilde{\Sigma} \in \R^{K \times K}$ from Lemma \ref{lem:CLT-mean-0-multinomial} and the $\mu_{a,b}$, $\sigma_{a,b}$ and $\sigma^2$ defined in Lemma \ref{lem:dG-to-comb}, there exist signed PBDs $\bW_{a,b}$, $1 \leq a \leq b \leq K$, each of which is either trivial (a constant random variable) or has $\Var[\bW_{a,b}] \geq \sigma^{1/2} = \BIG^{1/4}/K^{1/4}$, such that
the random variables $\bS_a,$ $a \in [K]$, defined as
\begin{equation} \label{eq:def-Sa}
\bS_a =  \bW_{a,a} + \sum_{c<a} \bW_{c,a} + \sum_{a<d} \sign(\tilde{\Sigma}_{a,d})  \bW_{a,d},
\end{equation}
satisfy
\[
\dtv\left((\bS_1,\dots,\bS_K),\calN_D(\tilde{\mu},\tilde{\Sigma})\right) \leq O\left({\frac {K^2} {\sigma^{1/4}}}\right)
= O\left(
{\frac {K^{17/8}}{\BIG^{1/8}}}
\right).
\]
\end{lemma}
\begin{proof}
Let $\bU_{a,b}$, $\bX_a$ be as defined in Lemma~\ref{lem:dG-to-comb}.  We ``swap out'' each discretized Gaussian
$\bU_{a,b} = \calN_D(\mu_{a,b},\sigma^2_{a,b})$ in $\bX_a$ for a signed PBD $\bW_{a,b}$ as follows:  Given $1 \leq a \leq b \leq K$,
\begin{itemize}

\item [(I)] If $\sigma^2_{a,b} \geq \sigma^{1/2}$, we define $\bW_{a,b}$ to be a signed PBD that has $|\E[\bW_{a,b}] - \mu_{a,b}|\leq 1/2$ and $\Var[\bW_{a,b}]=\sigma_{a,b}^2$.  (To see that there exists such a signed PBD, observe that we can take $N_1$ many Bernoulli random variables each with expectation $p$, satisfying $N_1 p (1-p) = \sigma^2_{a,b}$, to exactly match the variance, and then take an additional $N_2$ many constant-valued random variables (each of which is 1 or $-1$ depending on whether $N_1 p$ is greater or less than $\mu_{a,b}$) to get the mean of the signed PBD to lie within an additive $1/2$ of $\mu_{a,b}.$)

\item [(II)] If $\sigma^2_{a,b} < \sigma^{1/2}$ we define $\bW_{a,b}$ to be a trivial signed PBD that has $\E[\bW_{a,b}] = \round{\mu_{a,b}}$
and $\Var[\bW_{a,b}] = 0.$

\end{itemize}
In the above definition all $\bW_{a,b}$'s are independent of each other.  We note that Lemma~\ref{lem:dG-to-comb} implies 
that when $b=a$ the PBD $\bW_{a,a}$ has $\Var[\bW_{a,a}] = \sigma^2_{a,a} \geq \sigma^2 = \BIG/K \gg \sigma^{1/2}$, and hence the PBD $\bW_{a,a}$ falls into the ``large-variance''  Case~(I) above.

The random variable $\bS_a$ defined in Equation (\ref{eq:def-Sa}) is the analogue of $\bX_a$ from Lemma~\ref{lem:dG-to-comb} but with $\bW_{a,b}$ replacing each $\bU_{a,b}$.\ignore{We note that since $\bS_a$ is a signed sum of independent signed PBDs it itself is a signed PBD, \rnote{These $\bS_a$'s are not independent of each other, but it's fine; fix the presentation} and moreover $\Var[\bS_a] \geq \Var[\bW_{a,a}] \geq \sigma^2.$}  To establish the variation distance bound, fix $a \in [K]$; we first argue that the variation distance between $\bS_a$ and $\bX_a$ is small.  We start by observing that
since $\Var[\bW_{a,a}] \geq \sigma^2$,
Theorem~\ref{thm:CGS} and Lemma~\ref{lem:B5plus} give
\begin{equation} \label{eq:first}
\dtv(\bW_{a,a},\bU_{a,a}) \leq O(1/\sigma),
\end{equation}
and moreover $\bW_{a,a}$ is $O(1/\sigma)$-shift-invariant by Fact~\ref{fact:good-shift-invariance}.

Now consider a $c<a$ such that
$\bW_{c,a}$ falls into Case~(II).  By the standard concentration bound for the Gaussian $\bU'_{c,a} \sim \calN(\mu_{c,a},\sigma^2_{c,a})$ on which $\bU_{c,a}$ is based, we have that $\Pr[\bU'_{c,a} \notin [\mu_{c,a} - t \sigma_{c,a},\mu_{c,a} + 
t \sigma_{c,a}]] \leq 2e^{-t^2/2}$ for all $t > 0.$  It follows from Claim \ref{claim:small-var-ok} (stated and justified below) and the $O(1/\sigma)$-shift-invariance of $\bW_{a,a}$  that 
\[
\dtv(\bW_{a,a} + \bW_{c,a}, \bW_{a,a} + \bU_{c,a}) \leq O
\left({\frac {t \sigma_{c,a} + 1} {\sigma}}
\right)
+ 2e^{-t^2/2}.
\]
Selecting $t = \sigma^{1/4}$ so that $t \sigma_{c,a}+1 \leq \sigma^{1/4} \cdot \sigma^{1/4} + 1 = O(\sigma^{1/2})$ and $e^{-t^2/2} = o(1/\sigma^{1/2})$, we get that
\begin{equation} \label{eq:second}
\dtv(\bW_{a,a} + \bW_{c,a}, \bW_{a,a} + \bU_{c,a}) \leq O
\left( {\frac 1 {\sigma^{1/2}}}\right).
\end{equation}
A similar argument holds for each $d>a$ such that $\bW_{a,d}$ falls into Case~(II), giving
\begin{equation} \label{eq:second-and-a-half}
\dtv(\bW_{a,a} + \sign(\tilde{\Sigma}_{a,d}) \cdot \bW_{a,d}, \bW_{a,a} + \sign(\tilde{\Sigma}_{a,d}) \cdot \bU_{a,d}) \leq O
\left( {\frac 1 {\sigma^{1/2}}}\right).
\end{equation}
Finally, for each $c < a$ such that $\bW_{c,a}$ falls into Case~(I),
once again applying Theorem~\ref{thm:CGS} and Lemma~\ref{lem:B5plus}, we get
\begin{equation} \label{eq:third}
\dtv(\bW_{c,a},\bU_{c,a}) \leq O(1/\sigma^{1/4}),
\end{equation}
and similarly for $d>a$ such that $\bW_{a,d}$ falls into Case~(I) we have
\begin{equation} \label{eq:third-and-a-half}
\dtv(\sign(\tilde{\Sigma}_{a,d}) \cdot \bW_{a,d}, \sign(\tilde{\Sigma}_{a,d}) \cdot \bU_{a,d}) \leq O(1/\sigma^{1/4}).
\end{equation}

Combining (\ref{eq:first}---\ref{eq:third}) and recalling the definitions of $\bS_a$ and $\bX_a$, by the triangle inequality for each 
$a \in [K]$ we have
\[
\dtv(\bS_a,\bX_a) \leq O\left({\frac K {\sigma^{1/4}}}\right).
\]
Finally, another application of the triangle inequality gives
\[
\dtv((\bS_1,\dots,\bS_K),(\bX_1,\dots,\bX_K)) \leq O\left({\frac {K^2} {\sigma^{1/4}}}\right),
\]
which with Lemma~\ref{lem:dG-to-comb} gives the claimed bound.
\end{proof}

The following claim is an easy consequence of the definition of shift-invariance:
\begin{claim} \label{claim:small-var-ok}
Let $\bA$ be an integer random variable that is $\alpha$-shift-invariant, and let $\bB$ be an integer random variable such that $\Pr[\bB \notin [u,u+r]] \leq \delta$ for some integers $u,r.$  Then for any integer $r' \in [u,u+r]$ we have $\dtv(\bA + \bB,\bA + r') \leq \alpha r + \delta.$
\end{claim}

\subsection{$\bS'$ is close to a shifted weighted sum of signed PBDs} \label{sec:sicsirv-to-pure}

Recall that $\bS' = (q_1,\dots,q_K) \cdot \bM$ is as defined in (\ref{eq:Sprime_as_dot}). Combining Lemmas \ref{lem:CLT-mean-0-multinomial} and~\ref{lem:dG-to-comb-signed-PBDs}, and taking the dot-product with $(q_1,\dots,q_K)$ to pass from $\bM$ to $\bS'$, we get that the variation distance between $\bS'=(q_1,\dots,q_K) \cdot \bM$ and $(q_1,\dots,q_K) \cdot (\bS_1,\dots,\bS_K)$ is at most $O(K^{71/20} / \BIG^{1/20})$. We can express $(q_1,\dots,q_K) \cdot (\bS_1,\dots,\bS_K)$ as
\begin{align*}
&\sum_{a=1}^K q_a \left(\bW_{a,a} + \sum_{c<a} \bW_{c,a} + \sum_{a<d} \sign(\tilde{\Sigma}_{a,d})  \bW_{a,d}\right)\\
&= \sum_{a=1}^K q_a \bW_{a,a} + 
\sum_{1 \leq a < b \leq K} (q_b + \sign(\tilde{\Sigma}_{a,b})\cdot q_a) \bW_{a,b}.
\end{align*}
Recalling that each $\bW_{a,b}$ is either a constant random variable or a signed PBD with variance at least $\sigma^{1/2} = \BIG^{1/4}/K^{1/4}$, that each $\Var[\bW_{a,a}] \geq \sigma^2 > \sigma^{1/2},$ and that all of the distributions $\bW_{a,a},\bW_{a,b}$ are mutually independent, we get the following result showing that
$c_{q_1},\dots,c_{q_K} \geq \BIG$ implies that
$\bS'$ is close to a weighted sum of signed PBDs.

\begin{lemma} \label{lem:close-to-pure-signed-sicsirv}
Assume that $c_{q_1},\dots,c_{q_K} \geq \BIG$.   Then there is an integer $V'$, a subset of pairs $A \subseteq \{(a,b): 1 \leq a < b \leq K\}$, and a set of sign values $\{\tau_{a,b}\}_{(a,b) \in A}$ where each $\tau_{a,b} \in \{-1,1\}$,
such that $\dtv(\bS',\bB) =  O(K^{71/20} / \BIG^{1/20})$, where $\bB$ is 
a shifted sum of signed PBDs
\begin{equation} \label{eq:def-bB}
\bB = V' +  \sum_{a=1}^K q_a \bW_{a,a} + \sum_{(a,b) \in A} (q_b + \tau_{a,b}\cdot q_a) \bW_{a,b}
\end{equation}
in which all the $\bW_{a,a}$ and $\bW_{a,b}$ distributions are independent signed PBDs with variance at least $\BIG^{1/4}/K^{1/4}.$
\end{lemma}

\subsection{A useful limit theorem:  Simplifying by coalescing multiple large-variance scaled PBDs into one} \label{sec:mix}

Lemma~\ref{lem:close-to-pure-signed-sicsirv} leads to consideration of
distributions of the form $\bT = r_1 \bT_1 + \cdots + r_D \bT_D$, where $\bT_1, \dots, \bT_D$ are independent signed large-variance PBDs.  Let us consider for a moment the
case that $D = 2$, so that $\bT = r_1 \bT_1 + r_2 \bT_2$.  (As we will see, to handle the case of general $D$ it suffices to consider
this case.)  Since
$\mathrm{gcd}(r_1,r_2)$ divides every outcome of $\bT$, we may assume that
$\mathrm{gcd}(r_1,r_2) = 1$ essentially without loss of generality.
\ignore{As illustrated in
Figure~\ref{f:three-sicsirvs}(d),} When $\mathrm{gcd}(r_1,r_2) = 1$, if the variance of $\bT_2$ is large enough
relative to $r_1$, then the gaps between multiples of $r_1$ are filled in, and
$\bT$ is closely approximated by a single PBD.  This is helpful for learning, because
it means that cases in which $\Var[\bT_2]$ is this large are subsumed by cases in which there
are fewer PBDs.  This phenomenon is the subject of Lemma~\ref{l:mix}.

\begin{lemma}
\label{l:mix}
Let $\bT = r_1 \bT_1 + \cdots + r_D \bT_D$ where $\bT_1, \dots, \bT_D$ are independent signed PBDs and $r_1,\dots,r_D$ are nonzero integers with the following properties:
\begin{itemize}

\item $\Var[r_1 \bT_1] \geq {\frac 1 D} \Var[\bT]$;

\item For each $a \in \{2,\dots,D\}$ we have $\Var[\bT_a] \geq \max\{\sigma_{\min}^2,\left(\frac {r_1} {\eps'}\right)^2\},$ where $\sigma_{\min}^2 \geq (1/\eps')^8$.

\end{itemize}
Let $r' = \gcd(r_1,\dots,r_D).$
Then there is a signed PBD $\bT'$ with $\Var[r' \bT'] = \Var[\bT]$ such that
\[
\dtv \left(\bT,\; r' \bT' \right)
   \leq O(D \eps').
\]
\end{lemma}

\begin{proof}
{\bf Reduction to the case that $D = 2$.}  We begin by showing that the case $D=2$ implies the general case by induction, and thus it suffices to prove the $D=2$ case.   Let us suppose that we have proved the lemma in the $D=2$ case and in the $D=t-1$ case; we now use these to prove the $D=t$ case.  By the $D=2$ case, there is an absolute constant $C>0$ and a signed PBD $\bT_{12}$ such that\ignore{, for all small enough
$\epsilon'$,} we have
\begin{equation}
\label{e:first.step}
\dtv (r_1 \bT_1 + r_2 \bT_2, \gcd(r_1, r_2) \bT_{12} ) \leq C \epsilon' \mbox{ and }
   \Var[r_1 \bT_1 + r_2 \bT_2] = \Var[\gcd(r_1, r_2) \bT_{12}].
\end{equation}

Since for all $a = 3,...,t$ we have
\[
\Var[\bT_a] \geq  \left( \frac{ r_1 }{\epsilon'} \right)^2 \geq \left( \frac{ \gcd(r_1, r_2) }{\epsilon'} \right)^2,
\]
the $D=t-1$ case implies that, if $\bT_{12}, \bT_3,..., \bT_{T}$ are mutually independent, then there is a PBD $\bT'$ such that
\[
\dtv\left(\gcd(r_1, r_2) \bT_{12} + \sum_{a = 3}^{t} r_a \bT_a,\;
       r' \bT'\right) \leq C(t-1) \epsilon'
\]
and
\[
\Var[r' \bT'] = \Var\left[\gcd(r_1, r_2) \bT_{12} + \sum_{a = 3}^{t} r_a \bT_a\right]
 = \Var[\gcd(r_1, r_2) \bT_{12}] + \Var\left[\sum_{a = 3}^{t} r_a \bT_a\right],
\]
which, combined with (\ref{e:first.step}), completes the proof of the $D=t$ case.
We thus subsequently focus on the $D=2$ case.

\medskip \noindent
{\bf Reduction to the case that $r' = 1$.} Next, we note that we may assume without loss of generality that
$r' = \gcd(r_1, r_2) = 1$, 
since dividing each $r_a$ by $r'$ scales down
$\Var[r_1 \bT_1 + r_2 \bT_2]$ by $(r')^2$ and
$
\dtv \left(r_1 \bT_1 + r_2 \bT_2,\; r' \bT' \right)$ is easily seen to equal
$\dtv \left((r_1/r') \bT_1 + (r_2/r') \bT_2,\; \bT' \right).$

\medskip \noindent
{\bf Main proof of the $D=2,r'=1$ case.}
\newcommand{\tT}{\tilde{\bT}}
Recall that $\bT = r_1 \bT_1 + r_2 \bT_2$.
Let $\mu$ denote 
$\E[\bT]$, and let $\sigma^2$ denote $\Var[\bT]$.

As in \cite{GMRZ11}, we will use shift-invariance to go from bounds on
$\dk$ to bounds on $\dtv$.  
%
Our first step is to give a bound on $\dk$.  For this we will use
the following well-known Berry-Esseen-like inequality, which can be shown using Lemma~\ref{lem:translated Poisson approximation}, Theorem~\ref{thm:CGS} and Gaussian anti-concentration:
\begin{lemma}
\label{l:be}
There is a universal constant $c$ such 
\[
\dk(TP(\mu,\sigma^2), N(\mu, \sigma^2)) \leq \frac{c}{\sigma}
\]
for all $\mu$ and all $\sigma^2 > 0.$
\end{lemma}

Now we are ready for our bound on the Kolmogorov distance:
\begin{lemma}
\label{l:kolm}
$\dk(\bT, TP(\mu,\sigma^2))
  \leq O(1/\sigma_{\min}).$
\end{lemma}
{\bf Proof:} 
Lemma~\ref{lem:translated Poisson approximation}
implies that for $a = 1,2$ we have
\[
\dk(\bT_a, TP(\mu(\bT_a), \sigma(\bT_a)^2))
  \leq O(1/\sigma_{\min}),
\]
which directly implies
\[
\dk(r_a \bT_a, r_a TP(\mu(\bT_a), \sigma(\bT_a)^2))
  \leq O(1/\sigma_{\min}).
\]
Lemma~\ref{l:be} and the triangle inequality then give
\[
\dk(r_a \bT_a, N(r_a \mu(\bT_a), r_a^2 \sigma(\bT_a)^2))
  \leq O(1/\sigma_{\min}),
\]
and applying Lemma~\ref{l:be} and the triangle inequality again, we
get
\[
\dk(r_a \bT_a, TP(r_a \mu(\bT_a), r_a^2 \sigma(\bT_a)^2))
  \leq O(1/\sigma_{\min}).
\]
The lemma follows from the fact that
$\dk(\bX+\bY,\bX'+\bY') \leq \dk(\bX,\bX') + \dk(\bY,\bY')$ when $\bX,\bY$
are independent and $\bX',\bY'$ are independent.
\qed

\medskip

Facts~\ref{fact:SI} and~\ref{fact:good-shift-invariance} together
imply that $\bT$ is $O(1/\sigma_{\min})$-shift invariant at scales
$r_1$ and $r_2$, but, to apply Lemma~\ref{lem:GMRZ}, we need it to be
shift-invariant at a smaller scale.  Very roughly, we will do this by
effecting a small shift using a few shifts with steps with sizes in
$\{ r_1, r_2 \}$.  The following generalization of B\'{e}zout's
Identity starts to analyze our ability to do this.

\begin{lemma}\label{l:CRT}
Given any integer $0 \le u < r_1 \cdot r_2$, there are integers $v_1, v_2$ such that $u = v_1 \cdot r_1 + v_2 \cdot r_2$ with $|v_1|<r_2, |v_2|<r_1$.
\end{lemma}
\begin{proof}
By B\'{e}zout's Identity, there exist $x_1$ and $x_2$ with $|x_1| < r_2$
and $|x_2| < r_1$ such that
\[
x_1 r_1 + x_2 r_2 = 1.
\]
Let $y_1$ be obtained by adding $r_2$ to $x_1$
if $x_1$ is negative, and otherwise
just taking $x_1$, and define $y_2$ similarly; i.e.,
$y_1 = x_1 + r_2 \mathbf{1}[x_1 < 0] $ and $y_2 = x_2 + r_1 \mathbf{1}[x_2 < 0]$.  Then
\[
y_1 r_1 + y_2 r_2 = 1 \!\!\mod (r_1 r_2)
\]
and $0 \leq y_1 < r_2$ and $0 \leq y_2 < r_1$.  Thus
\[
u y_1 r_1 + u y_2 r_2 = u \!\!\mod (r_1 r_2).
\]
This in turn implies that
\[
u = u y_2 r_2 \!\!\mod r_1 \mbox{ and } u = u y_1 r_1 \!\!\mod r_2,
\]
so if $z_1 \in \{0,1,\dots,r_2-1\}$ and $z_2 \in \{0,1,\dots,r_1-1\}$ satisfy
 $z_1 = u y_1 \!\!\mod r_2$ and $z_2 = u y_2 \!\!\mod r_1$, we get
\[
(z_1 r_1 + z_2 r_2) = u \!\!\mod r_1
\]
and
\[
(z_1 r_1 + z_2 r_2) = u \!\!\mod r_1.
\]
By the Chinese Remainder Theorem, $z_1 r_1 + z_2 r_2 = u \!\!\mod (r_1 r_2)$.
Furthermore, as $0 \leq z_1 < r_2$ and $0 \leq z_2 < r_1$, we have
$0 \leq z_1 r_1 + z_2 r_2 < 2 r_1 r_2$.  If $z_1 r_1 + z_2 r_2 < r_1 r_2$, then
we are done; we can set $v_1 = z_1$ and $v_2 = z_2$.  
If not, either $z_1 > 0$ or $z_2 > 0$.  If $z_1 > 0$,
setting $v_1 = z_1 - r_2$ and $v_2 = z_2$ makes 
$z_1 r_1 + z_2 r_2 = z_1 r_1 + z_2 r_2 - r_1 r_2 = u$, and the corresponding
modification of $z_2$ works if $z_2 > 0$.
\end{proof}

Armed with Lemma~\ref{l:CRT}, we are now ready to work on the
``local'' shift-invariance of $\bT$.  The following
more general lemma will do the job.
\begin{lemma}
\label{l:key}
Let $\bX,\bY$ be independent integer random variables where $\bX$ is $\alpha$-shift-invariant at scale $1$ and $\bY$ is $\beta$-shift-invariant at scale $1$. Let $\bZ =
r_1 \cdot \bX + r_2 \cdot \bY$. Then for any positive integer $d$ we have
$\dtv(\bZ, \bZ+d) \leq r_2 \alpha + r_1 \beta
+ \min \left\{\frac{d}{r_1} \alpha, \frac{d}{r_2} \beta\right\}.$
\end{lemma}
\begin{proof}
Note that $d= s \cdot r_1 \cdot r_2 + z$ where $0 \le z < r_1\cdot r_2$, $0 \le s \le  d/(r_1 \cdot r_2) $, and $s$ is an integer.  By using Lemma~\ref{l:CRT}, we have that $d = s \cdot r_1 \cdot r_2 + v_1 \cdot r_1 + v_2 \cdot r_2$ where $|v_1|<r_2$ and $|v_2|<r_1$, so $d = (s \cdot r_2 + v_1) \cdot r_1 + v_2 \cdot r_2$. Note that
$$
|s \cdot r_2 + v_1| \le |v_1| + s \cdot r_2 \le (r_2-1) + d/r_1 .
$$
Thus, $d = t_1 \cdot r_1 + t_2 \cdot r_2$ where $t_1, t_2$ are integers and $|t_1| \le r_2-1 + d/r_1$ and $|t_2| \le (r_1-1)$.
Hence we have
\begin{eqnarray*}
\dtv(\bZ, \bZ+d) &=& \dtv(r_1 \cdot \bX + r_2 \cdot \bY , r_1 \cdot \bX + r_2 \cdot \bY + t_1 \cdot r_1 + t_2 \cdot r_2) \\
&\le& |t_1| \cdot \alpha + |t_2| \cdot \beta \\
&\leq& \left(\frac{d}{r_1} + r_2\right) \cdot  \alpha + r_1 \cdot \beta.
\end{eqnarray*}
By swapping the roles of $r_1$ and $r_2$ in the above analysis, we get the stated claim.
\end{proof}

Now we have everything we need to prove 
Lemma~\ref{l:mix} in the case that
$D = 2$.

\newcommand{\TPt}{\bV}
Let $\TPt = TP(\mu,\sigma^2)$.
Let $\bU_{d}$ denote the uniform distribution over $\{0,1,\dots,d-1\}$, where $d$ will be chosen later.  We will
bound $\dtv(\bT + \bU_{d}, \TPt + \bU_{d})$, $\dtv(\bT, \bT + \bU_{d})$,
and $\dtv(\TPt, \TPt + \bU_{d})$, and apply the triangle inequality via
\[
\dtv(\bT,\TPt)
 \leq \dtv(\bT,\bT+\bU_{d}) + \dtv(\bT+\bU_{d},\TPt)
 \leq \dtv(\bT,\bT+\bU_{d}) + \dtv(\bT+\bU_{d},\TPt+\bU_{d}) + \dtv(\TPt, \TPt+\bU_{d}).
\]

First, recalling that $\bT = r_1 \cdot \bT_1 + r_2 \cdot \bT_2$ and that (by Fact \ref{fact:good-shift-invariance})
$\bT_1$ is $O(1/\sigma(\bT_1))$-shift-invariant at scale 1 and 
$\bT_2$ is $O(1/\sigma(\bT_2))$-shift-invariant at scale 1, we have that

\begin{eqnarray*}
\dtv(\bT,\bT+\bU_{d}) &\le& \mathbf{E}_{x \sim \bU_{d}} \left[ \dtv (\bT,\bT+x)                    \right] \\
 & \le & \max_{x \in \supp(\bU_{d})}\dtv (\bT,\bT+x)  \\
& \leq & O\left( \frac{r_1}{\sigma(\bT_2)} + \frac{r_2}{\sigma(\bT_1)} + \min \left\{ \frac{d}{r_1  \sigma(\bT_1)}  , \frac{d}{r_2  \sigma(\bT_2)}  \right\} \right)
\mbox{\hspace{25pt}(by Lemma~\ref{l:key})} \\
&\leq&  O\left(\frac{r_1}{\sigma(\bT_2)} + \frac{d}{r_1 \sigma(\bT_1)} \right)
              \mbox{\hspace{25pt}(since $r_1 \sigma(\bT_1) > r_2 \sigma(\bT_2)$)} \\
&\leq& O(\epsilon') +  O\left(\frac{d}{r_1 \sigma(\bT_1)}\right),
\end{eqnarray*}
since $\sigma(\bT_2) > r_1 /\eps'$.

Next, Fact~\ref{fact:good-shift-invariance} implies that
$\TPt$ is $O\left( \frac{1}{r_1 \sigma(\bT_1)} \right)$-shift-invariant,
so repeated application of the triangle inequality gives
\[
\dtv(\TPt,\TPt+\bU_{d}) \leq O\left(\frac{d }{r_1 \sigma(\bT_1)} \right).
\]

Finally, we want to bound $\dtv(\bT + \bU_{d}, \TPt + \bU_{d})$.
Observe that $\Pr[|\bT-\mu| < \sigma/\eps']$ and $\Pr[|\TPt
- \E[\TPt]| \leq \sigma/\eps']$ are both $2^{-\poly(1/\eps')}.$ Hence
applying Lemma~\ref{lem:GMRZ} and recalling that $r_1 \sigma(\bT_1) >
r_2 \sigma(\bT_2)$, we get
\[
\dtv(\bT + \bU_{d}, \TPt + \bU_{d})
 \leq o(\epsilon') + O \left( {\sqrt{ (1/\sigma_{\min}) \cdot ( (r_1 \sigma(\bT_1))/\epsilon') \cdot (1/d) }} \right).
\]
Combining our bounds, we get that
\[
\dtv(\bT,\TPt)
\leq  O(\epsilon') + O \left( {\sqrt{ (1/\sigma_{\min})
                                    \cdot ( (r_1 \sigma(\bT_1))/\epsilon')
                                    \cdot (1/d) }}
                      + \frac{d }{r_1 \sigma(\bT_1)} \right).
\]
Taking $d = r_1 \sigma(\bT_1) /(\sigma_{\min}\epsilon')^{1/3}$, we get
\[
\dtv(\bT,\TPt)
  \leq O(\epsilon') + 1/(\sigma_{\min}\epsilon')^{1/3} = O(\epsilon')
\]
since $(1/\sigma_{\min})^2 > (1/\epsilon')^8$.

\ignore{\rnote{Changed this slightly from ``Finally, if $\bT_{\MIX'}$ is a PBD with the same expectation and variance as $\TPt$,'' since it's not immediately clear to me that there is a signed PBD that exactly matches the mean and variance.  (Probably this can be shown easily with a small calculation but I thought we might as well skip the calculation if possible.)}}
Finally, let $\bT'$ be a signed PBD that has $|\E[\bT'] - \mu|\leq 1/2$ and $\Var[\bT']=\sigma^2$.  (The existence of such a signed PBD can be shown as in (I) in the proof of Lemma~\ref{lem:dG-to-comb-signed-PBDs}.)  Lemmas~\ref{lem:translated Poisson approximation} and~
\ref{lem: variation distance between translated Poisson distributions}
imply that
$\dtv (\bT_{\MIX'}, \TPt) \leq 1/\sigma(\TPt) \leq \epsilon'$, completing the proof.
\ignore{
Finally, if $\bT_{\MIX'}$ is a PBD with the same expectation and variance as $\TPt$,
Lemma~\ref{lem:translated Poisson approximation} implies that
$\dtv (\bT', \TPt) \leq 1/\sigma \leq \epsilon'$, completing the proof.
}
\end{proof}

\ignore{
}

\newcommand{\Mix}{\mathrm{Mix}}
\newcommand{\cH}{{\cal H}}

\section{The learning result:  Learning when $|\supportset| \geq 4$} \label{sec:learn}

With the kernel-based learning results from Section~\ref{sec:kernel} and the structural results from Section~\ref{sec:structural} in hand, we are now ready to learn a distribution $\bS^*$
that is $c\epsilon$-close to a distribution
$\bS = \bS' + V$, where $\bS'$ is described in Section~\ref{sec:setup-upper-bound}.  We give two distinct learning algorithms, one for each of two mutually exclusive cases.  The overall learning algorithm works by running both algorithms and using the hypothesis selection procedure, Proposition \ref{prop:log-cover-size}, to construct one final hypothesis.

The high-level idea is as follows. In Section~\ref{sec:allsmall} we first easily handle a special case in which all the $c_{q_a}$ values are ``small,'' essentially using a brute-force algorithm which is not too inefficient since all $c_{q_a}$'s are small.  We then turn to the remaining general case, which is that some $c_{q_a}$ are large while others may be small.

The idea of how we handle this general case is as follows. First, via an analysis in the spirit of the ``Light-Heavy Experiment'' from  \cite{DDOST13}, we approximate the distribution $\bS' + V$ as a sum of two independent distributions
$\bS_{\mathrm{light}} + \bS_{\mathrm{heavy}}$
where intuitively $\bS_{\mathrm{light}}$ has ``small support'' and $\bS_{\mathrm{heavy}}$ is a 0-moded 
$\supportset$-sum
supported on elements all of which have large weight
(this is made precise in Lemma~\ref{lem:light-heavy}).
Since $\bS_{\mathrm{light}}$ has small support, it is helpful to think of 
$\bS_{\mathrm{light}} + \bS_{\mathrm{heavy}}$
 as a mixture of shifts of
$\bS_{\mathrm{heavy}}$.
We then use structural results from Section~\ref{sec:structural} to approximate this distribution in turn by a mixture of not-too-many shifts of
%
a weighted sum of signed PBDs,
whose component independent PBDs satisfy a certain technical condition on their variances (see Corollary~\ref{corr:light-heavy1}).  Finally, we exploit the kernel-based learning tools developed in Section \ref{sec:kernel} to give an efficient learning algorithm for this mixture distribution.  Very roughly speaking, the final $\log \log a_k$ sample complexity dependence (ignoring other parameters such as $\eps$ and $k$) comes from making $O(\log a_k)$ many ``guesses'' for parameters (variances) of the 
weighted sum of signed PBDs;
this many guesses suffice because of the technical condition alluded to above.

We now proceed to the actual analysis.
Let us reorder the sequence $q_1, \ldots, q_K$ so that 
$c_{q_1} \le \ldots \le c_{q_K}$. 
Let us now define the sequence $t_1, \ldots, t_K$ as $t_a= (1/\epsilon)^{2^a}$. (For intuition on the conceptual role of the $t_i$'s, the reader may find it helpful to review the discussion given in the ``Our analysis'' subsection of Section~\ref{sec:Aequalsk}.) Define the ``largeness index" of the sequence $c_{q_1} \le \ldots \le c_{q_K}$ as the minimum $\ell \in [K]$ such that $c_{q_\ell} > t_\ell $, and let $\ell_0$ denote this value.   If there is no $\ell \in [K]$ such that  $c_{q_\ell} > t_\ell$, then we set $\ell_0  = K+1$. 

We first deal with the easy special case that $\ell_0 = K+1$ and then turn to the main case.  

\subsection{Learning when $\ell_0 = K+1$} \label{sec:allsmall}
Intuitively, in this case all of $c_{q_1},\dots,c_{q_K}$ are ``not too large'' and we can learn via brute force.  More precisely,
since each $c_{q_a} \leq 1/\eps^{2^K}$, in a draw from $\bS'$ the expected number of random variables $\bX'_1,\dots,\bX'_N$ that take a nonzero value is at most $K/\eps^{2^K}$, and a Chernoff bound implies that in a draw from 
$\bS'$ we have $\Pr[$more than $\poly(K/\eps^{2^K})$ of the $\bX'_i$'s take a nonzero value$] \leq \eps.$  Note that for any $M$, there are at most $M^{O(K)}$ possible outcomes for $\bS = \bS' + V$ that correspond to having at most $M$ of the $\bX'_i$'s take a nonzero value. Thus it follows that in this case the random variable $\bS'$ (and hence $\bS$) is $\eps$-essentially supported on a set of size at most $M^{O(K)} = (K/\eps^{2^K})^{O(K)}.$  
Thus $\bS^*$ is $O(\epsilon)$-essentially supported on
a set of the same size.
Hence the algorithm of Fact \ref{fact:learn-sparse-ess-support} can be used to learn $\bS^*$ to accuracy $O(\eps)$ in time $\poly(1/\eps^{O(2^K)})= \poly(1/\eps^{2^{O(k^2)}}).$

\ignore{
\begin{proposition}
If $\ell_0 = K+1$, then there is a set $A$ of size $M = (O(t_K))^{K}= O(1/\eps^{K2^K})$ such that $\bS'$ is $\eps$-essentially-supported on $A$.
\end{proposition}
\begin{proof}
For every $\bX_i$, we define the indicator random variable $\bW_i$ such that $\bW_i =1$ if and only if $\bX_i \not =0$. Note that for $\bW$ defined as $\bW = \bW_1 + \ldots  + \bW_n$, both $\mathbf{E}[\bW]$ and $\Var[\bW]$ is bounded by $O(t_1 + \ldots + t_K) = O(t_K)$. 
By applying the Chernoff-Hoeffding bound, 
Thus, 
$\Pr[\bW \le t_K + O(\sqrt{t_K \cdot \log (1/\epsilon)})] \geq 1- \epsilon$.  Let us call this event  as $E$.  Define $M$ as the set $$ A = \big\{ \sum_{a=1}^K \alpha_a \cdot q_a : 0 \le \alpha_a \le O(t_K) \big\}. $$
Note that $|A| = M = O(t_K)^K$. Note that if $E$ occurs, then $\bS' \in A$. As $\Pr[E] \ge 1- \epsilon$, this concludes the proof.  
%
\end{proof}
}

\subsection{Learning when $\ell_0 \le K$.} \label{sec:notallsmall}
Now we turn to the main case, which is when $\ell_0 \le K$.  The following lemma is an important component of our analysis of this case. Roughly speaking, it says that $\bS'$ is close to a sum of two independent random variables, one of which ($\bS_{\mathrm{light}}$) has small support and the other of which ($\bS_{\mathrm{heavy}}$) is 
the sum of $0$-moded random variables
that all have large weight.

\begin{lemma}~\label{lem:light-heavy}
Suppose that $\ell_0 \le K$. Then there exists $\tilde{\bS} = \bS_{\mathrm{heavy}} + \bS_{\mathrm{light}}$ such that $\dtv(\tilde{\bS}, \bS') \le O(\epsilon)$ and the following hold:

\begin{enumerate}

\item $\bS_{\mathrm{heavy}}$ and  $\bS_{\mathrm{light}}$ are independent of each other;

\item The random variable $\bS_{\mathrm{light}}$ is $\bS_{\mathrm{light}} = \sum_{1\leq b<\ell_0} q_b \cdot \bS_{b}$ where for each $1 \le b < \ell_0$, $\bS_b$ is supported
on the set 
$[-(1/\epsilon) \cdot t_{\mathrm{cutoff}}, (1/\epsilon) \cdot t_{\mathrm{cutoff}}] \cap \mathbb{Z}$ where  $t_{\mathrm{cutoff}} =  (t_1 + \ldots + t_{\ell_0 - 1})$ and the $\{\bS_b\}$ are not necessarily independent of each other;

\item The random variable $\bS_{\mathrm{heavy}}$ is 
the sum of $0$-moded random variables supported in
$\{0, \pm q_{\ell_0},$ $ \ldots,$ $\pm q_{K}\}$. Further, for all $b \ge \ell_0$, we have 
 $c_{q_b,\mathrm{heavy}}> \frac{t_{\ell_0}}{2}$ where $c_{q_b,\mathrm{heavy}}$ is defined as in Section~\ref{sec:setup-upper-bound} but now with respect to $\bS_{\mathrm{heavy}}$ rather than with respect to $\bS'$.
\end{enumerate}
\end{lemma}

\begin{proof}
The proof follows the general lines of the proof of Theorem 4.3 of \cite{DDOST13}.  
Let ${\cal L} = \{\pm q_1, \ldots, \pm q_{\ell_0-1}\}$ 
and ${\cal H} =\{0, \pm q_{\ell_0}, \ldots, \pm q_{K}\}.$  
(It may be helpful to think of ${\cal L}$ as the ``light'' integers,
and ${\cal H}$ as ``heavy'' ones.)
We recall the following experiment that can be used to make a draw from $\bS'$, referred to in  \cite{DDOST13} as the ``Light-Heavy Experiment'':

\begin{enumerate}
\item {[Stage 1]:}  Informally, sample from the conditional
distributions given membership in ${\cal L}$.
Specifically, independently we sample for each $i \in [N]$ a random variable $\ul{\bX}'_i \in {\cal L}$ as follows:
$$\text{for each $b \in \calL$,} \quad \ul{\bX}'_i = b, \text{with probability~}
{\frac {\Pr[\bX'_i = b]} {\Pr[\bX'_i \in {\cal L}]}};
$$
i.e.\ $\ul{\bX}'_i$ is distributed according to the conditional distribution of $\bX'_i$, conditioning on $\bX'_i \in {\cal L}$. In the  case that $\Pr[\bX'_i \in {\cal L}]=0$ we define $\ul{\bX}'_i=0$ with probability $1$.

\item {[Stage 2]:}  Sample analogously for $\cal H$.
Independently we sample for each $i \in [N]$ a random variable $\ol{\bX}'_i \in {\cal H}$ as follows:
$$
\text{for each $b \in  \calH$,} \quad
\ol{\bX}'_i =  b , \text{with probability~}
{\frac {\Pr[\bX'_i = b]} {\Pr[\bX'_i  \in {\cal H}]}};
$$
i.e.\ $\ol{\bX}'_i$ is distributed according to the conditional distribution of $\bX'_i$, conditioning on $\bX'_i  \in {\cal H}$.
\item {[Stage 3]:} Choose which $\bX_i'$ take values in
${\cal L}$: sample a random subset $\bL \subseteq [N]$, by independently including each $i$ into $\bL$ with probability $\Pr[\bX'_i \in {\cal L}]$.
\end{enumerate}

After these three stages we output $\sum_{i \in \bL}\ul{\bX}'_i + \sum_{i \notin \bL}\ol{\bX}'_i$ as a sample from $\bS'$, where $\sum_{i \in \bL}\ul{\bX}'_i$ represents ``the contribution of ${\cal L}$'' and $\sum_{i \notin \bL}\ol{\bX}'_i$  ``the contribution of ${\cal H}$.'' 
Roughly, Stages 1 and 2 provide light and heavy options
for each $\bX'_i$, and Stage 3 chooses among the options.
We note that the two contributions are not independent, but they are independent conditioned on the outcome of~$\bL$.   Thus we may view a draw of $\bS'$ as a mixture, over all possible outcomes $L$ of $\bL$, of the distributions $\sum_{i \in L}\ul{\bX}'_i + \sum_{i \notin L}\ol{\bX}'_i$;
i.e. we have $\bS' = \Mix_{L \leftarrow \bL}(\sum_{i \in L}\ul{\bX}'_i + \sum_{i \notin L}\ol{\bX}'_i).$
This concludes the definition of the Light-Heavy Experiment.

Let $t_{\mathrm{cutoff}}= t_1 + \ldots +t_{\ell_0-1}$. Note that $\E[|\bL|] \leq t_{\mathrm{cutoff}}$. Let $\BAD$ denote the set of all outcomes $L$ of $\bL$ such that $|L| > (1/\eps) \cdot  t_{\mathrm{cutoff}}.$  A standard application of the Hoeffding bound  implies that $\Pr[\bL \in \BAD] = \Pr[|\bL| > (1/\eps) \cdot t_{\mathrm{cutoff}}] \leq 2^{-\Omega(1/\eps)}.$
  It follows that if we define the distribution $\bS''$ to be an outcome of the Light-Heavy Experiment conditioned on $\bL \notin \BAD$, i.e.\ $\bS'' = \Mix_{L \leftarrow \bL \ | \ \bL \notin \BAD}(\sum_{i \in L}\ul{\bX}'_i + \sum_{i \notin L}\ol{\bX}'_i)$,  we have that $\dtv(\bS'',\bS') \leq 2^{-\Omega(1/\eps)}$.  Consequently it suffices to show the existence of $\tilde{\bS}$ satisfying the properties of the lemma such that $\dtv(\tilde{\bS}, \bS'') \le \epsilon$. 
  
We will now show that for any $L_1, L_2 \notin \BAD$,  the random variables
$\bS_{L_j} = \sum_{i \not \in L_j} \ol{\bX}'_i$ (for $j \in \{1,2\}$) are close to each other in total variation distance. 
(If we think of $L_1$ and $L_2$ as different possibilities for the 
final
step in the process of sampling from the distribution of $\bS'$, 
recall that the values of $\ol{\bX}'_i$ are {\em always} in $\cH$ -- loosely speaking, during the first sample from
$\bS'$ the values of $\ol{\bX}'_i$ for $i \in L_1$ are not used, and during the second sample, 
the values for $i \in L_2$ are not used.)
Let $L_{\mathrm{union}} = L_1 \cup L_2$.\ignore{ and let 
$L_{\mathrm{common}} = [n] \setminus L_{\mathrm{union}}
$. }
 Note that by definition 
 $$|L_2 \setminus L_1|, \ |L_1 \setminus L_2| \le (1/\epsilon) \cdot t_{\mathrm{cutoff}}. $$
Define
\[
\bS_{L_{\mathrm{union}}}
  =  \sum_{i \notin L_{\mathrm{union}}}\ol{\bX}'_i
   \;\;=\;\; \bS_{L_1} - \sum_{i \in L_2 - L_1} \ol{\bX}'_i
         \;\; = \;\; \bS_{L_2} - \sum_{i \in L_1 - L_2} \ol{\bX}'_i.
\]

Choose $b \geq \ell_0$.
We have that
\begin{align*}
\dshift{q_b}(\ol{\bX}'_i) 
 &= 1 - \sum_j (\min\{\Pr[\ol{\bX}'_i = j],\Pr[\ol{\bX}'_i=j+q_b\})\\
 &  \leq 1 - \min\{\Pr[\ol{\bX}'_i = 0],\Pr[\ol{\bX}'_i=q_b] \} - \min\{\Pr[\ol{\bX}'_i = -q_b],\Pr[\ol{\bX}'_i=0] \} \\
 & = 1 - \Pr[\ol{\bX}'_i=-q_b] - \Pr[\ol{\bX'}_i=q_b],
\end{align*}
since $\ol{\bX'}_i$ is $0$-moded.
By Corollary~\ref{corr:pshift}, this implies that
$$
\dshift{q_b}(\bS_{L_{\mathrm{union}}}) \le \frac{O(1)}{\sqrt{\sum_{i \notin L_{\mathrm{union}}}     \Pr[\ol{\bX}'_i = \pm q_b]   }} \le \frac{O(1)}{\sqrt{c_{q_b} - |L_{\mathrm{union}}|}} \le \sqrt{\frac{2}{t_{\ell_0}}}. 
$$
Here the penultimate inequality uses the fact that 
$$
\sum_{i \notin L_{\mathrm{union}}} \Pr[\ol{\bX}'_i = \pm q_b]  = \sum_{i \in [n]} \Pr[\ol{\bX}'_i = \pm q_b] - \sum_{i \in L_{\mathrm{union}}} \Pr[\ol{\bX}'_i = \pm q_b]  \ge c_{q_b} - |L_{\mathrm{union}}|. 
$$
The last inequality uses that
$$c_{q_b} - |L_{\mathrm{union}}| \ge c_{q_b}-  |L_1| - |L_2| \ge c_{q_b} - 2 \cdot (1/\epsilon) \cdot t_{\mathrm{cutoff}}
 \ge t_b - 2 \cdot (1/\epsilon) \cdot t_{\mathrm{cutoff}} \ge \frac{t_{\ell_0}}{2}. $$
As each of the summands in the sum $\sum_{i \in L_2 \setminus L_1} \ol{\bX}'_i$ is supported on the set $\{0, \pm q_{\ell_0}, \ldots, \pm q_K\}$, viewing $\bS_{L_1}$ as a mixture of distributions each of which is obtained by shifting $\bS_{L_\mathrm{union}}$ at most $|L_2 \setminus L_1|$ many times, each time by an element of $\{0,\pm q_{\ell_0},\dots,q_K\}$,
we immediately obtain that
$$
\dtv(\bS_{L_{\mathrm{union}}}, \bS_{L_1}) \le |L_2 \setminus L_1| \cdot \sqrt{\frac{2}{t_{\ell_0}}} \le 2 (1/\epsilon) \cdot  t_{\mathrm{cutoff}} \cdot  \sqrt{\frac{2}{t_{\ell_0}}} \le O(\epsilon).
$$

Choose any $L^{\ast} \notin \BAD$ arbitrarily, and define $\bS_{\mathrm{heavy}} := \sum_{i \not \in L^{\ast}} \ol{\bX}'_i$. By the above analysis, for any $L' \notin \BAD $ it holds that
$\dtv(\sum_{i \not \in L'} \ol{\bX}'_i, \bS_{\mathrm{heavy}} ) = O(\epsilon)$. 
Thus, for any outcome $L \notin \BAD$, we have
$\dtv(\sum_{i \in L}\ul{\bX}'_i + \sum_{i \notin L}\ol{\bX}'_i, \sum_{i \in L}\ul{\bX}'_i + \bS_{\mathrm{heavy}}) = O(\epsilon)$. Define $\bS_{\mathrm{light}}
:= 
\Mix_{L \leftarrow \bL \ | \ \bL \notin \BAD}\sum_{i \in L}\ul{\bX}'_i.
$

We now verify that the above-defined $\bS_{\mathrm{heavy}}$ and $\bS_{\mathrm{light}}$ indeed satisfies the claimed properties.
Note that
 for each $L  \notin \BAD$, $\sum_{i \in L}\ul{\bX}'_i $ is supported on the set 
 $$\left\{\sum_{b \le \ell_0} q_b \cdot S_b : S_b \in [-(1/\epsilon) \cdot t_{\mathrm{cutoff}}, (1/\epsilon) \cdot t_{\mathrm{cutoff}}] \right\},$$
so the second property holds as well.  
Likewise, $\bS_{\mathrm{heavy}}$ is a 
sum of $0$-moded random variables with
support in $\{0, \pm q_{\ell_0}, \ldots, \pm q_K\}$. 
Note that we have already shown that 
 $\sum_{i \notin L^{\ast}} \Pr[\overline{\bX}'_i  = \pm q_b]  \ge t_{\ell_0}/2$
 , giving the third property.
Finally, combining the fact that $\Pr[\bL \in \BAD] \le \epsilon$ with $\dtv(\sum_{i \in L}\ul{\bX}'_i + \sum_{i \notin L}\ol{\bX}_i, \sum_{i \in L}\ul{\bX}_i + \bS_{\mathrm{heavy}}) = O(\epsilon)$, we obtain the claimed variation distance bound $\dtv(\tilde{\bS}, \bS') \le O(\epsilon)$, finishing the proof. 
 \end{proof}

With Lemma~\ref{lem:light-heavy} in hand, we now apply Lemma~\ref{lem:close-to-pure-signed-sicsirv} to the distribution $\bS_{\mathrm{heavy}}$ with 
\begin{equation}
\label{eq:setbig}
\BIG = K^{25}/\eps^{32}.
\end{equation}  This gives the following corollary: 
\begin{corollary}~\label{corr:light-structure}
The distribution $\bS'$ is $\delta$-close in total variation distance to a distribution $\bS''= \bS'_{\mathrm{light}} + \sum_{a=\ell_0}^{K} q_a \cdot \bW_{a,a} + \sum_{(a,b) \in A} (q_b + \tau_{a,b} \cdot q_a) \bW_{a,b}$ where $\delta  = O(K^{71/20} \cdot \epsilon^{1/20})$ and
\begin{enumerate}
\item $A \subseteq \{(a,b): \ell_0 \leq a < b \leq K\}$, $\tau_{a,b} \in \{-1,1\}$, and $\bW_{a,a}$, $\bW_{a,b}$ are\ignore{{pure}} signed PBDs. 
\item $\Var[\bW_{a,a}]$, $\Var[\bW_{a,b}] \ge (\BIG/K)^{-1/4} >  K^6/\eps^8$, \ignore{and $\dtv(\bS'', \bS') = O(K^{71/20} \cdot \epsilon^{1/20})$.}
\item The random variables $\bS'_{\mathrm{light}}$, $\bW_{a,a}$ and $\bW_{a,b}$ are all independent of each other, and
\item $\bS'_{\mathrm{light}}$ is supported on a set of cardinality $M \le (2t_K/\epsilon)^{K}$. 
\end{enumerate}
\end{corollary}
\begin{proof}
By Lemma~\ref{lem:light-heavy} $\bS'$ is $O(\epsilon)$-close to  $\tilde{\bS} = \bS_{\mathrm{heavy}} + \bS_{\mathrm{light}}$ where the decomposition of $\tilde{\bS}$  is as described in that lemma.  Applying Lemma~\ref{lem:close-to-pure-signed-sicsirv}, we further obtain that $\bS_{\mathrm{heavy}}$ is $\delta = O(K^{71/20} \cdot \epsilon^{1/20})$-close to a distribution of the form 
$$
\sum_{a=\ell_0}^{K} q_a \cdot \bW_{a,a} + \sum_{(a,b) \in A} (q_b + \tau_{a,b} \cdot q_a) \bW_{a,b} + V',
$$
with $\bW_{a,a}$ and $\bW_{a,b}$ satisfies the conditions stated in the corollary. Defining  $\bS'_{\mathrm{light}} = V + \bS_{\mathrm{light}}$ will satisfy all the required conditions. We note that the size of the support of $\bS'_{\mathrm{light}}$ is the same as the size of the support of $\bS_{\mathrm{light}}$, so applying item (2) of Lemma~\ref{lem:light-heavy}, we get that the size of the support of $\bS'_{\mathrm{light}}$ is bounded by $(2t_K/\epsilon)^{K}$.
\end{proof}

Let us look at the structure of the distribution 
\[
\bS'' = \bS'_{\mathrm{light}} + \sum_{a=\ell_0}^{K} q_a \cdot \bW_{a,a} + \sum_{(a,b) \in A} (q_b + \tau_{a,b} \cdot q_a) \bW_{a,b}.
\]
For $(a,b) \in A$ let $q_{(a,b)}$ denote $q_b + \tau_{a,b} \cdot q_a$, and let $B$ denote the set 
$B=\{\ell_0,\dots,K\} \cup A$.  
For any $B' \subseteq B$, let us write $\mathsf{gcd}(B')$ to denote $\gcd(\{q_\alpha\}_{\alpha \in B}).$  
Let $\alpha^\ast$ denote the element of $B$ for which $\Var[q_{\alpha^\ast} \cdot \bW_{\alpha^\ast}]$ is largest (breaking ties arbitrarily) and let $\MIX$ denote the following subset of $B$:
\begin{equation} \label{eq:mixdef}
\MIX = \{\alpha^\ast\} \cup \{\alpha \in B: \Var[\bW_\alpha] \geq \max \{1/\eps^8, q_{\alpha^\ast}^2/\eps^2 \}\}.
\end{equation}
By applying Lemma~\ref{l:mix} to the distribution $\sum_{\alpha \in \MIX} q_\alpha \cdot \bW_{\alpha}$, with its $\sigma_{\min}^2$ set to $K^6/\eps^8$ and its $\eps'$ set to $\eps$, noting that $|B| \eps' = O(K^2 \eps) = o(K^{71/20} \cdot \epsilon^{1/20})$ we obtain the following corollary: 
%
%
\begin{corollary}~\label{corr:light-heavy1}
The distribution $\bS'$ is $\delta'= O(K^{71/20} \cdot \epsilon^{1/20})$-close in total variation distance to a distribution $\bS^{(2)}$ of the following form, where $\emptyset \subsetneq \MIX \subset B$ is as defined in (\ref{eq:mixdef}):
\[
\bS^{(2)}= \bS'_{\mathrm{light}} + q_{\MIX} \cdot  \bS_{\MIX} + \sum_{q_{{\alpha}} \in B \setminus \MIX} q_{{\alpha}} \cdot \bS_{{\alpha}},
\]
where $q_{\MIX} = \mathsf{gcd}( \MIX)$ and the following properties hold:
\begin{enumerate}

\item The random variables $\bS'_{\mathrm{light}}$, $\bS_{\MIX} $ and $\{\bS_{\alpha} \}_{q_{\alpha} \in B \setminus \MIX}$ are independent of each other. 

\item $\bS'_{\mathrm{light}}$ is supported on a set of at most $M$ integers, where 
$M \le (2t_K/\epsilon)^{K}$.

\item $\bS_{\MIX} $ and $\{\bS_{\alpha} \}_{q_{\alpha} \in B \setminus \MIX}$ are\ignore{{pure}} signed PBDs such that\ignore{\rnote{Was ${\frac 1 {K'}} \cdot \Var[\bS^{(2)}] \le \Var[\bS_{\MIX}] \le  \Var[\bS^{(2)}]$, but I don't see that we have this}}\ignore{$\frac{1}{{K}} \cdot \Var[\bS^{(2)}] \le \Var[\bS_{\MIX}] \le  \Var[\bS^{(2)}]$ and} 
for all $q_{\alpha} \in B \setminus \MIX$, we have $K^6/\eps^8 \leq \Var[S_\alpha] \leq  r^2/\eps^2$,
where $r = \max_{q_\alpha \in B} |q_\alpha|$.  Moreover $\Var[\bS_\MIX] \geq K^6/\eps^8.$

\item $ \Var[q_{\MIX} \cdot \bS_{\MIX}] = c \cdot \left(\Var[q_{\MIX} \cdot  \bS_{\MIX}] + \sum_{q_\alpha \in B \setminus \MIX} \Var[q_\alpha \cdot \bS_\alpha]
\right)$ for some $c \in [{\frac 1 {K^2}},1].$

\end{enumerate}
\end{corollary}

The above corollary tells us that our distribution $\bS'$ is close to a ``nicely structured'' distribution $\bS^{(2)}$; we are now ready for our main learning result, which uses kernel-based tools developed in Section~\ref{sec:kernel} to learn such a distribution.  The following theorem completes the $\ell_0 \leq K$ case:

\begin{theorem} \label{thm:learn-some-large}
There is a learning algorithm and a positive constant $c$ with the following properties:  It is given as input $N$, values $\eps,\delta>0$, and integers $0 \leq a_1 < \cdots < a_k$, and can access draws from an unknown distribution
$\bS^*$ that is $c \epsilon$-close to a
$\{a_1,\dots,a_k\}$-sum $\bS$.
The algorithm runs in time
$(1/ \eps) ^{2^{O(k^2)}} \cdot (\log a_k)^{\poly(k)}$
and uses 
$(1/ \eps) ^{2^{O(k^2)}} \cdot \log \log a_k$
samples, and has the following property:  Suppose that for the zero-moded 
distribution
$\bS'$ such that $\bS' + V= \bS$ (as defined in Section~\ref{sec:setup-upper-bound}), the largeness index $\ell_0$ (as defined at the beginning of this section) is at most $K$ (again recall Section~\ref{sec:setup-upper-bound}).  Then with probability $1-o(1)$ the algorithm outputs a hypothesis distribution $\bH$ with $\dtv(\bH,\bS) \leq O(K^{71/20} \cdot \epsilon^{1/20})$.\ignore{\rnote{Is it worth even mentioning, maybe in a parenthetical between the theorem and its proof, that if you want to learn to accuracy $\eps$ you (obviously) just run it to learn to $\eps' = (\eps/K^{71/20})^{20}$-accuracy?  Or don't even bother saying this?  (Doing this doesn't change any of the claimed asymptotic runtimes or sample complexities, I think, because of how coarse our claimed bounds are.)}\pnote{It might be nice to state the theorem with $\epsilon$ accuracy, and then remark at the beginning
of the proof that the statement follows from the $\eps'$ result that is more convenient to prove.  I'm OK with leaving it as is also, with or without the
parenthetical.}}
\end{theorem}

(To obtain an $O(\eps)$-accurate hypothesis, simply run the learning algorithm with its accuracy parameter set to $\eps' = \eps^{20}/K^{71}$.)

\ignore{
\gray{
\begin{theorem}
There is a learning algorithm with the following properties: It is given as input $N$, values $\epsilon, \delta>0$, and integers $q_1, \ldots, q_K$. Further, assume that there is a distribution $\bS^{(2)}$ with the structure as in Corollary~\ref{corr:light-heavy1} and that the algorithm has access to draws from a distribution $\bS'$ such that $\dtv(\bS', \bS^{(2)})  = O(K^{71/20} \cdot \epsilon^{1/20})$. The algorithm runs in time BLAH, uses BLAH samples, and succeeds with probability 
$1 - O(K^{71/20} \cdot \epsilon^{1/20}+ \delta)$. 
\end{theorem}
}
}

\begin{proof}
The high level idea of the algorithm is as follows:  The algorithm {repeatedly} samples
two points from the distribution $\bS^*$
and, for each pair, uses those two points to guess (approximately) parameters of the distribution 
\[
\bS_{\mathrm{pure}} := q_{\MIX} \cdot  \bS_{\MIX} + \sum_{q_\alpha \in B \setminus \MIX} q_\alpha \cdot \bS_\alpha
\] 
from Corollary~\ref{corr:light-heavy1}.
The space of possible guesses will be of size 
$(1/\eps)^{2^{O(k^2)}} \cdot (\log a_k)^{\poly(k)}$, 
which leads to a $\poly(2^{k^2},\log(1/\eps)) \cdot \log \log a_k$ factor in the sample complexity by Corollary~\ref{cor:guess}.  For each choice 
of parameters in this space, Lemma~\ref{lem:siirv-kernel} allows us to produce a candidate hypothesis distribution  (this lemma leads to a 
$\exp(\poly(k))/ \eps^{2^{O(k)}}$ 
factor in the sample complexity); by the guarantee of Lemma~\ref{lem:siirv-kernel}, for the (approximately) correct choice of parameters the corresponding candidate hypothesis distribution will be close to the target distribution $\bS'$.  Given that there is a high-accuracy candidate hypothesis distribution in the pool of candidates, by Corollary \ref{cor:guess} (which details how our algorithms can ``make guesses''), the algorithm of that corollary will with high probability select a high-accuracy hypothesis distribution  $\bH$ from the space of candidates.

We now give the detailed proof.  To begin, the algorithm computes $K$ and the values $q_1,\dots,q_K$.  It guesses an ordering of $q_1,\dots,q_K$ such that $c_{q_1} \leq \cdots \leq c_{q_K}$ ($K! = 2^{\poly(k)}$ possibilities), guesses the value of the largeness index $\ell_0$ ($O(K)=\poly(k)$ possibilities), guesses the subset $A \subseteq \{(a,b): \ell_0 \leq a < b \leq K\}$ and the associated bits $(\tau_{a,b})_{(a,b) \in A}$ from Corollary~\ref{corr:light-structure} ($2^{\poly(k)}$ possibilities) , and  guesses the subset $\MIX \subseteq B$ from Corollary~\ref{corr:light-heavy1} ($2^{\poly(k)}$ possibilities).
The main portion of the algorithm consists of the following three steps:

\medskip \noindent 
\textbf{First main step of the algorithm: Estimating the variance of $\bS^{(2)}.$}  In the first main step, the algorithm constructs a space of $1/\eps^{2^{O(k^2)}}$ many guesses, one of which with very high probability is a multiplicatively accurate approximation of  $\sqrt{\Var[\bS_{\mathrm{pure}}]}$. This is done as follows:  the algorithm makes two independent draws from
$\bS^*$.
Since $\bS^*$ is $c \epsilon$-close to 
$\bS = \bS' + V$,
by Corollary~\ref{corr:light-heavy1} and Lemma~\ref{l:coupling}, 
the distribution over these two draws could be obtained by sampling twice independently from $\bS^{(2)}$,
and modifying the result with probability $O(K^{71/20} \cdot \epsilon^{1/20})$.
Let us write these two draws as $s^{(j)}=s^{(j)}_{\mathrm{light}} + s^{(j)}_{\mathrm{pure}}$ where $j \in \{1,2\}$ and $s^{(j)}_{\mathrm{light}} \sim \bS_{\mathrm{light}}$ and  $s^{(j)}_{\mathrm{pure}} \sim \bS_{\mathrm{pure}}$ (where $s^{(1)}_{\mathrm{light}}$,  $s^{(1)}_{\mathrm{pure}}$, $s^{(2)}_{\mathrm{light}}$,  $s^{(2)}_{\mathrm{pure}}$ are all independent draws). 
By part (2) of Corollary~\ref{corr:light-heavy1}, with probability at least $1 /|\bS_{\mathrm{light}}| \geq 1/M \geq
(\eps/2t_K)^{K} = \eps^{2^{O(k^2)}}$, it is the case that $s^{(1)}_{\mathrm{light}}=s^{(2)}_{\mathrm{light}}$. 
In that event, with probability at least 
 $1/2^{\poly(K)}$, we have
\begin{equation}
\label{eq:factor-of-two}
\frac12 \cdot  \sqrt{\Var[\bS_{\mathrm{pure}}]}\le          |s^{(2)} - s^{(1)} |  \le 2 \cdot \sqrt{\Var[\bS_{\mathrm{pure}}]}.
\end{equation}
To see this, observe that since each of the $O(K^2)$ independent constituent PBDs comprising $\bS_{\mathrm{pure}}$ has variance at least $K^6$, for each one with probability at least ${\frac 1 {\Theta(K^2)}}$ the difference between two independent draws will lie between $(1 - {\frac 1 {\Theta(K^2)}})$ and $(1 + {\frac 1 {\Theta(K^2)}})$ times the square root of its variance.  If this happens then we get (\ref{eq:factor-of-two}).  By repeating $2^{\poly(k)} \cdot \eps^{2^{O(k^2)}}$ times, the algorithm can obtain $2^{\poly(k)} /\eps^{2^{O(k^2)}}$ many guesses, one of which will, with overwhelmingly high probability, be a quantity $\gamma_{\mathrm{pure}}$ that is a multiplicative $2$-approximation of $\sqrt{\Var[\bS_{\mathrm{pure}}]}$.

\medskip \noindent 
\textbf{Second main step of the algorithm: Gridding in order to approximate variances.} Consider the set  $J$ defined as 
$$  
J =  \bigcup_{j=-1}^{1+\log({K})} \{2^j \cdot \gamma_{\mathrm{pure}}/q_\MIX\}.
$$
Given that $\gamma_{\mathrm{pure}}$ is within a factor of two of $\sqrt{\Var[\bS_{\mathrm{pure}}]}$ (by (\ref{eq:factor-of-two})) and given part (4) of Corollary~\ref{corr:light-heavy1}, it is easy to see that there is an element $\gamma_{\MIX}  \in J$ such that 
$\gamma_{\MIX}$ is within a multiplicative factor of 2 of $\sqrt{\Var[\bS_{\MIX}]}$. Likewise, for each $\alpha \in B \setminus \MIX$, define the set $J_{q_\alpha}$ as
\[
J_{q_\alpha} = \bigcup_{j=-1}^{1+\log (\sqrt{\max\{1/\eps^8,r^2/\epsilon^2}\})} \{ 2^j \cdot (\eps \cdot K)^{-1/4}/q_\alpha \},
\]
where, as in Corollary~\ref{corr:light-heavy1}, $r = \max_{q_\alpha \in B} |q_\alpha|.$  By part (3) of Corollary~\ref{corr:light-heavy1}, for each $q_\alpha \in B \setminus \MIX$ there is an element $\gamma_{q_\alpha} \in J_{q_\alpha}$ such that $\gamma_{q_\alpha}$ is within a multiplicative factor of two of $\Var[\bS_{q_\alpha}]$.  These elements of $J$ and of $J_{q_\alpha}$ are the guesses for the values of $\sqrt{\Var[\bS_{\MIX}]}$ and of $\sqrt{\Var[\bS_{q_{\alpha}}]}$ that are used in the final main step described below.  We note that the space of possible guesses here is of size at most $O(\log k) \cdot (\log (a_k/\eps))^{\poly(k)}$.

\medskip \noindent 
\textbf{Third main step of the algorithm:  Using guesses for the variances to run the kernel-based learning approach.} For each  outcome of the guesses described above (denote a particular such outcome by $\overline{\gamma}$; note that a particular outcome for $\overline{\gamma}$ comprises an element of $J$ and an element of $J_{q_\alpha}$ for each $\alpha \in B \setminus \MIX$), let us define the distribution $\bZ_{\MIX,\overline{\gamma}}$ to be uniform on the set $[-(c \epsilon \cdot \gamma_{\MIX})/K,(c \epsilon \cdot \gamma_{\MIX})/K] \cap \Z$ and $\bZ_{q_\alpha,\overline{\gamma}}$ to be uniform on the set $[-(c\epsilon \cdot \gamma_{q_\alpha})/K,(c\epsilon \cdot \gamma_{q_\alpha})/K] \cap \Z$, where $c$ is the hidden constant in the definition of $c_j$ in Lemma~\ref{lem:siirv-kernel}.  Applying Lemma~\ref{lem:siirv-kernel},  we can draw
$\frac{\exp(\poly(K))}{\epsilon^{\poly(K)}} \cdot m^2 \cdot \log (m/\delta)$ 
samples from $\bS$, where $m =(1/\eps)^{2^{O(k^2)}} \geq  (2t_K/\eps)^{K} \geq |\bS_{\mathrm{light}}|$, and we get a hypothesis $\bH_{\overline{\gamma}}$ resulting from this outcome of the guesses and this draw of samples from $\bS.$\ignore{  outcome $ \overline{\gamma} \in  J \times_{q_\alpha \in B \setminus \MIX} J_{q_\alpha}$, we get $\mathbf{H}_{\overline{\gamma}}$.} The guarantee of Lemma~\ref{lem:siirv-kernel} ensures that for the outcome $\overline{\gamma}$ all of whose components are factor-of-two accurate as ensured in the previous step, the resulting hypothesis $\mathbf{H}_{\overline{\gamma}}$ satisfies $\dtv(\mathbf{H}_{\overline{\gamma}} , \bS') \le O(K^{71/20} \cdot \epsilon^{1/20} + \epsilon)= O(K^{71/20} \cdot \epsilon^{1/20})$ with probability at least $1-\delta.$  Finally, an application of  Corollary \ref{cor:guess} concludes the proof.
\end{proof}

\section{Learning $\{ a_1, a_2, a_3 \}$-sums} \label{sec:learn3}

In this section we show that
when $|\supportset| = 3$
the learning algorithm can be sharpened to have no dependence on $a_1,a_2,a_3$ at all.  Recall Theorem~\ref{known-k-is-three-upper}:

\medskip
\noindent {\bf Theorem~\ref{known-k-is-three-upper}} (Learning when
$|\supportset|=3$ with known support){\bf .}
\emph{There is an algorithm and a positive constant $c$ with the following properties:  The algorithm is given $N$, an accuracy parameter $\eps>0$, distinct values $a_1 < a_2 < a_3 \in \Z_{\geq 0}$, and access to i.i.d. draws from an unknown 
random variable $\bS^*$  that is $c \epsilon$-close to
an $\{a_1,a_2,a_3\}$-sum $\bS$.
The algorithm uses $\poly(1/\eps)$ draws from $\bS^*$, runs in $poly(1/\eps)$ time, and with probability at least $9/10$ outputs a concise representation of a hypothesis distribution $\bH$ such that $\dtv(\bH,\bS^*) \leq \eps.$}
\medskip

The high-level approach we take follows the approach for general $k$; as in the general case, a sequence of transformations will be used to get from the initial 
target
to a ``nicer'' distribution (whose exact form depends on the precise value of the ``largeness index'') which we learn using the kernel-based approach.\ignore{\red{If the largeness index is 3 \red{or larger} then as before we will use a kernel-based approach as the heart of our learning algorithm to learn the ``nicer'' distribution; in other cases we will take a more direct approach that bypasses the kernel approach.}\rnote{Update this, it's no longer true, right?}}
(Lemmas~\ref{lem:learn-sum-two-high-variance-PBDs} and 
\ref{lem:learn-sum-3-high-variance-PBDs}, which establish learning results for distributions in two of these nicer forms, are deferred to later subsections.)
  Intuitively, the key to our improved independent-of-$a_3$ bound is a delicate analysis that carefully exploits extra additive structure that is present when $k=3$, and which lets us 
avoid the ``gridding'' over $O(\log a_k)$ many multiplicatively spaced guesses for variances that led to our $\log \log a_k$ dependence in the general-$k$ case.

To describe this additive structure, let us revisit the framework established in Section~\ref{sec:setup-upper-bound}, now specializing to the case $k=3$, so $\bS$ is an $\{ a_1, a_2, a_3 \}$-sum with $a_1 < a_2 < a_3.$  We now have that for each $i \in [N]$ the support of the zero-moded random variable $\bX'_i$ is contained in $\{0\} \cup Q$ where $Q=\{\pm q_1,\pm q_2, \pm q_3\}$ where $q_1=a_2 - a_1, q_2 = a_3 - a_2,$ and $q_3 = a_3 - a_1$. 
Further, the support size of  each $\bX'_i$ is $3$ and hence it includes  
at most two of the elements from the set $\{q_1, q_2, q_3\}$.
The fact that $q_3 = q_1 + q_2$ is the additive structure that we shall crucially exploit.  Note that in the case $k=3$ we have $K=3$ as well, and  $\Pr[\bX'_i=0] \geq 1/k=1/K=1/3$ for each $i \in [N].$  

Recalling the framework from the beginning of Section~\ref{sec:learn}, we reorder $q_1,q_2,q_3$ so that $c_{q_1} \leq c_{q_2} \leq c_{q_3}$.  We define the ``largeness index'' $\ell_0 \in \{1,2,3,4\}$ analogously to the definition at the beginning of Section~\ref{sec:learn}, but with a slight difference in parameter settings:  we now define the sequence $t_1, \ldots, t_K$ as $t_\ell= (1/\epsilon)^{C^\ell}$ where $C$ is a (large) absolute constant to be fixed later.  Define the ``largeness index" of the sequence $c_{q_1} \le \ldots \le c_{q_K}$ as the minimum $\ell \in [K]$ such that $c_{q_\ell} > t_\ell $, and let $\ell_0$ denote this value.   If there is no $\ell \in \{1,2,3\}$ such that  $c_{q_\ell} > t_\ell$, then we set $\ell_0  = 4$.

Viewing $\bS$ as $\bS' + V$ as before,
our analysis 
now involves four distinct cases, one for each possible value of $\ell_0.$  

\subsection{The case that $\ell_0 = 4$.} \label{sec:a0-is-four}
This case is identical to Section~\ref{sec:allsmall} specialized to $K=3$, so we can easily learn to accuracy $O(\eps)$ in $\poly(1/\eps^{C^3})=\poly(1/\eps)$ time.

\subsection{The case that $\ell_0=3$.} \label{sec:a0-is-three}

In this case we have $c_{q_1} \leq (1/\eps)^C$ and $c_{q_2} \leq (1/\eps)^{C^2}$ but $c_{q_3} \geq (1/\eps)^{C^3}.$  By Lemma~\ref{lem:light-heavy}, we have that $\dtv(\tilde{\bS},\bS') \leq O(\eps)$ where $\tilde{\bS} = \bS_{\textrm{heavy}} + \bS_{\textrm{light}}$, $\bS_{\textrm{heavy}}$ and  $\bS_{\textrm{light}}$ are independent of each other, $\bS_{\textrm{light}}$ is supported on a set of $O(1/\eps^{2C^2+2})$ integers, and $\bS_{\textrm{heavy}}$ is simply $q_3  \bS_3$ where $\bS_3 = \sum_{i=1}^N \bY_i$ is a signed PBD with $\sum_{i=1}^N \Pr[\bY_i = \pm 1] \geq 1/(2\eps^{C^3}).$ Given this constrained structure, the $\poly(1/\eps)$-sample and running time learnability of 
$\bS^*$
follows as a special case of the algorithm given in the proof of Theorem~\ref{thm:learn-some-large}.  In more detail, as described in that proof, two points drawn from 
$\bS^*$
can be used to obtain, with at least $\poly(\eps)$ probability, a multiplicative factor-2 estimate of $\sqrt{\Var[\bS_{\textrm{heavy}}]}$.  Given such an estimate no gridding is required, as it is possible to learn 
$\bS^*$
to accuracy $O(\eps)$  simply by using the $K=1$ case of the kernel learning result Lemma~\ref{lem:siirv-kernel} (observe that, crucially, having an estimate of $\Var[\bS_{\textrm{heavy}}]$ provides the algorithm with the value $\gamma_1$ in Lemma~\ref{lem:siirv-kernel} which is required to construct $\bZ$ and thereby carry out the kernel learning of $\bS'+V$ using $\bZ$).

\subsection{The case that $\ell_0=2$.} \label{sec:a0-is-two}

In this case we have $c_{q_1} \leq (1/\eps)^C$ while  $c_{q_3}, c_{q_2} \geq (1/\eps)^{C^2}$.  As earlier we suppose that $q_1 + q_2 = q_3$.   (This is without loss of generality as the other two cases are entirely similar; for example, if instead we had $q_1 + q_3 = q_2$, then we would have $q_3 = -q_1 + q_2 $, and it is easy to check that replacing $q_1$ by $-q_1$ everywhere does not affect our arguments.)

Lemma~\ref{lem:light-heavy} now gives us a somewhat different structure, namely that $\dtv(\tilde{\bS},\bS') \leq O(\eps)$ where $\tilde{\bS} = \bS_{\textrm{heavy}} + \bS_{\textrm{light}}$, $\bS_{\textrm{heavy}}$ and  $\bS_{\textrm{light}}$ are independent of each other, $\bS_{\textrm{light}} = q_1  \bS_1$ where $\bS_1$ is supported on $[-O(1/\eps^{C+1}),O(1/\eps^{C+1})] \cap \mathbb{Z}$, and $\bS_{\textrm{heavy}}$ is a 
sum of $0$-moded integer random variables
over $\{\pm q_{2},\pm q_3\}$, and which satisfies  $c_{q_2,\textrm{heavy}}, c_{q_3,\textrm{heavy}} > 1/(2\eps^{C^2})$.  Applying Lemma \ref{lem:close-to-pure-signed-sicsirv} to 
$\bS_{\textrm{heavy}}$, we get that $\dtv(\bS_{\textrm{heavy}},\bB) = O(\eps^{C^2/20})$ where either
\begin{equation} \label{eq:BB1}
\bB = V' + q_2 \bW_{2} +  q_3 \bW_3
\end{equation}
(if the set $A$ from Lemma \ref{lem:close-to-pure-signed-sicsirv} is empty) or
\begin{equation}
\label{eq:BB2}
\bB = V' +  q_2  \bW_{2} +  q_3  \bW_3 + (q_3 + \tau_{2,3}  q_2) \bW_{2,3}
\end{equation}
(if $A=\{(2,3)\}$),
where all the distributions $\bW_2,\bW_3$ (and possibly $\bW_{2,3}$) are independent signed PBDs with variance at least $\Omega(1/\eps^{C^2/4})$ and $\tau_{2,3} \in \{-1,1\}.$  

Let us first suppose that (\ref{eq:BB1}) holds, so $\bS'+V$ is $O(\eps^{C^2/20})$-close to 
\begin{equation} \label{eq:pig}
V'' +  q_1 \bS_1 + q_2 \bW_{2} +  q_3 \bW_3,
\end{equation}
where $V'' = V + V'$. Since, by Fact~\ref{fact:good-shift-invariance}, $q_2  \bW_2$ is $O(\eps^{C^2/8})$-shift-invariant at scale $q_2$, recalling the support of $\bS_1$ we get that $\bS'+V$ is $(O(\eps^{C^2/20}) + O(\eps^{C^2/8}/\eps^{C+1}))$-close (note that this is $O(\eps^{C^2/20})$ for sufficiently large constant $C$) to 
\[
V'' +  (q_1  + q_2)\bS_1 + q_2 \bW_{2} +  q_3 \bW_3 =V'' +  q_2 \bW_{2} +  q_3 (\bW_3 + \bS_1).
\]
Again using the support bound on $\bS_1$ and Fact~\ref{fact:good-shift-invariance} (but now on $q_3 \bW_3$), we get that $\bS'+V$, and therefore $\bS^*$, is $O(\eps^{C^2/20})$-close to
\begin{equation} \label{eq:wig}
V'' +  q_2 \bW_{2} +  q_3 \bW_3.
\end{equation}
We can now apply the algorithm in Lemma~\ref{lem:learn-sum-two-high-variance-PBDs} to semi-agnostically learn the distribution $V'' +  q_2 \bW_{2} +  q_3 \bW_3$ with $\mathsf{poly}(1/\epsilon)$ samples and time complexity.

%
%
%

Next, let us consider the remaining possibility in this case which is that (\ref{eq:BB2}) holds.\ignore{ We henceforth assume that $\tau_{2,3} = -1$, as an entirely similar analysis to what we give below goes through for the other case.  (To see this, note that if $\tau_{2,3}=-1$ then the three coefficients are $q_2,q_3,q_3-q_2=q_1$ and hence the sum of two of the coefficients is equal to the third, while if $\tau_{2,3}=1$ then the three coefficients are $q_2,q_3,q_3+q_2$ and again the sum of two of the coefficients is equal to the third.)}  If $\tau_{2,3}=-1$, then $\bS'+V$ is $O(\eps^{C^2/20})$-close to 
\[
V' +  q_1 \bS_1 + q_2  \bW_{2} +  q_3  \bW_3 + (q_3 -  q_2) \bW_{2,3}=
V' +  q_1 \bS_1 + q_2  \bW_{2} +  q_3  \bW_3 + q_1 \bW_{2,3},
\]
and using Fact~\ref{fact:good-shift-invariance} as earlier, we get that $\bS'+V$ is $O(\eps^{C^2/20})$-close to
\begin{equation} \label{eq:bag}
V'' +  q_2 \bW_{2} +  q_3 \bW_3 + q_1 \bW_{2,3}.
\end{equation}
On the other hand, if $\tau_{2,3}=1$ then $\bS'+V$ is $O(\eps^{C^2/20})$-close to 
\[
V'' +  q_1 \bS_1 + q_2  \bW_{2} +  q_3  \bW_3 + (q_3 +  q_2) \bW_{2,3},
\]
and by the analysis given between (\ref{eq:pig}) and (\ref{eq:wig}) we get that $\bS'+V$ is $O(\eps^{C^2/20})$-close to
\begin{equation}  \label{eq:hag}
V'' + q_2  \bW_{2} +  q_3  \bW_3 + (q_3 +  q_2) \bW_{2,3},
\end{equation}

In either case (\ref{eq:hag}) or (\ref{eq:bag}),\ignore{ with trivial changes, using $q_1 + q_2 = q_3$ and interchanging the roles of $\bW_{2,3}$ and $\bW_3$). For either of these cases,} we can use Lemma~\ref{lem:learn-sum-3-high-variance-PBDs} to learn the target distribution with $\poly(1/\epsilon)$ samples and running time. 

\medskip

%

\subsection{The case that $\ell_0=1$.} \label{sec:a0-is-one}

In this case we have $c_{q_1}, c_{q_2}, c_{q_3} \geq (1/\eps)^C.$  Assuming that $C \ge 96$, we appeal to Lemma~\ref{lem:switch-all-big}
to obtain that there are independent signed PBDs $\bS_1$, $\bS_2$ and $\bS_3$, each with variance at least $1/\epsilon^{2}$, such that 
\[
\dtv(\bS', V + q_1 \bS_1 + q_2 \bS_2 + q_3 \bS_3) \leq O(\eps^2).
\]
As before, we can appeal to Lemma~\ref{lem:learn-sum-3-high-variance-PBDs} to learn the target distribution with  $\poly(1/\epsilon)$ samples and running time.
 
%
%

\subsection{Deferred proofs and learning algorithms from the earlier cases }

\subsubsection{Learning algorithm for weighted sums of two PBDs}

\begin{lemma} \label{lem:learn-sum-two-high-variance-PBDs}
There is a universal constant $C_1$ such that the following holds:  
Let $\bS^\twohigh$ be a distribution of the form $p \cdot \bS^{(p)} + q \cdot \bS^{(q)} + V$, where both $\bS^{(p)}$ and $\bS^{(q)}$ are independent PBD$_N$ distributions with variance at least $1/\eps^{C_1}$ and $V \in \Z.$  Let $\bS$ be a distribution with $\dtv(\bS,\bS^\twohigh) \leq \eps.$  There is an algorithm with the following property:  The algorithm is given $\eps,p,q$ and access to i.i.d.\ draws from $\bS.$  The algorithm makes $\poly(1/\eps)$ draws, runs in $\poly(1/\eps)$ time, and with probability 999/1000 outputs a hypothesis distribution $\bH$ satisfying $\dtv(\bH,\bS) \leq O(\eps).$
\end{lemma}

\begin{proof}
The high level idea of the algorithm is similar to Theorem~\ref{thm:learn-some-large}. First, assume that 
$\Var[ p \cdot \bS^{(p)}] \ge \Var[ q \cdot \bS^{(q)}]$ (the other case is identical, and the overall algorithm tries both possibilities and does hypothesis testing). Let $\sigma_p^2 = \Var [\bS^{(p)}]$, $\sigma_q^2 = \Var [\bS^{(q)}]$ and $\sigma_{\twohigh}^2 = \Var[\bS^\twohigh]$. We consider three cases depending upon the value of $\sigma_q$ and show that in each case the kernel based approach (i.e. Lemma~\ref{lem:siirv-kernel}) can be used to learn the target distribution $\bS^\twohigh$ with $\poly(1/\epsilon)$ samples (this suffices, again by hypothesis testing).  We now provide details.

\medskip
\noindent
\textbf{Estimating the variance of $\bS^\twohigh$:} The algorithm first estimates the variance of $\bS^\twohigh$. This is done by sampling two elements $s^{(1)}, s^{(2)}$ from $\bS^\twohigh$ and letting $|s^{(1)}- s^{(2)}|=\widehat{\sigma}_{\twohigh}$. Similar to the analysis of Theorem~\ref{thm:learn-some-large}, it is easy to show that with probability $\Omega(1)$, we have
\begin{equation}
\label{eq:twohigh}
\frac{1}{\sqrt{2}} \cdot \sigma_{\twohigh} \le \widehat{\sigma}_{\twohigh} \le \sqrt{2} \cdot \sigma_{\twohigh}.
\end{equation}
\noindent
\textbf{Guessing the dominant variance term and the relative magnitudes:} Observe that
\[
\Var[\bS^\twohigh] = \Var[p \cdot \bS^{(p)}] + \Var[q \cdot \bS^{(q)}].
\]
The algorithm next guesses whether $p \cdot \sigma_p \ge q \cdot \sigma_q$ or vice-versa. Let us assume that it is the former possibility. The algorithm then guesses one of the three possibilities: (i) $\sigma_q \le \epsilon \cdot p$, (ii) $\epsilon  \cdot p\le \sigma_q < p/\epsilon$, (iii) $\sigma_q > p/\epsilon$. The chief part of the analysis is in showing that in each of these cases, the algorithm can draw $O(1/\epsilon^2)$ samples from $\bS$ and (with the aid of Lemma~\ref{lem:siirv-kernel}) can produce a hypothesis $\bH$ such that $\dtv(\bH, \bS^\twohigh) = O(\epsilon)$. 
\begin{enumerate}

\item [(i)] In this case, we assume $\sigma_q \le \epsilon \cdot p$. This case is the \emph{crucial point of difference} where we save the factor of $\log \log p$ as opposed to the case $k>3$; this is done by working modulo $p$ to estimate $\sigma_q$.  (This is doable in this case because $\sigma_q$ is so small relative to $p$.) The algorithm samples two points $s^{(3)}, s^{(4)} \sim  \bS$; note that with probability $1-O(\eps)$ these points are distributed exactly as if they were drawn from $\bS^\twohigh$, so we may analyze the points as if they were drawn from $\bS^\twohigh.$  Let us assume that  $s^{(3)} = p \cdot s^{(3)}_p + q\cdot s^{(3)}_q + V$, $s^{(4)} = p \cdot s^{(4)}_p + q\cdot s^{(4)}_q + V$ where $s^{(3)}_p,s^{(4)}_p$ are i.i.d. draws from $\bS^{(p)}$ and similarly for $s^{(3)}_q,s^{(4)}_q$. 
 Then, note that with probability at least $1/10$, we have 
 $$\frac{1}{\sqrt{2}} \cdot  \sigma_q \le |s^{(4)}_q- s^{(3)}_q| \le \sqrt{2} \cdot \sigma_q .$$ 
This immediately implies that if we define $\widehat{\sigma}_q= q^{-1} \cdot (s^{(3)} - s^{(4)}) \ (\mathrm{mod} \  p)$, then $\widehat{\sigma}_q = |s^{(4)}_q- s^{(3)}_q|$, and thus 
\[
\frac{1}{\sqrt{2}} \cdot  \sigma_q \le \widehat{\sigma_q} \le \sqrt{2} \cdot \sigma_q .
\]
This gives one of the estimates required by Lemma~\ref{lem:siirv-kernel}; for the other one, we observe that defining
$\widehat{\sigma}_p:=\widehat{\sigma}_\twohigh/p$, having $p \sigma_p \in [{\frac {\Var[\bS^\twohigh]} 2},\Var[\bS^\twohigh]]$ and (\ref{eq:twohigh}) together give that
\[
\frac{1}{2} \cdot \sigma_p \le \widehat{\sigma}_p \le 2\sigma_p. 
\]
We can now apply Lemma~\ref{lem:siirv-kernel} to get that using $\poly(1/\epsilon)$ samples, we can produce a hypothesis distribution $\bH_{\mathrm{low}}$ such that $\dtv(\bH_{\mathrm{low}}, \bS) = O(\epsilon)$. 

\item [(ii)] In this case, we assume $\epsilon \cdot p < \sigma_q \le  p \cdot (1/\epsilon)$. In this case we simply guess one of the $O(\log(1/\eps)/\eps)$ many values 
\[
\hat{\sigma}_q \in \left\{ {\frac p {(1 + \eps/10)^i}} \right\}_{i \in \{-O(\ln(1/\eps)/\eps),\dots,O(\ln(1/\eps)/\eps)\}}
\]
and one of these guesses $\hat{\sigma}_q$ for $\sigma_q$ will be $(1+\eps/10)$-multiplicatively accurate. For each of these  values of $\hat{\sigma_q}$, as in case (ii) we can get a multiplicatively accurate estimate $\widehat{\sigma}_p$ of $\sigma_p$, so again by invoking Lemma~\ref{lem:siirv-kernel} we can create a hypothesis distribution $\bH_{\mathrm{med}, i}$, and for the right guess we will have that 
$\dtv(\bH_{\mathrm{med}, i}, \bS) = O(\epsilon)$. 

\item [(iii)] In this case, we invoke Lemma~\ref{l:mix} to get that there is a signed PBD $\bS'$ such that $\dtv(\bS', p \cdot \bS^{(p)} + q \cdot \bS^{(q)}) = O(\epsilon)$. This also yields that there is a signed PBD $\bS''= \bS' +V$ such that $\dtv(\bS'', \bS^{\twohigh}) = O(\epsilon)$. By a trivial application of Lemma~\ref{lem:siirv-kernel}, using $\poly(1/\epsilon)$ samples,  we obtain a hypothesis $\bH_{\mathrm{high}}$ such that $\dtv(\bH_{\mathrm{high}}, \bS) = O(\epsilon)$. 

\end{enumerate}
Finally, invoking the $\mathrm{Select}$ procedure from Proposition \ref{prop:log-cover-size} on the hypothesis distributions 
\[
\bH_{\mathrm{low}}, \{\bH_{\mathrm{med}, i}\}_{i \in \{-O(\ln(1/\eps)/\eps),\dots,O(\ln(1/\eps)/\eps)\}} \text{~and~}\bH_{\mathrm{high}},
\] 
we can use an additional $\poly(1/\epsilon)$ samples to output a distribution $\bH$ such that $\dtv(\bH, \bS) = O(\epsilon)$. 
\end{proof}

\subsubsection{Learning algorithm for weighted sums of three PBDs}~\label{sec:three-learn}
We now give an algorithm for learning a distribution of the form $p \cdot \bS^{(p)} + q \cdot \bS^{(q)} + r \cdot \bS^{(r)} + V$ where $r =  p + q$. 

\begin{lemma} \label{lem:learn-sum-3-high-variance-PBDs}
There is a universal constant $C_1$ such that the following holds:  
Let $\bS^\threehigh$ be a distribution of the form $p \cdot \bS^{(p)} + q \cdot \bS^{(q)} + r \cdot \bS^{(r)} + V$, where $\bS^{(p)}, \bS^{(q)}$ and $\bS^{(r)}$ are independent PBD$_N$ distributions with variance at least $1/\eps^{C}$ and $V \in \Z$ and $r=q+p.$   Let $\bS$ be a distribution with $\dtv(\bS,\bS^\threehigh) \leq \eps.$  There is an algorithm with the following property:  The algorithm is given $\eps,p,q,r$ and access to i.i.d.\ draws from $\bS.$  The algorithm makes $\poly(1/\eps)$ draws, runs in $\poly(1/\eps)$ time, and with probability 999/1000 outputs a hypothesis distribution $\bH$ satisfying $\dtv(\bH,\bS) \leq O(\eps).$
\end{lemma}

\begin{proof}
The algorithm begins by sampling two points $s^{(1)}, s^{(2)}$ from $\bS$.  Similar to the preceding proof, with probability $\Omega(1)$ we have 
\[
\frac{1}{\sqrt{2}} \cdot \sigma_{\threehigh}\le \hat{\sigma}_{\threehigh} \le \sqrt{2} \cdot \sigma_{\threehigh}, 
\] 
where $\sigma^2_{\threehigh} = \Var[\bS^{\threehigh}]$. Having obtained an estimate of $\sigma_{\threehigh}$, let us now assume (without loss of generality, via hypothesis testing) that $\sigma_p \ge \sigma_q \ge \sigma_r$. Similar to  Lemma~\ref{lem:learn-sum-two-high-variance-PBDs}, we consider various cases, and for each case (and relevant guesses) we run Lemma~\ref{lem:siirv-kernel} and obtain a hypothesis distribution for each of these guesses. Finally, we will use procedure $\mathrm{Select}$ (Proposition~\ref{prop:log-cover-size}) on the space of these hypotheses to select one. Let us now consider the cases: 
\begin{enumerate}

\item $\sigma_r \ge \epsilon^{5} \cdot \sigma_p$: In this case, note that given $\widehat{\sigma}_{\threehigh}$, we can construct a grid $J$ of $\poly(1/\eps)$ many triples such that there exists $\overline{\gamma}=(\gamma_p, \gamma_q, \gamma_r) \in J$  such that for $\alpha \in \{p,q,r\}$, 
$$
\frac{1}{\sqrt{2}} \cdot \sigma_\alpha \le \gamma_{\alpha} \le \sqrt{2} \cdot \sigma_{\alpha}. 
$$
For each such possibility $\overline{\gamma}$, we can apply Lemma~\ref{lem:siirv-kernel} which uses $\poly(1/\epsilon)$ samples; as before, for the right guess, we will obtain a hypothesis $\bH_{\overline{\gamma}}$ such that 
$\dtv(\bH_{\overline{\gamma}}, \bS) = O(\epsilon)$.  

\item $\sigma_r \le \epsilon^{5} \cdot \sigma_p$: In this case, since $r=p+q$, 
\[
p \cdot \bS^{(p)} + q \cdot \bS^{(q)} + r \cdot \bS^{(r)} = p \cdot \bS^{(p)} + q \cdot \bS^{(q)} + (p+q) \cdot \bS^{(r)}.
\]
As $\sigma_r \le \epsilon^{5} \cdot \sigma_p$, using the $O(1/\sigma_p)$-shift-invariance of $p \cdot \bS^{(p)}$ at scale $p$ that follows from Fact~\ref{fact:good-shift-invariance} and a Chernoff bound on $\bS^{(r)}$,
we get that for some integer $V'$,
\[
\dtv(p \cdot \bS^{(p)} + q \cdot \bS^{(q)} + r \cdot \bS^{(r)},p \cdot (\bS^{(p)} + V') + q \cdot (\bS^{(q)} + \bS^{(r)})) = O(\epsilon^4). 
\]
Thus we have
\[
\dtv(\bS^{\threehigh}, V''+p \cdot {\bS}^{(p)} + q \cdot (\bS^{(q)} + \bS^{(r)})) = O(\epsilon^4)
\]
for some integer $V''.$  However, now we are precisely in the same case as Lemma~\ref{lem:learn-sum-two-high-variance-PBDs}. Thus, using $\poly(1/\epsilon)$ samples, we can now obtain $\bH_{(2)}$ such that $\dtv(\bH_{(2)}, \bS) = O(\epsilon)$. 
\end{enumerate}
Finally, we apply  $\mathrm{Select}$ (Proposition~\ref{prop:log-cover-size}) on $\bH_{(2)}$ and  $\{\bH_{\overline{\gamma}}\}_{\overline{\gamma} \in J}$. This finishes the proof. 
\end{proof}

\subsubsection{Structural lemma for decomposing a heavy 
distribution into a sum of 
weighted sums of PBDs}

Our goal in this subsection is to prove Lemma~\ref{lem:switch-all-big}.  To do this, we will need a slightly more detailed version of Lemma~\ref{lem:close-to-pure-signed-sicsirv}
in the case that $K=2$, which is implicit in the proof of that lemma
(using the case that $A = \emptyset$ for (\ref{e:A.empty}) and
the case that $A = \{ (1,2) \}$ for (\ref{e:A.12})).

\begin{lemma} \label{lem:close-to-pure-signed-sicsirv.K=2}
Under the assumptions of Lemma~\ref{lem:close-to-pure-signed-sicsirv} in the
case that $K=2$, there is an integer $V'$, and independent signed PBDs 
$\bW_{1,1}$, $\bW_{2,2}$ and $\bW_{1,2}$, all with variance at least $\Omega(\BIG^{1/4})$,
such that either
\begin{equation}
\label{e:A.empty}
d_{TV} (\bS', V' + q_1 \bW_{1,1} + q_2 \bW_{2,2}) =  O(\BIG^{-1/20}),
\end{equation}
or
\begin{equation}
\label{e:A.12}
d_{TV} (\bS', V' + q_1 \bW_{1,1} + q_2 \bW_{2,2} + (q_2 + \sign(\Cov(M_1,M_2)) q_1) \bW_{1,2}) 
  =  O(\BIG^{-1/20}),
\end{equation}
where $\bM = (\bM_1, \bM_2)$ is as defined in (\ref{eq:M}).
\end{lemma}

Let  $\bS' =\sum_{i=1}^N \bX'_i$ where each $\bX'_i$ is $0$-moded and supported on $\{0,\pm p, \pm q , \pm r\}$ where $r = q+p$. Each random variable $\bX'_i$ has a support of size $3$, and by inspection of how $\bX'_i$ is obtained from $\bX_i$, we see that each $\bX'_i$ is supported either on $\{0,p,r\}$, or on $\{-p,0,q\}$, or on $\{-r,-q,0\}$. If for $\alpha \in \{p,q,r\}$, we define $c_{\alpha} = \sum_{i=1}^N \Pr[X'_i = \pm\alpha]$, then we have the following lemma. 

 \begin{lemma} \label{lem:switch-all-big}
Let $\bS' =\sum_{i=1}^N \bX'_i$ as described above where $c_p,c_q,c_r  > 1/\eps^C$.  Then we have
\begin{equation} \label{eq:Sprimecloseredux}
\dtv(\bS', V + p \cdot \bS^{(p)} + q \cdot \bS^{(q)} + r \cdot \bS^{(r)}) 
    \leq O(\eps^{C_1}),
\end{equation}
where $V$ is a constant, $C_1 = C/48$ and $\bS^{(p)},  \bS^{(q)}$ and $\bS^{(r)}$ are mutually independent PBD$_N$ distributions each of which has variance at least $1/(\eps^{C_1})$.
\end{lemma}

\begin{proof}
We first prove that 
\begin{equation}\label{eq:shiftpqr}
\dshift{p}(\bS') \le \epsilon^{C/2} ; \ \dshift{q}(\bS') \le \epsilon^{C/2} ; \ \dshift{r}(\bS') \le \epsilon^{C/2}. 
\end{equation}
To see this, note that,
as we showed in the proof of Lemma~\ref{lem:light-heavy},
 $\dshift{p}(\bX'_i) \leq 1- \Pr[\bX'_i = \pm p]$. By applying Corollary~\ref{corr:pshift}, we get that
$$
\dshift{p}(\bS') = O \left( \frac{1}{\sqrt{\sum_i \Pr[\bX'_i = \pm p]}}\right) = O(\eps^{C/2}).
$$
Likewise, we also get the other components of (\ref{eq:shiftpqr}). 

Let us next consider three families of i.i.d. random variables $\{\bY'_i\}_{i=1}^m$, $\{\bZ'_i\}_{i=1}^m$ and $\{\bW'_i\}_{i=1}^m$ defined as follows: for $1 \le i \le m$, 
\[
\Pr[\bY'_i=0] = \Pr[\bY'_i=p] = 1/2 ; \ \  \Pr[\bZ'_i=0] = \Pr[\bZ'_i=q] = 1/2 ; \Pr[\bW'_i=0] = \Pr[\bW'_i=r] = 1/2 ; 
\]
Let $m = \epsilon^{-C/4}$. Let $\sum_{i=1}^m \bY'_i = \bS^{(y)}$, $\sum_{i=1}^m \bZ'_i = \bS^{(z)}$ and $\sum_{i=1}^m \bW'_i = \bS^{(w)}$. Let $\bS_e = \bS^{(y)} + \bS^{(z)} + \bS^{(w)}$, and note that $\bS_e$ is supported on $\{i \cdot p + j \cdot q + k \cdot r\}$ where $0 \le i , j ,k \le m$. Using (\ref{eq:shiftpqr}), 
we have
\[
\dtv(\bS', \bS' + \bS_e) \le m \cdot \epsilon^{C/2} = O(\epsilon^{C/4}).
\]
Thus, it suffices to prove 
\begin{equation} \label{eq:Sprimecloseredux.prime}
\dtv(\bS' + \bS_e, V + p \cdot \bS^{(p)} + q \cdot \bS^{(q)} + r \cdot \bS^{(r)}) 
    \leq O(\eps^{C_1}).
\end{equation}

We assign each random variable $\bX'_i$ to one of three different types:
\begin{itemize}
\item \textbf{Type 1:} The support of $\bX'_i$ is $\{0, p , r\}$. 
\item \textbf{Type 2:} The support of $\bX'_i$ is $\{-p, 0 , q\}$. 
\item \textbf{Type 3:} The support of $\bX'_i$ is $\{-r, -q  , 0\}$.
\end{itemize}
Let the set of Type $1$ variables be given by $\mathcal{S}_1$. We will show that there exists independent 
signed PBDs $\bS^{(p,1)}$, $\bS^{(q,1)}$ and $\bS^{(r,1)}$ 
and a constant $V_1$ such that the variances of $\bS^{(p,1)}$ and $\bS^{(r,1)}$ are each at
least $\epsilon^{-C_1}$, and $\bS^{(q,1)}$ is either constant 
(when (\ref{e:A.empty}) holds)
or has variance
at least $\epsilon^{-C_1}$ (when (\ref{e:A.12}) holds), and that satisfy
\begin{equation}\label{po:1}
\dtv\bigg(\sum_{i \in \mathcal{S}_1} \bX'_i + \sum_{i=1}^{m/2} \bY'_i + \sum_{i=1}^{m/2} \bW'_i, V_1 + p \cdot \bS^{(p,1)} + q \cdot \bS^{(q,1)} + r \cdot \bS^{(r,1)}\bigg) = O(\epsilon^{C/48}). 
\end{equation}

If we can prove (\ref{po:1}), then we can analogously prove the symmetric
statements that
\[
\dtv\bigg(\sum_{i \in \mathcal{S}_2} \bX'_i 
 + \sum_{i=m/2 + 1}^{m} \bY'_i + \sum_{i=1}^{m/2} \bZ'_i, 
V_2 + p \cdot \bS^{(p,2)} + q \cdot \bS^{(q,2)} + r \cdot \bS^{(r,2)}\bigg)  = O(\epsilon^{C/48})
\]
and
\[
\dtv\bigg(\sum_{i \in \mathcal{S}_3} \bX'_i 
+ \sum_{i=m/2 + 1}^{m} \bW'_i + \sum_{i=m/2+1}^{m} \bZ'_i, 
V_3 + p \cdot \bS^{(p,3)} + q \cdot \bS^{(q,3)} + r \cdot \bS^{(r,3)}\bigg)  = O(\epsilon^{C/48}),
\]
with analogous conditions on the variances of 
$\bS^{(p,2)}, \bS^{(q,2)}, \bS^{(r,2)},\bS^{(p,3)}, \bS^{(q,3)}, \bS^{(r,3)}$, and
combining these bounds with (\ref{po:1}) will imply the desired inequality (\ref{eq:Sprimecloseredux.prime}).

Thus it remains to prove (\ref{po:1}).  Let $\gamma_1 = \sum_{i \in \mathcal{S}_1} \Pr[\bX'_i = p]$ 
and $\delta_1 =  \sum_{i \in \mathcal{S}_1} \Pr[\bX'_i = r]$. We consider the following cases: 

\medskip

\noindent {\bf Case (I):}  Assume $\gamma_1$ and $\delta_1 \ge \epsilon^{-C/8}$.
Since the possibilities $\bX'_i=p$ and $\bX'_i=r$ are mutually exclusive, for each $i$ we have that 
$\Cov(\mathbf{1}_{\bX_i' = p}, \mathbf{1}_{\bX_i' = r}) \leq 0$,
which implies that $\Cov(\sum_i \mathbf{1}_{\bX_i' = p}, \sum_i \mathbf{1}_{\bX_i' = r}) \leq 0$.  
Applying Lemma~\ref{lem:close-to-pure-signed-sicsirv.K=2} to the distribution $\sum_{i \in \mathcal{S}_1} \bX'_i$,
we obtain (\ref{po:1}) in this case.

\medskip

\noindent {\bf Case (II):} Now let us assume that at least one of $\gamma_1$ or $\delta_1$ is less than $\epsilon^{-C/8}$,
without loss of generality, that $\gamma_1 < \epsilon^{-C/8}$. For each variable $\bX'_i$, let us consider a corresponding random variable $\tilde{\bX'}_i$ 
defined by replacing the $p$-outcomes of $\bX'_i$ with $0$'s. 
If $\bZ$ is any distribution such that $\dshift{p}(\bZ) \le \kappa$, then 
\[
\dtv(\bZ + \bX'_i , \bZ + \tilde{\bX'}_i) \le \kappa \cdot \Pr[\tilde{\bX'}_i = p ]. 
\]
Applying this observation iteratively, we have
\[
\dtv\bigg(  \sum_{i \in \mathcal{S}_1} \bX'_i + \sum_{i=1}^{m/2} \bY'_i + \sum_{i=1}^{m/2} {\bW'}_i ,  \sum_{i \in \mathcal{S}_1} \tilde{\bX'}_i + \sum_{i=1}^{m/2} \bY'_i + \sum_{i=1}^{m/2} {\bW'}_i  \bigg) = O(\gamma_1 \cdot \epsilon^{C/4}) = O(\epsilon^{C/8}). 
\]
However, now note that $\sum_{i \in \mathcal{S}_1} \tilde{\bX'}_i + \sum_{i=1}^{m/2} \bY'_i + \sum_{i=1}^{m/2} {\bW}'_i$ 
can be expressed as 
$p \cdot \bS^{(p,1)} + r \cdot \bS^{(r,1)}$
for independent signed PBDs $\bS^{(p,1)}$, and $\bS^{(r,1)}$,
and that the variances of $\sum_{i=1}^{m/2} \bY'_i$ and $\sum_{i=1}^{m/2} {\bW'}_i $ ensure that
$\Var[\bS^{(p,1)}], \Var[\bS^{(r,1)}] \ge \epsilon^{-C/2}$.  This establishes
(\ref{po:1}) in this case, completing the proof of the lemma.
\end{proof}

\section{Unknown-support algorithms:  Proof of Theorems \ref{unknown-general-k-upper} and \ref{unknown-k-is-two-upper}}
\label{sec:unknown-support}

We begin by observing that the hypothesis selection procedure described in Section \ref{sec:hyp-test} provides a straightforward reduction from the case of unknown-support to the case of known-support.  More precisely, it implies the following:

\begin{observation} \label{obs:reduction}
For any $k$, let $A$ be an algorithm that semi-agnostically learns 
$\{ a_1,..., a_k \}$-sums,
with $0 \leq a_1 <  \cdots < a_k$, 
using $m(a_1,\dots,a_k,\eps,\delta)$ samples and running in time $T(a_1,\dots,a_k,\eps,\delta)$ to learn to accuracy $\eps$ with probability at least $1-\delta$, outputting a hypothesis distribution from which it is possible to generate a draw or evaluate the hypothesis's p.m.f. on a given point in time $T'(a_1,\dots,a_k).$  Then there is an algorithm $A'$ which
semi-agnostically learns
$(\amax, k)$-sums
using
\[
\left(\max_{0 \leq a_1 < \cdots < a_k \leq \amax} m(a_1,\dots,a_k,\eps/6,\delta/2) \right) + 
O(k \log(\amax)/\eps^2 + \log(1/\delta)/\eps^2)\]
 samples, 
and running
in time
\[
\poly({(}\amax)^k,1/\eps) \cdot \left(\max_{0 \leq a_1 < \cdots < a_k \leq \amax} \left( T(a_1,\dots,a_k,\eps) \right)  + T'(a_1,\dots,a_k)  \right),
\]
and, with probability at least $1-\delta$, outputting
a hypothesis with error at most $6\eps.$
\end{observation}

The algorithm $A'$  works as follows:  first, it tries all (at most $(\amax)^k$) possible vectors of values for $(a_1,\dots,a_k)$ as the parameters for algorithm $A$, using the same set of $\max_{0 \leq a_1 < \cdots < a_k \leq \amax} m(a_1,$ $\dots,a_k,\eps,\delta/2)$ samples as the input for each of these runs of $A$.  Having done this, $A'$ has a list of candidate hypotheses such that with probability at least $1-\delta/2$, at least one of the candidates is $\eps$-accurate.  Then $A'$  runs the $\mathrm{Select}$ procedure from Proposition \ref{prop:log-cover-size} on the resulting hypothesis distributions.  

Together with Theorems \ref{known-k-is-three-upper}
and \ref{known-general-k-upper}, Observation \ref{obs:reduction} immediately yields Theorem \ref{unknown-general-k-upper} (learning 
$(\amax, 3)$-sums).  It also yields a result for the unknown-support $k=2$ case, but a sub-optimal one because of the $\log(\amax)$ dependence.  In the rest of this section we show how the sharper bound of Theorem \ref{unknown-k-is-two-upper}, with no dependence on $\amax$, can be obtained by a different (but still simple) approach.

Recall Theorem \ref{unknown-k-is-two-upper}:

\begin{reptheorem}{unknown-k-is-two-upper}
\emph{Learning $(\amax, 2)$-sums} 
There is an algorithm and a positive constant $c$ with the following properties:  The algorithm is given $N$, accuracy and confidence parameters $\eps,\delta > 0$, an upper bound $\amax \in \Z_{+}$, and access to i.i.d. draws from an unknown 
random variable $\bS^*$ that is $c \epsilon$-close to an
$\{a_1,a_2\}$-sum $\bS$,
where $0 \leq a_1 \leq a_2 \leq \amax.$  
The algorithm uses $\poly(1/\eps)$ draws from $\bS^*$, runs in $\poly(1/\eps,\log(1/\delta))$ time, and with probability $1-\delta$ outputs a (concise representation of a) hypothesis distribution $\bH$ such that $\dtv(\bH,\bS^*) \leq \eps.$
\end{reptheorem}

\begin{proof}
Let 
the 
$\{a_1,a_2\}$-sum $\bS$ over (unknown values) $\{a_1,a_2\}$ 
be $\bS = \sum_{i=1}^N \bX_i$ where
$\bX_i(a_1)=1-p_i$, $\bX_i(a_2)=p_i.$  Let $\bX'_i$, $i \in [N]$ be independent Bernoulli random variables with $\bX'_i(1)=p_i.$
The distribution $\bS$ is identical to $a_1 N + (a_2-a_1) \bS'$ where $\bS'$ is the PBD $\bS' = \sum_{i=1}^N \bX'_i.$

Intuitively, if the values $a_1$ and $a_2$ (hence $a_1$ and $a_2-a_1$) were known then it would be simple to learn
using the algorithm for learning a PBD$_N.$  The idea of what follows is that either (i) learning is easy because the essential support is small, or (ii) it is possible to infer the value of $a_2-a_1$ (basically by taking the $\gcd$ of a few sample points) and, given this value, to reduce to the problem of learning a PBD$_N.$  Details follow.

If the PBD $\bS'$ is in sparse form (see Theorem \ref{thm: sparse cover theorem}), then it (and hence $\bS$) is $\eps$-essentially-supported on a set of size $O(1/\eps^3)$.  In this case the algorithm $A'$ of Fact~\ref{fact:learn-sparse-ess-support} can learn $\bS^*$ to accuracy $O(\eps)$ in $\poly(1/\eps)$ time using
$\poly(1/\eps)$ samples.  Thus we may subsequently assume that $\bS'$ is not in sparse form. (The final overall algorithm will run both $A'$ and the learning algorithm described below, and use the hypothesis selection procedure from Section \ref{sec:hyp-test} to choose the overall final hypothesis from these two.)

Since $\bS'$ is not in sparse form, by the last part of
Theorem \ref{thm: sparse cover theorem} it is in $1/\eps$-heavy binomial form.  We will require the following
proposition:

\begin{proposition} \label{prop:GCD}
Let $a \in \Z_+,b \in \Z$ be arbitrary constants.
For all small enough $\epsilon > 0$, if
$\bS' =\sum_{i=1}^N \bS'_i$ 
is
a PBD 
in $1/\eps$-heavy binomial form, then with
probability at least $1-O(\sqrt{\eps})$, the $\gcd$ of $m=\Omega(1/\sqrt{\eps})$ i.i.d. draws $v_1,\dots,v_m$ from $a(\bS' - b)$ is equal to $a$.
\end{proposition}

\begin{proof}
$\bS'$ is $\eps$-close to a translated Binomial distribution $\bY$ as described in the second bullet of Theorem \ref{thm: sparse cover theorem}, and (by Fact \ref{fact:heavy-PBD-nice}) $\bY$ is $O(\eps)$-close to $\bZ$, a 
discretized $N(\mu_{\bY},\sigma_{\bY}^2)$ Gaussian where $\sigma_{\bY}^2=\Omega(1/\eps^2)$.  
By Lemma~\ref{l:coupling},
a collection of $m$ i.i.d. draws from 
$\bS^*$
is distributed exactly as a collection of $m$ i.i.d. draws from $\bZ$ except with failure probability $O(m \eps)$.  Incurring this failure probability, we may suppose that $v_1,\dots,v_m$ are i.i.d. draws from $\bZ$.  Except with an additional failure probability $2^{-\Omega(m)}$, at least $\Omega(m)$ of these draws lie within $\pm \sigma_{\bY}$ of $\mu_{\bY}$, so we additionally suppose that this is the case.  Next, since any two points within one standard deviation of the mean of a discretized Gaussian have probability mass within a constant factor of each other, with an additional failure probability of at most $2^{-\Omega(m)}$ we may suppose that the multiset $\{v_1,\dots,v_m\}$ contains at least $\ell = \Omega(m)$
points that are distributed uniformly and independently over the integer interval $I := [\mu_{\bY}-\sigma_{\bY},\mu_{\bY} + \sigma_{\bY}] \cap \Z$.  Thus, to establish the proposition, it suffices to prove that for any $b \in \Z$, with high probability the $\gcd$ of $\ell$ points drawn uniformly and independently from the shifted interval $I-b$ is 1.

The $\gcd$ is 1 unless there is some prime $p$ such that all $\ell$ draws from $I-b$ are divisible by $p$.  
Since $\eps$ is at most some sufficiently small constant, we have that $|I|$ is at least (say) 100; since $|I|\geq 100$,
for any prime $p$ at most a $1.02/p$ fraction of the points in $I$ are divisible by $p$, so $\Pr[\text{all $\ell$ draws are divisible by }p] \leq (1.02/p)^{\ell}.$  Thus a union bound gives
\begin{align*}
\Pr[\gcd > 1] &\leq \sum_{\text{prime~p}}\Pr[\text{~all $\ell$ draws are divisible by~}p] 
\leq  \sum_{\text{prime~p}} (1.02/p)^\ell
<\sum_{n \geq 2} (1.02/n)^\ell < (2/3)^\ell,
\end{align*}
and the proposition is proved.
\end{proof}

We now describe an algorithm to learn 
$\bS^*$
when $\bS'$ is in\ignore{$\eps$-close to some PBD $\bY = \sum_{i=1}^n \bY_i \in \calS_{N,\eps}$ that is in } $1/\eps$-heavy binomial form.  The algorithm first makes a single draw from 
$\bS^*$
call this the ``reference draw''.  
With probability at least $9/10$, it is from 
$\bS$;  let us assume from for the rest of the proof that this is
the case,
 and let its value be $v=a_1(N-r)+a_2r$.
Next, the algorithm makes $m=\Omega(1/\sqrt{\eps})$ i.i.d. draws $u_1,\dots,u_m$ from 
$\bS^*$.  Since $d_TV(\bS, \bS^*) < c \epsilon$ for a small positive
constant $c$, and $\epsilon$ is at most a small constant, a union bound implies that,
except for a failure probability $O((1/\sqrt{\epsilon}) \epsilon) < 1/10$, all of these
draws come from $\bS$.  Let us assume from here on that this is the case.
For each $i$, the algorithm sets $v_i = u_i - v$, and computes the $\gcd$ of $v_1,\dots,v_m$.  Each $u_i$ equals $a_1 (N-n_i) + a_2n_i$ where $n_i$ is drawn from the PBD $\bS'$, so we have that
\[
v_i =a_1 (N-n_i) + a_2n_i - a_1(N-r) - a_2r = (a_2-a_1)(n_i -r),
\]
and Proposition \ref{prop:GCD} gives that with failure probability at most $O(\sqrt{\eps})$, the $\gcd$
of $n_1 -r, ..., n_m - r$ is 1, so that the $\gcd$ of $v_1,...,v_m$
is equal to
$a_2-a_1.$

With the value of $a_2-a_1$ in hand, it is straightforward to learn $\bS$.  Dividing each draw from $\bS$ by $a_2-a_1$, we get draws
from ${\frac {a_1N}{a_2-a_1}} + \bS'$ where $\bS'$ is the PBD$_N$ described above.  Such a ``shifted PBD'' can be learned easily as follows: if ${\frac {a_1N}{a_2-a_1}}$ is an integer then this is a PBD$_{(a_1+1)N}$, hence is a PBD$_{(\amax +1)N}$, and can be learned using the algorithm for learning a PBD$_{N'}$ given the value of $N'.$  If ${\frac {a_1N}{a_2-a_1}}$ is not an integer, then its non-integer part can be obtained from a single draw, and subtracting the non-integer part we arrive at the case handled by the previous sentence.  

The algorithm described above has failure probability $O(\sqrt{\eps})$, but by standard techniques 
(see Lemma 3.4 of \cite{haukealitwar91} and Proposition~\ref{prop:log-cover-size})
this failure probability can be reduced to an arbitrary $\delta$ at the cost of a $\log(1/\delta)$ factor increase in sample complexity.  This concludes the proof of Theorem \ref{unknown-k-is-two-upper}.
\end{proof}

\section{A reduction for weighted sums of PBDs} 
\label{sec:reduction}

Below we establish a reduction showing that an efficient algorithm for learning 
sums of weighted PBDs with weights $\{0=a_1,\dots,a_k\}$ 
implies the existence of an efficient algorithm for learning 
sums of weighted PBDs with weights $\{0=a_1,\dots,a_{k-1}\} \mod a_k$.
Here by ``learning
sums of weighted PBDs with weights $\{0=a_1,\dots,a_{k-1}\} \mod a_k$'' we mean an algorithm which is given access to i.i.d. draws from the distribution $\bS' := (\bS \mod a_k)$ where 
$\bS$ is a weighted sum of PBDs with weights $\{0=a_1,\dots,a_{k-1}\}$, and should produce a high-accuracy hypothesis distribution for $\bS'$ (which is is supported over $\{0,1,\dots,a_{k}-1\}$); so both the hypothesis distribution \emph{and the samples provided to the learning algorithm} are reduced mod $a_k$.  Such a reduction will be useful for our lower bounds because it enables us to prove a lower bound for learning a 
weighted sum of PBDs with $k$ unknown weights
by proving a lower bound for learning mod $a_k$ with
$k-1$ weights.

The formal statement of the reduction is as follows:

\begin{theorem} \label{thm:reduction}
Suppose that $A$ is an algorithm with the following properties:  $A$ is given $N$, an accuracy parameter $\eps > 0$, a confidence parameter $\delta>0$, and distinct non-negative integers $0=a_1,\dots,a_k$.  $A$ is provided with access to i.i.d. draws from a distribution $\bS$ where $\bS = a_2 \bS_2 + \cdots + a_{k} \bS_{k}$ and
each $\bS_i$ is an unknown PBD$_N$.  
For all $N$, $A$ makes $m(a_1,\dots,a_k,\eps,\delta)$ draws from $\bS$
and with probability at least $1-\delta$ outputs a hypothesis $\tilde{\bS}$ such that $\dtv(\bS',\tilde{\bS})
\leq \eps.$

Then there is an algorithm $A'$ with the following properties:  $A'$ is given $N,0=a_1,\dots,a_{k},\eps,\delta$ and is provided with access to i.i.d. draws from $\bT' := (\bT \mod a_k)$ where   $\bT = a_2 \bT_2 + \cdots + a_{k-1} \bT_{k-1}$    
where in turn each $\bT_i$ is a PBD$_N$.  $A'$ makes $m'=m(a_1,\dots,a_k,\eps,\delta/2)$ draws
from $\bT'$ and with probability $1-\delta$ outputs a hypothesis $\tilde{\bT}'$ such that $\dtv(\bT',\tilde{\bT}') \leq \eps.$
\end{theorem}

\begin{proof}The high-level idea is simple; in a nutshell, we leverage the fact that the algorithm $A$ works with sample complexity $m(a_1,\dots,a_k,\eps,\delta)$ independent of $N$ for all $N$ to construct a data set suitable for algorithm $A$ from a ``mod $a_k$'' data set that is the input to algorithm $A'$.   

In more detail, as above suppose the target distribution $\bT'$ is $(\bT \mod a_k)$ where $\bT = a_2 \bT_2 + \cdots + a_{k-1}\bT_{k-1}$ and each $\bT_i$ is an independent PBD$_N$.  Algorithm $A'$ works as follows:  First, it makes $m'$ draws $v'_1,\dots,v'_{m'}$ from $\bT'$, the $j$-th of which is equal to some value $(a_{2} N_{2,j} + \cdots + a_{k-1}N_{k-1,j}) \mod a_k$.  Next, using its own internal randomness it makes $m'$ draws $N_{k,1},\dots, N_{k,m'}$ from the PBD$_{N^\star}$ distribution $\bT_k:= \Bin(N^\star,1/2)$ (we specify $N^\star$ below, noting here only that $N^\star \gg N$) and constructs the ``synthetic'' data set of $m'$ values whose $j$-th element is
\[
u_j := v'_j +  a_k N_{k,j}.
\]
Algorithm $A'$ feeds this data set of values to the algorithm $A$, obtains a hypothesis $\bH$, and outputs $(\bH \mod a_k)$ as its final hypothesis.

To understand the rationale behind this algorithm, observe that if each value
$v'_j$
were an independent draw from $\bT$ rather than $\bT'$
(i.e., if it were not reduced mod $a_k$), then each $u_j$ would be distributed precisely as a draw from 
$\bT^\star := a_2 \bT_2 + \cdots + a_{k-1} \bT_{k-1} + a_k \bT_k$ (observe that each PBD$_{N_i}$ is also a PBD$_{N^\star}$, simply by having the ``missing'' $N^\star - N_i$ Bernoulli random variables all trivially output zero).  In this case we could invoke the performance guarantee of algorithm $A$ when it is run on such a target distribution.  The issue, of course, is that $v'_j$ is a draw from $\bT'$ rather than $\bT$, i.e. $v'_j$ equals $(v_j \mod a_k)$ where $v_j$ is some draw from $\bT.$  We surmount this issue by observing that since $a_k \bT_k$ is shift-invariant at scale $a_k$,  by taking $\bT_k$ to have sufficiently large variance, we can make the variation distance between the distribution of each $v'_j$ and the original $v_j$ sufficiently small that so it is as if the values $v'_j$ actually were drawn from $\bT$ rather than $\bT'.$

In more detail, let
us view $v'_j$ as the reduction mod $a_k$ of a draw $v_j$
from $\bT$ as just discussed; i.e., let
$v'_j \in \{ 0, ..., a_k - 1 \}$ satisfy $v'_j = v_j + a_k c_j$ for $c_j \in \Z$. We observe that each $c_j$ satisfies $|c_j| < \amax \cdot N.$  Recalling that $\bT_k = \Bin(N^\star, 1/2)$ has $\Var[\bT_k] = N^\star/4,$ by Fact \ref{fact:good-shift-invariance} we have that $\dtv(a_k \bT_k + \bT,a_k \bT_k + \bT') \leq O(1/\sqrt{N^\star}) \cdot \amax \cdot N$.  Hence the variation distance between $(a_k \bT_k + \bT')^{m'}$ (from which the sample $u_1,\dots,u_{m'}$ is drawn) and $(a_k \bT_k + \bT)^{m'}= (\bT^\star)^{m'}$
 (what we would have gotten if each $v_j'$ were replaced by $v_j$)
is at most $O(1/\sqrt{N^\star}) \cdot \amax \cdot N \cdot m'.$  By taking $N^\star = \Theta((\amax \cdot N \cdot m' / \delta)^2)$, this is at most $\delta/2$, so at the cost of a $\delta/2$ failure probability we may assume that the $m'=m(a_1,\dots,a_k,\eps,\delta/2)$-point sample $u_1,\dots,u_{m'}$ is drawn from 
$\bT^\star.$  Then with probability $1-\delta/2$ algorithm $A$ outputs an $\eps$-accurate hypothesis, call it $\tilde{\bT}^\star$ (this is the $\bH$ mentioned earlier), for the target distribution $\bT^\star$ from which its input sample was drawn, so $\dtv(\tilde{\bT}^\star,\bT^\star) \leq \eps.$  Taking $\tilde{\bT'}$ to be $(\tilde{\bT}^\star \mod a_k)$ and observing that $(\bT^\star \mod a_k) \equiv \bT'$, we have that 
\[
\dtv(\tilde{\bT'},\bT') = 
\dtv((\tilde{\bT}^\star \mod a_k),(\bT^\star \mod a_k)) \leq
\dtv(\tilde{\bT}^\star,\bT^\star)
\leq \eps,
\]
and the proof is complete.
\end{proof}

\section{Known-support lower bound for $|\supportset|=4$:  Proof of Theorem \ref{known-k-is-four-lower}} \label{sec:lower}

Recall Theorem \ref{known-k-is-four-lower}:
%

\begin{reptheorem}{known-k-is-four-lower}
\emph{($k=4$, known-support lower bound)} 
Let $A$ be any algorithm with the following properties:  algorithm $A$ is given $N$, an accuracy parameter $\eps > 0$, distinct values $0 \leq a_1 < a_2 < a_3 < a_4 \in \Z$, 
and access to i.i.d. draws from 
an unknown $\{a_1,a_2,a_3,a_4\}$-sum
and with probability at least $9/10$ algorithm $A$ outputs a hypothesis distribution $\tilde{\bS}$ such that $\dtv(\tilde{\bS},\bS) \leq \eps.$
Then there are infinitely many quadruples $(a_1,a_2,a_3,a_4)$ such that for sufficiently large $N$, $A$ must use $\Omega(\log \log a_4)$ samples even when run with $\eps$
set to a (suitably small) positive absolute constant.
\end{reptheorem}

\subsection{Proof of Theorem \ref{known-k-is-four-lower}} \label{sec:proof-of-lower}

Fix $a_1 = 0$ and $a_2 = 1$.  (It suffices to prove a lower bound for
this case.)  To reduce clutter in the notation, let $p = a_3$ and $q = a_4$.
Applying Theorem~\ref{thm:reduction},
it suffices to prove that $\Omega(\log \log q)$
examples are needed to learn  distributions of random variables of
the form $\bS = \bU + p \bV \mod q$, where $\bU$ and $\bV$ are unknown 
PBDs over $\Theta(N)$ variables.  We do this in the rest of this section.

Since an algorithm that achieves a small error with high probability can be used to achieve small error in expectation,
we may use Lemma~\ref{fano}, which provides lower bounds on the number of examples needed for small expected error, to
prove Theorem~\ref{known-k-is-four-lower}.
To apply Lemma~\ref{fano}, we must find a set of distributions $\bS_1,\dots,\bS_i,\dots,$ where $\bS_i = \bU_i + p \bV_i \mod q$, that
are separated enough that they must be distinguished by a successful
learning algorithm (this is captured by the variation distance lower bound of Lemma \ref{fano}), but close enough (as captured by the Kullback-Liebler divergence upper bound) that this is difficult. 
We sketched the ideas behind our construction of these distributions $\bS_1,\dots,\bS_T$, $T=\log^{\Theta(1)}q$, earlier in Section \ref{s:ltechniques}, so we now proceed to the actual construction and proof.

%

The first step is to choose $p$ and $q$.   The choice is inspired by the theory of rational approximation of irrational numbers. The core of the construction requires us to use an irrational number which is hard to approximate as a ratio of small integers but such that, expressed as a continued fraction, the convergents do not grow very rapidly. For concreteness, we will consider the inverse of the golden ratio $\phi$:  \[
\frac{1}{\phi} = 
\cfrac{1}{1+\cfrac{1}{1+\cfrac{1}{1+\cdots}}}
\]

Let $f_0=1,f_1=1,f_2=2,\dots$ denote the Fibonacci numbers; It is easy to see that the $t^{th}$ convergent of $1/\phi$ is given by $f_{t-1}/f_t$. 
 We take $p = f_{L}$, $q = f_{L+1}$ where we think of $L$ as an asymptotically large parameter (and of course
$ q = f_{L+1}$ implies $L = \Theta(\log q)$).  Looking ahead, 
 the two properties of $1/\phi$ which will be useful will be: (a) For any $t$, $f_t/f_{t+1}$ is a very good approximation of $1/\phi$, and moreover, (b)\ignore{On the other hand, because of the property of continued fractions,\pnote{I can see that an upper bound on the quality of approximation is due to a property of continued fractions, but what is this true of a lower bound?}} the approximations obtained by these convergents are essentially the best possible. 
 
\begin{definition}
Let $\rho_q(a,b)$ be the Lee metric on $\mathbb{Z}_q$,
i.e., the minimum of $|j|$ over all $j$ such
that $a = b + j \mod q$.
\end{definition}

The following lemma records the properties of $p$ and $q$ that we will use. \ignore{\pnote{I don't mind deleting the informal
description if either of you think it's corny.  It does not take account of the fact that the lemma concerns walking in both
directions from zero, but you could get that from a walk in one direction by walking twice as long, and then shifting everything
back halfway, so I think that the informal description is morally the same as the lemma.}}To interpret this lemma,
it may be helpful to imagine starting at 0, taking steps of size $p$ through 
$[q]$, wrapping around when
you get to the end, and dropping a breadcrumb after each step.  Because $p$ and $q$ are relatively prime,
after $q$ steps, each member of $[q]$ has a breadcrumb.  Before this, however,
the lemma captures two ways in which the
breadcrumbs are ``distributed evenly'' (in fact, within constant factors of optimal)
throughout the walk: (a) they are pairwise well-separated,
and (b) all positions have a breadcrumb nearby.
\begin{lemma}
\label{l:space}
There are absolute constants $c_1,c_2>0$ such that,
for all integers $v \neq v', v,v' \in (-c_2q,c_2q)$, we have
\begin{equation}
\label{e:space.lower}
\rho_q( p v, p v') > \frac{c_1 q}{ \max \{ |v|, |v'| \}}.
\end{equation}
Furthermore, for any $i \in [q]$, for any $t \leq L$, there is a
$v \in \{ -f_t, ..., f_t \}$ such that 
\begin{equation}
\label{e:space.upper}
\rho_q(i, p v) \leq \frac{3 q}{f_t}.
\end{equation}
\end{lemma}

To prove the first part of Lemma~\ref{l:space}, we need the following lemma
on the difficulty of approximating $1/\phi$ by rationals.
\begin{lemma}[\cite{HarEtAl08}]
\label{l:phi.approx.lower}
There is a constant $c_3>0$ such that for all positive integers
$m$ and $n$,
\[
\left| \frac{m}{n} - \frac{1}{\phi} \right| \geq \frac{c_3}{n^2}.
\]
\end{lemma}

We will also use the fact that $\frac{f_{t-1}}{f_t}$ is a good approximation.
\begin{lemma}[\cite{HarEtAl08}]
\label{l:phi.approx.upper}
For all $t$,
\[
\left| \frac{f_{t-1}}{f_t} - \frac{1}{\phi} \right| < \frac{1}{f_t^2}.
\]
\end{lemma}

\medskip

\begin{proof}[Proof of Lemma~\ref{l:space}] 
Assume without loss of generality that $v' < v$.
By the definition of $\rho_q$, 
there is an integer $u$ such that
\[
|p v - p v' - u q| = \rho_q( p v, p v').
\]
Dividing both sides by $q (v - v')$, we get
\[
\left| \frac{p}{q} - \frac{u}{v - v'} \right| = \frac{\rho_q( p v, p v')}{q (v - v')}.
\]
Hence we have
\[
\frac{c_3}{(v-v')^2} \leq
\left| \frac{1}{\phi} - \frac{u}{v - v'} \right| < \frac{\rho_q( p v, p v')}{q (v - v')} + \frac{1}{q^2},
\]
where we have used Lemma~\ref{l:phi.approx.lower} for the first inequality and Lemma~\ref{l:phi.approx.upper} for the second.

If $|v| \leq \sqrt{c_3/4} \cdot q$ and $|v'| \leq \sqrt{c_3/4} \cdot q$,
then we get
\[
\frac{\rho_q( p v, p v')}{q (v - v')} > \frac{c_3}{2 (v-v')^2},
\]
so that
\[
\rho_q( p v, p v') > \frac{c_3 q}{2 (v - v') }
    \geq \frac{c_3 q}{ 4 \max \{ |v|, |v'| \}},
\]
completing the proof of (\ref{e:space.lower}).

Now we turn to  (\ref{e:space.upper}).  Two applications of Lemma~\ref{l:phi.approx.upper} and the triangle inequality together imply
\[
\left| \frac{f_{t-1}}{f_t} - \frac{p}{q} \right|  = 
\left| \frac{f_{t-1}}{f_t} - \frac{f_L}{f_{L+1}} \right|  \leq
\left| \frac{f_{t-1}}{f_t} - \frac{1}{\phi} \right|  +
\left| \frac{f_L}{f_{L+1}} - {\frac 1 \phi} \right|  \leq
{\frac 1 {f_t^2}} + {\frac 1 {f_{L+1}^2}} \leq \frac{2}{f_t^2}.
\]
Thus, for all integers $j$ with $|j| \leq f_t$, we have
\begin{equation}
\label{e:s.by.ft}
\left| \frac{j f_{t-1}}{f_t} - \frac{j p}{q} \right| \leq \frac{2}{f_t}.
\end{equation}
Since $f_{t-1}$ and $f_t$ are relatively prime, the elements
$\left\{ \frac{ j f_{t-1} }{f_t} \mod 1 : j \in [f_t] \right\}$ are exactly
equally spaced over $[0,1]$, so there is some $\tilde{j} \in [q]$ such that
the fractional part of $\frac{\tilde{j} f_{t-1} }{f_t}$ is within
$\pm \frac{1}{f_t}$ of $i/q$.  Applying (\ref{e:s.by.ft}), the fractional
part of $\frac{\tilde{j} p}{q}$ is within $\pm \frac{3}{f_t}$ of
$i/q$, and scaling up by $q$ yields (\ref{e:space.upper}), completing the proof of Lemma~\ref{l:space}.
\end{proof}

\medskip

Let $\ell = \lfloor \sqrt{L} \rfloor$.
Now that we have $p$ and $q$, we turn to defining
$(\bU_1, \bV_1),..., (\bU_{\ell}, \bV_{\ell})$.  The distributions that
are hard to distinguish will be chosen from among
\[
\bS_1 = \bU_1 + p \bV_1 \mod q, \ldots, \bS_{\ell} = \bU_{\ell} + p \bV_{\ell} \mod q.
\]
For a  positive integer $a$ 
let $\Bin(a^2,1/2)$ be the Binomial distribution which is a sum of $a^2$ independent Bernoulli random variables with expectation $1/2$, 
and let $\bW(a) = \Bin(2 a^2,1/2) - a^2$.
Let $C = \lceil a^2/q \rceil$ so that $Cq-a^2 \geq 0$, and observe that 
\[
\bW(a) + C q \mod q
\quad \text{
is identical to 
} \quad
\bW(a) \mod q,
\]
and that $\bW(a) + C q$ is a PBD$_{\Theta(q + a^2)}$ distribution.



We will need a lemma about how the $\bW$ random variables ``behave like discretized Gaussians''
that is a bit stronger in some cases than the usual Chernoff-Hoeffding bounds.
We will use the following bound on binomial coefficients:
\begin{lemma}[\cite{Ros99}]
\label{l:asymp.bin.coeff}
If $k = o(n^{3/4})$, then
\[
\frac{1}{2^n} \cdot { n \choose {\lfloor n/2 \rfloor + k} } = O \left( \frac{1}{\sqrt{n}} \exp\left(\frac{-2 k^2}{n} \right) \right).
\]
\end{lemma}

Now for our bound regarding $\bW$.
\begin{lemma}
\label{l:gaussian}
There is a constant $c_4>0$ such that for all $k$,
\[
\Pr[\bW(a) = k] \leq \frac{c_4 \exp\left(- \frac{k^2}{2 a^2} \right) }{a}.
\]
\end{lemma}
\begin{proof} If $|k| \leq a^{4/3}$, the lemma follows directly from Lemma~\ref{l:asymp.bin.coeff}.  If
$|k| > a^{4/3}$, we may suppose w.l.o.g. that $k$ is positive.  Then Hoeffing's inequality implies that 
\begin{align*}
& \Pr[\bW(a) = k] \leq \Pr[\bW(a) \geq k] \leq \exp\left(- \frac{k^2}{a^2}\right) \\
& \leq\exp\left(- \frac{k^2}{2 a^2}\right) \exp\left(- \frac{k^2}{2 a^2}\right) 
          \leq \exp\left(- \frac{k^2}{2 a^2}\right) \exp\left(- \frac{a^{2/3}}{2} \right) 
             = O\left( \exp\left(- \frac{k^2}{2 a^2}\right)/a \right),
\end{align*}
completing the proof.
\end{proof}

The following lemma may be considered a standard fact about binomial coefficients, but for completeness we give a proof below.

\begin{lemma}
\label{u:gaussian}
For every $c>0$, there exists $c'>0$ such that for all integers $x,a$ with $|x|< c \cdot a$, 
\[
\Pr[\bW(a) = x]  \ge \frac{c'}{a}. 
\]
\end{lemma}
\begin{proof}
Note that 
\[
\Pr[\bW(a) = x]  = \frac{1}{2^{2a^2}} \cdot \binom{2a^2}{a^2 + x}
\]
Thus, 
\begin{align*}
& \Pr[\bW(a) =x] \geq \frac{1}{2^{2a^2}} \cdot \binom{2a^2}{a^2} \cdot \frac{(a^2)! (a^2)!}{(a^2+x)! (a^2-x)!} \\
& \geq \frac{1}{10 \cdot a} \cdot \prod_{j=1}^{x} \frac{a^2 - j+1}{a^2 +j } = \frac{1}{10 \cdot a} \cdot \prod_{j=1}^{x} \frac{1 -\frac{ j-1}{a^2}}{1+\frac{j}{a^2} }\\
&\ge  \frac{1}{10 \cdot a} \cdot \prod_{j=1}^{x} e^{\frac{-10 \cdot j}{a^2}} \ge \frac{1}{10 \cdot a} \cdot e^{-\frac{10 \cdot x^2}{a^2}}. 
\end{align*}
Here the second inequality uses that $2^{-2n} \cdot \binom{2n}{n} \ge \frac{1}{10 \cdot \sqrt{n}}$ and the third inequality uses that for $j \le a^2/2$, 
$
e^{\frac{-10 \cdot j}{a^2}} \le (1 -\frac{ j-1}{a^2})/(1+\frac{j}{a^2} ).
$ The bound on $\Pr[\bW(a) = x]$   follows immediately. 
\end{proof}

Now, let $\bU_t = \bW(\lfloor \frac{p}{c_5 f_t} \rfloor)$, 
where $c_5$ is a constant that we will set in the analysis, and let
$\bV_t = \bW (f_t)$.  Recall that $\bS_t = \bU_t + p \bV_t \mod q$.
Let $\ell' = \lfloor L^{1/4} \rfloor$ and recall that $\ell = \lfloor L^{1/2} \rfloor$.
Let ${\cal S} = \{ \bS_{\ell'}, ... , \bS_{\ell} \}$
This is the set of $\Omega(\log^{1/2} q)$ distributions to which we will apply Lemma \ref{fano}.

Now, to apply Lemma~\ref{fano} we need upper bounds on the
KL-divergence between pairs of elements of $\cal S$, and lower bounds
on their total variation distance.  Intuitively, the upper bound on the KL-divergence will follow
from the fact that each $\bS_t$ ``spaces apart by $\Theta(p/f_t)$ PBDs
with a lot of measure in a region of size $p/f_t$'' (i.e. the translated $\bU_t$ distributions), so the probability never gets very small
between consecutive ``peaks'' in the distribution; consequently, all of
the probabilities in all of the distributions are within a
constant factor of one another.  The following lemma makes this precise:

\begin{lemma}\label{lem:KLbounds}
There is a constant $c_6>0$ such that, for large enough $q$,
if $\lfloor L^{1/4} \rfloor < t < L^{1/2}$, 
for all $i \in [q]$,
\[
\frac{1}{c_6 q} \leq \Pr[\bS_t = i] \leq \frac{c_6}{q}.
\]
\end{lemma}
\begin{proof}
Fix an arbitrary $t$ for which $\lfloor L^{1/4} \rfloor <
t < L^{1/2}$.  Since $t$ is fixed, we drop it from all subscripts.

First, let us work on the lower bound.  Roughly, we will show that a random
$v \sim \bV$ has a good chance of translating $\bU$ within $\Theta(\sigma(\bU))$ of $i$.  Specifically, (\ref{e:space.upper}) implies that there is a
$u \in \left[- 3 \lfloor q/f_t \rfloor,3 \lfloor q/f_t \rfloor \right]$
and a $v \in \left[-f_t,f_t\right]$ such that
$i = u + p v \mod q$. Thus
\[
\Pr[\bS = i] \geq \Pr[\bU = u] \cdot \Pr[\bV = v] 
\geq \Omega\left( \frac{f_t}{p} \cdot \frac{1}{f_t } \right)
  \geq \Omega\left( \frac{1}{p} \right) = \Omega\left( \frac{1}{q} \right),
\]
where the second inequality follows from an application of Lemma~\ref{u:gaussian} 
(recalling that $q=\Theta(p)$).

Now for the upper bound.  We have
\begin{align*}
\Pr[\bS = i] & = \sum_v \Pr[\bS = i | \bV = v] \Pr[\bV = v] \\
             & = \sum_v \Pr[\bU = i - p v \mod q] \Pr[\bV = v] \\
             & < o(1/q) + \sum_{v: |v| \leq \sigma(V) \ln q} \Pr[\bU = i - p v \mod q] \Pr[\bV = v],
\end{align*}
since $\Pr[|v| > \sigma(V) \ln q] = o(1/q)$.  Let ${\cal V}_1 = [-\sigma(\bV),\sigma(\bV)]$, and, for each $j > 1$, let
${\cal V}_{j} = [- j \sigma(\bV),j \sigma(\bV)] - {\cal V}_{j-1}$.
Then
\begin{align}
\nonumber
\Pr[\bS = i] & \leq o(1/q) + \sum_{j=1}^{\lfloor \ln q \rfloor} \sum_{v \in {\cal V}_j} \Pr[\bU = i - p v \mod q] \Pr[\bV = v] \\
\label{e:u.remains}
           & \leq o(1/q) + O(1) \cdot \sum_{j=1}^{\lfloor \ln q \rfloor} \frac{ e^{-j^2/2}}{\sigma(\bV)}
                 \sum_{v \in {\cal V}_j} \Pr[\bU = i - p v \mod q],
\end{align}
by Lemma~\ref{l:gaussian}.

Now fix a $j \leq \ln q$. 
Let $(v'_k)_{k=1,2,\dots}$ be the ordering of the elements of ${\cal V}_j$ by order of increasing $\rho_q$-distance from $pv'_k$ to $i$. Since each $|v'_k| \leq j \cdot \sigma(\bV) \ll c_2 q$,
Lemma~\ref{l:space} implies that $\rho_q$-balls of
radius $\Omega \left( \frac{q}{j \sigma(\bV)} \right)$ centered at members
of $p v_1',...,p v_k'$ are disjoint, so
\[
k \cdot \Omega \left( \frac{q}{j \sigma(\bV)} \right)
  < 2 \rho_q(pv'_k,i) + 1.
\]
Since $\sigma(\bU) \sigma(\bV) = \Theta(q)$, we get
\begin{equation}
\label{e:vprime}
\rho_q(pv'_k,i) = \Omega \left( \frac{k \sigma(\bU)}{j} \right).
\end{equation}
Applying Lemma~\ref{l:gaussian}, we get 
$$
\sum_{v \in {\cal V}_j} \Pr[\bU = i - p v \mod q] \le 
\frac{1}{\sigma(\bU)} \sum_{k>0}
\exp\bigg(-\Omega\bigg(\frac{k^2\cdot \sigma^2(\bU)  }{j^2 \sigma^2 (\bU)}\bigg)\bigg)
= O\bigg( \frac{j}{\sigma(\bU)} \bigg).
$$

Combining with (\ref{e:u.remains}), we get
\[
\Pr[\bS = i]  = \sum_{j=1}^{\infty}O\bigg( \frac{j}{\sigma(\bU)} \bigg) \cdot O\bigg( \frac{1}{\sigma(\bV)} \bigg) \cdot e^{-j^2/2} = O\bigg( \frac{1}{\sigma(\bU) \cdot \sigma(\bV)}\bigg) =O\left(\frac{1}{q}\right). 
\]
This finishes the upper bound on $\Pr[\bS = i]$, concluding our proof. 
\end{proof}

We have  the following immediate corollary.
\begin{lemma}
\label{l:kl.lower}
There is a constant $c_7$ such that, for all 
$i,j \in \{ \ell',..., \ell \}$, we have
$D_{KL} (\bS_i || \bS_j) \leq c_7$.
\end{lemma}

It remains only to give a lower bound on the total variation distance.
\begin{lemma}
\label{l:tv.lower}
There is a positive constant $c_8$ such that, for large enough $q$, for $\ell \geq i > j$,
we have $d_{TV} (\bS_i, \bS_j) > c_8$.
\end{lemma}
\begin{proof}  Let ${\cal W}$ be the union of two integer intervals
\[
{\cal W} = [ -f_i ,..., -f_{j+1} ] \cup 
           [ f_{j+1} ,..., f_i ].
\]
It may be helpful to think of $\cal W$ as
being comprised of $v$ such that $p v$ is the location
of a ``peak'' in $\bS_i$, but not in $\bS_j$.  
We will show that $\bS_i$ assigns significantly more probability to points close to $pv$ than $\bS_j$ does.

Choose $v \in {\cal W}$, and $u$ with $|u| \leq \frac{p}{c_5 f_i}$.   (We will later set $c_5>0$ to be a sufficiently large absolute constant.) For large enough $q$, standard facts about binomial coefficients give that
\begin{equation}
\label{e:lower.point.dtv}
\Pr[\bS_i = p v + u \mod q] \geq 
\Pr[\bV_i = v] \cdot \Pr[\bU_i = u]  \geq \frac{1}{5 f_i} \cdot \frac{c_5 f_i}{ 5 p}
                   = \frac{c_5}{25 p}.
\end{equation}
Now, let us upper bound $\Pr[\bS_j = p v + u \mod q]$.  
Let $a \in [q]$ be such that $p v + u = a \mod q$.  We have
\begin{align*}
\Pr[\bS_j = a]
 &   = \sum_v \Pr[\bS_j = a \  | \  \bV_{j} = v] \Pr[\bV_{j} = v] \\
 &  < o(1/q^2) + \sum_{v : |v| \leq \sigma(\bV_j) \ln q} \Pr[\bS_j = a \  | \  \bV_{j} = v] \Pr[\bV_{j} = v].
\end{align*}
As before, let 
${\cal V}_1 = [-\sigma(\bV_{j}),\sigma(\bV_{j})]$, and, for each $h > 1$, let
${\cal V}_{h} = [- h \sigma(\bV_{j}),h \sigma(\bV_{j})] - {\cal V}_{h-1}$, so that
Lemma~\ref{l:gaussian} implies
\begin{align}
\Pr[\bS_j = a ] & \leq o(1/q^2) + \sum_{h=1}^{\lfloor \ln q \rfloor} \sum_{v \in {\cal V}_{h}} \Pr[\bU_{j} = a-p v \mod q] \Pr[\bV_{j} = v] \\
\label{e:u.remains.dtv}
           & \leq o(1/q^2) + \sum_{h=1}^{\lfloor \ln q \rfloor} \frac{ c_4 e^{-(h-1)^2/2}}{\sigma(\bV_{j})}
                 \sum_{v \in {\cal V}_{h}} \Pr[\bU_{j} = a-p v \mod q].
\end{align}
Let $(v'_k)_{k=1,2,\dots}$ be the ordering of the elements of 
${\cal V}_{h}$ by order of increasing $\rho_q$ distance from $a$ to $pv'_k.$
Since each $|v'_k| \leq h \cdot \sigma(\bV) \ll c_2 q$,
Lemma~\ref{l:space} implies that $\rho_q$-balls of
radius $\frac{c_1 q}{ 2 h f_j }$ centered at members
of $p v_1',...,p v_k'$ are disjoint, so
\[
k \cdot \frac{c_1 q}{ h f_j }  < 2 \rho_q(a, pv'_k) + \frac{c_1 q}{ h f_j }.
\]
so, for large enough $q$, we have
\[
\rho_q(a, pv'_k) > \frac{c_1 (k-1) q}{ 5 h f_j }.
\]
Using this with Lemma~\ref{l:gaussian}, we get that, for large enough $q$,
\begin{align*}
\sum_{v \in {\cal V}_{h}} \Pr[\bU_{j} = a-p v \mod q]
& \leq \ignore{o(1/q^2) +}
{\frac 1 {\sigma(\bU_j)}}
   \sum_{k} 
      c_4 \exp\left(- \frac{(k-1)^2 c_1^2 q^2 c_5^2}{100 h^2 p^2} \right) \\
& \leq \ignore{o(1/q^2) +} 2 \cdot
      \frac{c_4 c_5 f_j }{p}
    \sum_{k} 
      \exp\left(- \frac{(k-1)^2 c_1^2 c_5^2}{100 h^2} \right)  \\
& \leq \ignore{o(1/q^2) +}
      \frac{c_4 c_5 f_j }{p}
      \cdot \frac{40 (h+1)}{c_1 c_5} =
      \frac{40 (h+1) c_4 f_j }{c_1 p}.
     \end{align*}
Plugging back into (\ref{e:u.remains.dtv}), 
we get
\[
\Pr[\bS_j = a ]
\leq 
o(1/q^2) + 40 \sum_{h=1}^{\lfloor \ln q \rfloor} \frac{(h+1)e^{-(h-1)^2/2}  c_4^2}{c_1 p}.
\]
Thus, if $c_5$ is a large enough absolute
constant, there is a constant $c_7$ such that
\[
\Pr[\bS_i = p v + u \mod q] 
 - \Pr[\bS_j = p v + u \mod q] > \frac{c_7}{p},
\]
for all 
$v \in {\cal W}$, and $u$ with $|u| \leq \frac{p}{c_5 f_i}$.
Lemma~\ref{l:space} implies that, if $c_5$ is large enough,
the resulting values of $pv+u$ are distinct elements of $[q]$, and the number of
such pairs is at least 
$(f_{i+1} - f_{i} ) \cdot \lfloor \Omega( \frac{p}{f_i} )\rfloor$
which is $\Omega(p)$, which completes the proof.
\end{proof}

\section{Lower bound for $(\amax,3)$-sums:  Proof of Theorem \ref{unknown-k-is-three-lower}} \label{sec:unknown-lower}

Theorem \ref{unknown-k-is-three-lower} follows from the following stronger result, which gives a lower bound for learning 
a weighted sum of $PBDs$ with weights $\{0=a_1,a_2,a_3\}$
even if  the largest support value $a_3$ is known.

\begin{theorem} [$k = 3$, strengthened unknown-support lower bound] \label{thm:stronger-unknown-lower} 
Let $A$ be any algorithm with the following properties:  algorithm $A$ is given $N$, an accuracy parameter $\eps > 0$, a value $0 < \amax \in \Z$, and access to i.i.d. draws from an unknown 
$\bS=a_2 \bS_2 + a_3 \bS_3$,
where $a_3$ is
the largest prime that is at most $\amax\}$ and $0 < a_2 < a_3$. (So the values $a_1= 0$ and $a_3$ are ``known'' to the learning algorithm $A$, but the value of $a_2$ is not.)  
Suppose that $A$ is guaranteed to output, with probability at least $9/10$, a hypothesis distribution $\tilde{\bS}$ satisfying $\dtv(\bS,\tilde{\bS}) \leq \eps.$
Then  for sufficiently large $N$, $A$ must use $\Omega(\log \amax)$ samples even when run with $\eps$ set to a (suitably small) positive absolute constant.
\end{theorem}

Via our reduction, Theorem \ref{thm:reduction}, we obtain Theorem \ref{thm:stronger-unknown-lower} from the following lower bound for learning a single scaled PBD mod $a_3$ when the scaling factor is unknown:

\begin{theorem} [lower bound for learning 
                 mod $a_3$]
\label{thm:unknown-mod3-lower}
Let $A$ be any algorithm with the following properties:  algorithm $A$ is given $N$, an accuracy parameter $\eps > 0$, a value $0 < \amax \in \Z$, and access to i.i.d. draws from 
from a distribution  $\bS'=(a_2 \bS_2 \mod a_3)$ where $\bS_2$ is a PBD$_N$, $a_3$ is the largest prime that is at most $\amax$,  and $a_2 \in \{1,\dots,a_3-1\}$ is ``unknown'' to $A$.
Suppose that $A$ is guaranteed to output a hypothesis distribution $\tilde{\bS}'$ satisfying $\E[\dtv(\bS',\tilde{\bS}')] \leq \eps$ (where the expectation is over the random samples drawn from $\bS'$ and any internal randomness of $A$).  
Then for sufficiently large $N$, $A$ must use $\Omega(\log \amax)$ samples  when run with $\eps$ set to some sufficiently small absolute constant.
\end{theorem}

While Theorem \ref{thm:unknown-mod3-lower} lower bounds the expected error of the hypothesis produced by a learning algorithm that uses too few samples, such a lower bound is easily seen to imply an $(\eps,\delta)$-type bound as well.  Thus to prove Theorem \ref{thm:stronger-unknown-lower} (and thus Theorem \ref{unknown-k-is-three-lower}) it suffices to prove Theorem \ref{thm:unknown-mod3-lower}.

\subsection{Proof of Theorem \ref{thm:unknown-mod3-lower}}

Recall that by the Bertrand-Chebychev theorem we have $a_3 = \Theta(\amax)$; throughout what follows we view \ignore{$N$ as ``sufficiently large'' relative to $a_2$ and for convenience we suppose that $N+1$ is a multiple of $a_2$; we view }$a_3$ as a ``sufficiently large'' prime number\ignore{ which is ``sufficiently large'' relative to $1/\eps$, and we view $\eps$ as a ``sufficiently small'' positive absolute constant}.
Let $\bS_2$ be the distribution $\Bin(N',{\frac 1 2}) + a_3 - \left({\frac {N'} 2} - {\frac {c \sqrt{N'}} 2}\right)$, where\ignore{\pnote{Removed ``2'', hoping to simplify some expressions downstream.}}
$N' = \lceil \left( {\frac {a_3}{cK}}\right)^2\rceil$
and $c,K>0$ are absolute constants that will be specified later.  (It is helpful to think of $c$ as being a modest number like, say, 10,  and to think of $K$ as being extremely large relative to $c$.)     Note that $\bS_2$ is a PBD$_N$ for $N=\poly(a_3)$, and that $\bS_2$ has $\Var[\bS_2] = N'/4 = \sigma_{\bS_2}^2$ where $\sigma_{\bS_2} = (a_3)/(cK) + O(1)$.  
Note further that nothing is ``unknown'' about $\bS_2$ --- the only thing about $\bS'=a_2 \bS_2$ that is unknown to the learner is the value of $a_2$.

\begin{remark} \label{rem:intuition}
  For intuition, it is useful to consider the distribution $\bS_2$ mod $a_3$, and to view it as a coarse approximation of the distribution $\bU$ which is uniform over the interval
$\{1,\dots,c\sqrt{N'} \}$ 
  where $c \sqrt{N'} \approx a_3/K$; we will make this precise later.
\end{remark}

The lower bound of Theorem \ref{thm:unknown-mod3-lower} is proved by considering distributions $\bS'_r$, $1 \leq r \leq a_3-1$, which are defined as $\bS'_r := (r \cdot \bS_2 \mod a_3)$.  The theorem is proved by applying the generalized Fano's Inequality (Theorem \ref{fano}) to a subset of the distributions $\{\bS'_r\}_{r \in [a_3-1]}$; recall that this requires both an upper bound on the KL divergence and a lower bound on the total variation distance.  The following technical lemma will be useful for the KL divergence upper bound:

\begin{lemma}  \label{lem:for-KL}
For any $1 \leq r_1 \neq r_2 \leq a_3-1$ and any $j \in \{0,1,\dots,a_3-1\}$, the ratio
$\bS'_{r_1}(j)/\bS'_{r_2}(j)$ lies in $[1/C,C]$ where $C$ is 
a constant (that is independent of $a_3$ but depends on $c,K$).
\end{lemma}

\begin{proof}
Recalling that $a_3$ is prime, for any 
 $r \in [a_3-1]$ and any $j \in \{0,1,\dots,a_3-1\}$,
if $r^{-1} \in [a_3]$ is such that $r^{-1} r = 1 \mod a_3$,
we have 
\[
\bS'_{r}(j)=\Pr[r \cdot \bS_2 \equiv j \mod a_3]
= \Pr[\bS_2 \equiv j r^{-1} \mod a_3] = \Theta(1) \cdot \bS_2(M),
\]
where $M \in \{0,1,\dots,N\}$ is the integer in 
$a_3 \Z + j r^{-1}$ 
that is closest to $N/2.$  Since
$|M-N/2| \leq a_3/2 =\Theta(\sqrt{N})$ and $\bS_2 = \Bin(N,1/2)$, standard facts about binomial coefficients imply that  $\bS_2(M) = {N \choose M}/2^N = \Theta(1)/\sqrt{N}$, from which the lemma follows. \end{proof}

From this, recalling the definition of KL-divergence $D_{KL}(\bS'_{r_1} || \bS'_{r_2}) = 
\sum_i \bS'_{r_1}(i) \ln {\frac {\bS'_{r_1}(i)} {\bS'_{r_2}(i)}}$, we easily obtain
\begin{corollary} \label{cor:KL-good}
For any $1 \leq r_1 \neq r_2 \leq a_3-1$ we have $D_{KL}(\bS'_{r_1} || \bS'_{r_2}) = O(1).$
\end{corollary}

Next we turn to a lower bound on the variation distance; for this we will have to consider only a restricted subset of the distributions $\{\bS'_r\}_{r \in [a_3-1]}$, and use a number theoretic equidistribution result of Shparlinski. To apply this result it will be most convenient for us to work with the distribution $\bU$ instead of $\bS_2$ (recall Remark \ref{rem:intuition}) and to bring $\bS_2$ and the $\bS'_r$ distributions into the picture later (in Section \ref{sec:concluding}) once we have established an analogue of our desired result for some distributions related to $\bU$.

\subsubsection{Equidistribution of scaled modular uniform distributions $\bU'_r.$}

For $1 \leq r \leq a_3-1$ we consider the distributions $\bU'_r := (r \cdot \bU \mod a_3)$ (note the similarity with the distributions $\bS'_r$).  Since for each $r \in [a_3-1]$ the distribution $\bU'_r$ is uniform on a $\Theta(1/K)$-fraction of the domain $\{0,1,\dots,a_3-1\}$, it is natural to expect that $d_{TV}(\bU'_{r_1},\bU'_{r_2})$ is large for $r_1 \neq r_2$.  To make this intuition precise, we make the following definition.

\begin{definition} \label{def:Npr}
Given integers $r, p, Y, Z$ and a set $\mathcal{X}$ of integers,\ignore{ \subset \{0,1,\dots,p-1\}} we define
\[
N_{r,p} (\mathcal{X}, Y , Z) := \bigg|\bigg\{ (x,y) \in \mathcal{X} \times  [Z +1, Z +Y] : r \cdot x \equiv y \ (\text{mod ~} p) \bigg\} \bigg|.
\]
\end{definition}
We will use the following, which is due to Shparlinski \cite{Shparlinski:08}.
\begin{lemma}[\cite{Shparlinski:08}]
For all integers $p, Z, X,Y$ 
such that $p \geq 2$ and $0 < X,Y < p$,
for any $\mathcal{X} \subseteq \{ 1,..., X \}$, we have
\[
\sum_{r=1}^p \left| N_{r,p} (\mathcal{X}, Y , Z) - \frac{|\mathcal{X}|  \cdot Y}{p}\right|^2 \le |\mathcal{X}| \cdot (X+Y) \cdot p^{o(1)}.
\]
\end{lemma}

We shall use the following corollary.
Set $\mathcal{X} = \{1, ..., X \}$.
Let us define the quantity 
\[
\mathcal{N}_{r,X} =| \{(x,y) : x,  y \in \mathcal{X}, r \cdot x \equiv y \ (\text{mod~}p)\}|.\]
Taking $Z=0$ and $Y=X$, we get:

\begin{corollary}
For all integers $p$ and $X$ such that
$p > 0$ and $0 < X < p$, we have
\[
\sum_{r=1}^p \left| \mathcal{N}_{r,X} -\frac{X^2}{p} \right|^2 \le X^2  \cdot p^{o(1)}.
\]
\end{corollary}
This easily yields
\[
\mathbf{E}_{r \in [p]}  \left[\left| \mathcal{N}_{r,X} -\frac{X^2}{p} \right|^2\right]\le \frac{X^2}{p^{1 - o(1)}}
  \]
which in turn implies
\ignore{\pnote{Because for $v > 0$, we have $\E_u((u-v)^2) \leq b \Rightarrow \Pr_v(u \geq 2 v) \leq b/v^2$.}}
\[
\mathbf{Pr}_{r \in [p]}  \left[\mathcal{N}_{r,X} \geq \frac{2 X^2}{p} \right]\le \frac{p^{1 + o(1)}}{X^2}.
\]

We specialize this bound by setting $X$ to be $\lceil c \sqrt{N'} \rceil$ and
$p = a_3$ which gives
${\frac {X^2}{p}} =
      {\frac {a_3}{K^2}} + O(1)$, from which we get that
\begin{equation}
  \Pr_{r \in [a_3]} \left[\mathcal{N}_{r,X}  \ge \frac{2 a_3}{K^2} \right] \le \frac{a_3^{o(1)}}{a_3}.
  \label{eq:a.k3l}
\end{equation}

Using (\ref{eq:a.k3l}) it is straightforward to show that there is a large subset of the distributions 
$\{\bU'_r\}_{r \in [a_3]}$ any two of which are very far from each other in total variation distance:

\begin{theorem} \label{thm:equidistribution-uniform}
Given any sufficiently large prime $a_3$, there is a subset of $t\geq a_3^{1/3}$ many values 
$\{q_1,\dots,q_t\} \subset [a_3]$ such that for any $q_i \neq q_j$ we have
$\dtv(\bU'_{q_i},\bU'_{q_j}) \geq 1 - {\frac{3}{K}}.$
\end{theorem}
\begin{proof}
  To prove the theorem it suffices to show that if $q_1, q_2$ are chosen randomly from $[a_3]$ then $\dtv(\bU'_{q_i},\bU'_{q_j}) \geq 1 - {\frac{3}{K}}$ with probability $1- \frac{a_3^{o(1)}}{a_3}$.
  (From there, the theorem follows from a standard deletion argument \cite{ASE:92}.)
  Since $a_3$ is prime, to show this it suffices to prove that for a randomly chosen $r \in [a_3]$ we have that $\dtv(\bU, \bU'_r) \geq 1 - {\frac{3}{K}}$
with probability $1 - \frac{a_3^{o(1)}}{a_3}$.   
Observe that for a given outcome of $r$, since both $\bU$ and $\bU'_r$ are uniform distributions over their domains ($\calX$ and $(r \cdot \calX \mod a_3)$ respectively) which are both of size $X=\lceil c\sqrt{N'} \rceil$, we have that $\dtv(\bU, \bU'_r) = 1 - {\frac {|{\cal X} \cap (r \cdot {\cal X} \mod a_3)|} {X}}.$  Moreover, we have that

\[
\mathcal{N}_{r,X} =| \{(x,(r x \mod \ a_3)) : x, (r x \mod \ a_3)   \in \mathcal{X}\}| = |{\cal X} \cap(r \cdot {\cal X} \mod a_3)|,
\]
so 
\[
\dtv(\bU, \bU'_r) = 1 - {\frac {{\cal N}_{r,X}}{X}} = 1 - {\frac {{\cal N}_{r,X}} {\lceil c \sqrt{N'} \rceil}}
=1 - {\frac {{\cal N}_{r,X}} {a_3/K + O(1)}}, 
\]
which is at least $1 - {\frac{3}{K}}$ provided that ${\cal N}_{r,X} < {\frac {2 a_3}{K^2}}.$ 
So by (\ref{eq:a.k3l}) we get that $\dtv(\bU, \bU'_r) \geq 1 - {\frac{3}{K}}$  with probability $1 - \frac{a_3^{o(1)}}{a_3}$  over a random $r$, as desired, and we are done.
\end{proof}

\subsubsection{Concluding the proof of Theorem \ref{thm:unknown-mod3-lower}.} \label{sec:concluding}

Given Theorem \ref{thm:equidistribution-uniform} it is not difficult to argue that the related family of distributions
$\{\bS'_{q_1},\dots,\bS'_{q_t}\}$ are all pairwise far from each other with respect to total variation distance.  First, recall that $\bS_{2}$ is a $\Bin(N',{\frac 1 2})$ distribution (mod $a_3$) which has been shifted so that its mode is in the center of $\supp(\bU)$ and so that the left and right endpoints of $\supp(\bU)$ each lie $c/2$ standard deviations away from its mode.  From this, the definition of $\bS'_{q_i}$, and well-known tail bounds on the Binomial distribution it is straightforward to verify that a $1 - e^{-\Theta(c^2)}$ fraction of the probability mass of $\bS'_{q_i}$ lies on $\supp(\bU'_{q_i})$, the support of $\bU'_{q_i}.$  Moreover, standard bounds on the Binomial distribution imply that for any two points $\alpha,\beta \in \supp(\bU'_{q_i}),$ we have that 
\begin{equation} \label{eq:Gamma}
{\frac 1 {\Gamma(c)}} \leq {\frac {\bS'_{q_i}(\alpha)}{\bS'_{q_j}(\beta)}} \leq \Gamma(c)
\end{equation}
where $\Gamma(c)$ is a function depending only on $c$.  Let $\bS''_{r}$ denote $(\bS'_r)_{\supp(\bU'_{r})}$, i.e. the conditional distribution of $\bS'_r$ restricted to the domain $\supp(\bU'_{r}).$  Recalling Theorem~\ref{thm:equidistribution-uniform} and the fact that $\dtv(\bU'_{q_i},\bU'_{q_j}) = 1 - {\frac {|\supp(\bU'_{q_i}) \cap \supp(\bU'_{q_j})|}{
|\supp(\bU'_{q_i})|}},$ by (\ref{eq:Gamma}) we see that by choosing $K$ to be suitably large relative to $\Gamma(c)$, we can ensure that
$\dtv(\bS''_{q_i},\bS''_{q_j})$ is at least  $9/10.$  Since $\dtv(\bS''_{r},\bS'_{r}) \leq e^{-\Theta(c^2)}$, taking $c$ to be a modest positive constant like 10, we get that
$\dtv(\bS'_{q_i},\bS'_{q_j})$ is at least  $8/10$ (with room to spare).  Thus we have established:

\begin{theorem} \label{thm:equidistribution-binomial}
Given any sufficiently large prime $a_3$, there is a subset of $t\geq a_3^{1/3}$ many values 
$\{q_1,\dots,q_t\} \subset [a_3]$ such that for any $q_i \neq q_j$ we have
$\dtv(\bS'_{q_i},\bS'_{q_j}) \geq 4/5.$
\end{theorem}

All the pieces are now in place to apply Fano's Inequality.  In the statement of Theorem \ref{fano} we may take $\alpha=1$ (by Theorem \ref{thm:equidistribution-binomial}), $\beta = O(1)$ (by Corollary \ref{cor:KL-good}), and $\eps$ to be an absolute constant, and Theorem \ref{fano} implies that any algorithm achieving expected error at most $\eps$ must use $\Omega(\ln t) = \Omega(\ln a_3)$ samples.  This concludes the proof of Theorem \ref{thm:unknown-mod3-lower}. \qed

\appendix

\section{Time complexity of evaluating and sampling from our hypotheses} 
\label{sec:evaluation}

Inspection of our learning algorithms reveals that any possible hypothesis distribution $\bH$ that the algorithms may construct, corresponding to any possible vector of outcomes for the guesses that the algorithms may make, must have one of the following two forms:

\begin{itemize}

\item[(a)] (see Sections~\ref{sec:allsmall} and~\ref{sec:a0-is-four}) $\bH$ is uniform over a multiset $S$ of at most $1/\eps^{2^{\poly(k)}}$ many integers (see the comment immediately after Fact~\ref{fact:learn-sparse-ess-support}; note that the algorithm has $S$).

\item [(b)] (see Lemma~\ref{lem:siirv-kernel} and Definition~\ref{def:kernel}) $\bH$ is of the form $\bU_{\widehat{Y}} + \bZ$ where
$\widehat{Y}$ is a multiset of at most $1/\eps^{2^{\poly(k)}}$ integers (note that the algorithm has $\widehat{Y}$) and $\bZ = \sum_{a=1}^{K=\poly(k)} p_a \bZ_a$ where $\bZ_a$ is the uniform distribution on the interval $[-c_a, c_a] \cap \mathbb{Z}$ (and the algorithm has $K$, the $p_a$'s, and the $c_a$'s).

\end{itemize}

It is easy to see that a draw from a type-(a) distribution can be simulated in time $1/\eps^{2^{\poly(k)}}$, and likewise it is easy to simulate an $\eval_{\bH}$ oracle for such a distribution in the same time.  It is also easy to see that a draw from a type-(b) distribution can be simulated in time $1/\eps^{2^{\poly(k)}}$.  The main result of this appendix is Theorem \ref{thm:evaluation}, stated below.  Given this theorem it is easy to see that a type-(b) $\eval_{\bH}$ oracle can be simulated in time $1/\eps^{2^{\poly(k)}}$, which is the final piece required to establish that our hypotheses can be efficiently sampled and evaluated as required by Corollary~\ref{cor:guess}.

\begin{theorem}\label{thm:evaluation}
Let $\bH$ be a distribution which is of the form $\bH = \bY + \sum_{a=1}^K p_a \cdot \bZ_a$ where the distributions $\bY, \bZ_1 , \ldots, \bZ_K$ are independent integer valued random variables. Further, 
\begin{enumerate}

\item Each $\bZ_a$ is uniform on the integer intervals $[-\gamma_a, \ldots, \gamma_a]$.

\item $\bY$ is supported on a set $A_{\bY} \subseteq \mathbb{Z}$ of size $m$ with the probabilities given by $\{\alpha_V\}_{V \in A_{\bY}}$. 

\end{enumerate}
Given as input the sets $A_{\bY}$, $\{p_a\}_{a=1}^K$,  $\{\alpha_V\}_{V \in A_{\bY}}$ and $\{\gamma_a\}_{a=1}^K$ and a point $x \in \mathbb{Z}$, we can evaluate $\Pr[\bH=x]$ in time $\mathcal{L}^{O(K)}$ where $\mathcal{L}$ is the length of the input.\ignore{ In other words, if $K=O(1)$, the algorithm runs in polynomial time. }
\end{theorem} 
Our chief technical tool to prove this will be the following remarkable theorem of Barvinok~\cite{Barvinok:94}, which shows that the number of integer points in a rational polytope can be computed in polynomial time when the dimension is fixed:

\begin{theorem}\label{thm:Barvinok}[Barvinok]
There is an algorithm with the following property: given as input a list of $m$ pairs $(a_1,b_1),\dots,(a_m,b_m)$ where each $a_i \in \Q^d, b_i \in \Q$, specifying a polytope $X \subseteq \mathbb{R}^d$, $X = \{x \in \mathbb{R}^d: \langle a_i, x \rangle \le b_i\}_{i=1}^m$, the algorithm outputs the number of integer points in $X$ in time $\mathcal{L}^{O(d)}$ where $\mathcal{L}$ is the length of the input, i.e. the description length of $\{a_i\}_{i=1}^m$ and $\{b_i\}_{i=1}^m$. 
\end{theorem}

We will use this algorithm via the following lemma. 

\begin{lemma}\label{lem:density-count1}
Let $\bH'$ be a distribution which is of the form $\bH' = V+  \sum_{a=1}^K p_a \cdot \bZ_a$ where $V,p_1,\dots,p_K \in \mathbb{Z}$ and $\bZ_1,\dots,\bZ_K$ are independent integer valued random variables and for $1 \le a \le K$, 
$\bZ_a$ is uniform on $[-\gamma_a, \ldots, \gamma_a]$. Then, given any point $x \in \mathbb{Z}$, $\{p_a\}_{a=1}^K$, $\{\gamma_a\}_{a=1}^K$ and $V$, the value $\Pr[\bH'=x]$ 
can be computed in time $\mathcal{L}^{O(K)}$ where $\mathcal{L}$ is the description size of the input. 
\end{lemma}

\begin{proof}
Consider the polytope defined by the following inequalities:  
\begin{align*}
\text{for~}1 \leq a \leq k, \ \ \  -\gamma_a  \le  y_a \le \gamma_a, \quad \quad \quad\text{and} \quad \quad
 x-V-0.1  \le \sum_{a=1}^K p_a \cdot y_a \le x -V+0.1.
\end{align*}
Then it is easy to see that if $\mathcal{N}_x$ is the number of integer points in the above polytope, then 
\[
\Pr[\bH'=x] =  \mathcal{N}_x \cdot \prod_{a=1}^K \frac{1}{2 \gamma_a +1}. 
\]
Combining this observation with Theorem~\ref{thm:Barvinok} proves the lemma. 
\end{proof}

\begin{proofof}{Theorem~\ref{thm:evaluation}}
Let $V \in A_{\bY}$. Then, using Lemma~\ref{lem:density-count1}, we obtain that for $\bH_V = V + \sum_{a=1}^K p_a \cdot \bZ_a$, $\Pr[\bH_V=x]$ can be computed in time $\mathcal{L}^{O(K)}$. Now, observe that 
$$
\Pr[\bH=x] = \sum_{V \in A_{\bY}} \Pr[\bH_V = x] \cdot \alpha_V. 
$$
As each term of the above sum can be computed in time $\mathcal{L}^{O(K)}$, hence the total time required to compute the above sum is bounded by $\mathcal{L}^{O(K)}$ (note that $\mathcal{L} \ge |A_\bY|$).
\end{proofof}

\subsection*{Acknowledgments} 
We would like to thank Igor Shparlinski and Aravindan Vijayaraghavan for useful discussions,
and the JMLR reviewers for many helpful comments.
\bibliographystyle{alpha}	
\bibliography{allrefs}	

\end{document}